\pdfoutput=1
\documentclass[
    affil-it, 
    auth-lg, 
    british, 
    contents, 
    compression, 
    custom-packages={caption,subcaption,setspace}, 
    references={references}, 
]{lib/preprint}

\newcommand{\bg}[1]{\mathring{#1}}
\newcommand{\vol}[1]{{\rm vol}_{#1}}

\begin{document}
    
    \date{\today}

    \email{lewis.napper@surrey.ac.uk,i.roulstone@surrey.ac.uk,volodya@univ-angers. fr,m.wolf@surrey.ac.uk}

    \preprint{DMUS--MP--23/03}
    \title{Monge--Amp{\`e}re Geometry and Vortices}

    \author[a]{Lewis~Napper\,\orcidlink{0009-0001-6961-7055}\,}
    \author[a]{Ian~Roulstone\,\orcidlink{0000-0003-2091-0952}\,}
    \author[b]{Vladimir~Rubtsov\,\orcidlink{0000-0002-2765-0957}\,}
    \author[a]{Martin~Wolf\,\orcidlink{0009-0002-8192-3124}\,}

    \affil[a]{School of Mathematics and Physics, University of Surrey,\\ Guildford GU2 7XH, United Kingdom}
    \affil[b]{University of Angers, CNRS, LAREMA, SFR MATHSTIC,\\ F-49000 Angers, France}

    \abstract{We introduce a new approach to Monge--Amp{\`e}re geometry based on techniques from higher symplectic geometry. Our work is motivated by the application of Monge--Amp{\`e}re geometry to the Poisson equation for the pressure that arises for incompressible Navier--Stokes flows. Whilst this equation constitutes an elliptic problem for the pressure, it can also be viewed as a non-linear partial differential equation connecting the pressure, the vorticity, and the rate-of-strain. As such, it is a key diagnostic relation in the quest to understand the formation of vortices in turbulent flows. We study this equation via an associated (higher) Lagrangian submanifold in the cotangent bundle to the configuration space of the fluid. Using our definition of a (higher) Monge--Amp{\`e}re structure, we study an associated metric on the cotangent bundle together with its pull-back to the (higher) Lagrangian submanifold. The signatures of these metrics are dictated by the relationship between vorticity and rate-of-strain, and their scalar curvatures can be interpreted in a physical context in terms of the accumulation of vorticity, strain, and their gradients. We show explicity, in the case of two-dimensional flows, how topological information can be derived from the Monge--Amp{\`e}re geometry of the Lagrangian submanifold. We also demonstrate how certain solutions to the three-dimensional incompressible Navier--Stokes equations, such as Hill's spherical vortex and an integrable case of Arnol'd--Beltrami--Childress flow, have symmetries that facilitate a formulation of these solutions from the perspective of (higher) symplectic reduction.}

    \acknowledgements{We gratefully acknowledge stimulating conversations with Jonathan Bevan, Thomas Bridges, Roberto D'Onofrio, Jan Gutowski, Giovanni Ortenzi, Christian Saemann, and Paul Skerritt. We are particularly grateful to Jan Gutowski and Paul Skerritt for comments on the first version of this paper. LN was supported by the STFC grant ST/V507118/1. VR thanks the IHES for hospitality during the final stage of this work as well as the Centre Henri Lebesgue, programme ANR-11-LABX-0020-0.}

    \datalicencemanagement{No additional research data beyond the data presented and cited in this work are needed to validate the research findings in this work. For the purpose of open access, the authors have applied a Creative Commons Attribution (CC BY) licence to any Author Accepted Manuscript version arising.} 

    \begin{body}

        \section{Introduction}\label{sec:intro}

        The anatomy and dynamics of vortices are subjects of fundamental importance in the study of the incompressible Euler and Navier--Stokes equations in two and three dimensions. According to the incompressible Navier--Stokes equations on a three-dimensional Euclidean domain, the evolution of the vorticity, ${\bm\zeta}\coloneqq{\bm\nabla}\times\bmv$, is given by
        \begin{equation}\label{eq:NSV}
            \frac{\rmD{\bm\zeta}}{\rmD t}\ =\ {\bm\sigma}+\nu\lap{\bm\zeta}
            \ewith
            \frac{\rmD}{\rmD t}\ \coloneqq\ \parder{t}+\bmv\cdot{\bm\nabla}~,
        \end{equation}
        where $\bmv$ is the fluid velocity, ${\bm\sigma}\coloneqq({\bm\zeta}\cdot{\bm\nabla})\bmv$ is the vortex-stetching vector, $\lap$ is the Laplacian, and $\nu$ is the viscosity. The equation for divergence-free flow is 
        \begin{equation}\label{eq:incomp}
            \nabla\cdot\bmv\ =\ 0~. 
        \end{equation}
        When the flow is inviscid, $\nu=0$, we obtain the Euler equations. In two dimensions, the vortex-stretching vector is identically zero, and the vorticity is a materially conserved scalar field when the flow is inviscid. 

        The vortex-stretching vector ${\bm\sigma}$ can be written in terms of the rate-of-strain matrix, $S$, which is the symmetric part of the velocity-gradient matrix, as follows
        \begin{equation}\label{eq:S1}
            {\bm\sigma}\ =\ S{\bm\zeta}~.
        \end{equation}
        The magnitude and direction of vorticity are critical features in studies of turbulent flows in three dimensions. If, for the moment, we focus on inviscid flows in three dimensions, then using~\eqref{eq:NSV} and the Euler equation relating flow velocity and the pressure gradient
        \begin{equation}\label{eq:E1}
            \frac{\rmD\bmv}{\rmD t}\ =\ -{\bm\nabla}p~,
        \end{equation}
        where $p$ is pressure, it can be shown (see e.g.~\cite{Gibbon:2008aa}) that the vortex-stretching vector evolves according to
        \begin{equation}
            \frac{\rmD\bm\sigma}{\rmD t}\ =\ -P{\bm\zeta}~,
        \end{equation}
        where $P$ is the Hessian of the pressure. Upon taking the divergence of~\eqref{eq:E1} and using~\eqref{eq:incomp}, the rate-of-strain, the vorticity, and the pressure field are related by 
        \begin{equation}\label{eq:PEP}
            \tr(P)\ =\ \lap p\ =\ \tfrac12{\bm\zeta}^2-\tr(S^2)~.
        \end{equation}
        This equation holds for the incompressible Euler and Navier--Stokes equations in both two and three dimensions. In the standard literature,~\eqref{eq:PEP} is recognised as the Poisson equation for the pressure, and depends on time only as a parameter. Consequently, the time evolution of $\lap p$ will depend on whether we are considering the Euler or the Navier--Stokes equations.

        Equation~\eqref{eq:PEP} has been employed in studies of the accumulation of vorticity and of the plausible conditions under which such accumulations may be considered to be `a vortex'. From~\eqref{eq:PEP} it follows that, when vorticity dominates over strain, then $\lap p>0$, and conversely, when strain dominates over vorticity, $\lap p<0$. In two dimensions, the vorticity is a scalar field, the rate-of-strain matrix has only two independent components (due to incompressibility) and therefore~\eqref{eq:PEP} is a useful diagnostic relation involving the velocity gradients and the pressure field, as studied in~\cite{Larcheveque:1990aa,Larcheveque:1993aa,Weiss:1991aa}. In three dimensions, although the components of vorticity and strain can interact in more complicated ways to determine the sign of the Laplacian of the pressure,~\eqref{eq:PEP} has still been widely used in studies that address the enduring question as to what a vortex is; see e.g.~\cite{Jeong:1995aa,Dubief:2000aa}.

        Equation~\eqref{eq:PEP}, or rather its reformulation on an arbitrary Riemannian manifold, is a focal point of this paper. In particular,~\cite{Larcheveque:1990aa,Larcheveque:1993aa,Weiss:1991aa} studied this equation in the context of incompressible flows on a two-dimensional Euclidean domain, whereupon the velocity can be expressed in terms of the derivatives of a stream function. Using such a representation of the velocity field, the right-hand side of~\eqref{eq:PEP} becomes proportional to the determinant of the Hessian of the stream function. Consequently, the Gau{\ss}ian curvature of the stream function is related to the sign of the Laplacian of pressure. When vorticity dominates over strain, the stream function, viewed as a graph in Euclidean space, has positive Gau{\ss}ian curvature, and it has negative Gau{\ss}ian curvature when strain dominates. By introducing the stream function,~\eqref{eq:PEP} can then be viewed as a non-linear Monge--Amp{\`e}re equation for this function, assuming $\lap p$ is known and time, $t$, is considered a parameter. When $\lap p>0$, this equation is elliptic; conversely, when $\lap p<0$, the equation is hyperbolic (we shall return to these points in greater detail later in this paper). 

        The appearance of a Monge--Amp{\`e}re equation for two-dimensional incompressible flows led~\cite{Roulstone:2009bb} to study this problem from the point of view of the Monge--Amp{\`e}re geometry of~\cite{Kushner:2007aa}. In this context, one considers a pair of differential two-forms, $(\omega,\alpha)$, on $T^*M$, where $M$ is the configuration space of the fluid, $\omega$ is the symplectic form, and $\alpha$ is called the Monge--Amp{\`e}re form, which encodes~\eqref{eq:PEP}. This pair of forms satisfies a non-degeneracy condition, and such a pair is called a Monge--Amp{\`e}re structure. With this geometric picture in mind, the conditions for ellipticity and hyperbolicity noted by~\cite{Larcheveque:1990aa,Larcheveque:1993aa,Weiss:1991aa} then translate, via the Monge--Amp{\`e}re structure, into almost (para-)complex structures on $T^*M$ which, in fact, extend to almost quaternionic (para-)Hermitian structures~\cite{Roulstone:2009bb,Banos:2015aa}. 
                
        When incompressible flows in three dimensions are considered in terms of~\eqref{eq:PEP}, the absence of a stream function prohibits a generalisation of the classification of flows in terms of an elliptic or hyperbolic Monge--Amp{\`e}re equation~\cite{Larcheveque:1990aa,Larcheveque:1993aa}. Nevertheless,~\cite{Roulstone:2009aa} showed how, on a three-dimensional Euclidean domain,~\eqref{eq:PEP} can still be described using a suitably-defined Monge--Amp{\`e}re structure. This construction facilitated a generalisation of the criteria derived by~\cite{Larcheveque:1990aa,Larcheveque:1993aa,Weiss:1991aa} to three-dimensional incompressible flows. In this present paper, we shall show how these earlier results can be reformulated and unified by combining the ideas of Monge--Amp{\`e}re geometry with that of the higher symplectic geometry of~\cite{Cantrijn:1999aa,Baez:2010,Rogers:2011}, thereby providing a generalisation of this approach to incompressible flows on arbitrary Riemannian manifolds of arbitrary dimensions.

        Concretely, when classifying Monge--Amp{\`e}re equations in two or three independent variables,~\cite{Lychagin:1993aa} introduced a certain Riemannian or Kleinian metric on $T^*M$, whose signature has been related to the elliptic or hyperbolic nature of the underlying Monge--Amp{\`e}re equation. In view of this, one makes use of generalised solutions to the Monge--Amp{\`e}re equation associated with a Monge--Amp{\`e}re structure $(\omega,\alpha)$, which are Lagrangian submanifolds of $T^*M$ on which $\alpha$ vanishes. The notion of a generalised solution to a partial differential equation was first introduced in~\cite{Vinogradov:1973aa,Vinogradov:1975aa,Vinogradov:1977aa} and corresponds to admitting solutions that are multivalued or are not globally defined. In this context, classical (global, single valued) solutions to a Monge--Amp{\`e}re equation are described precisely by the graphs of differentials of functions. We shall adopt this view of generalised solutions when extending the aforementioned ideas to higher-dimensional incompressible flows. In particular, beginning from a (higher) Monge--Amp{\`e}re structure, now a specific pair of differential $m$-forms on $T^*M$ in the $m$-dimensional case, we construct a metric on $T^*M$. Moreover, it will be useful to study its pull-back to a certain (higher) Lagrangian submanifold of $T^*M$, noting that the submanifold is of the same dimension as the configuration space of the fluid.

        A physical motivation for considering the pull-back metric is as follows. In~\cite{Gibbon:2008aa}, it was noted that~\eqref{eq:PEP} locally holds the key to the formation of vortical structures through the sign of $\lap p$. The equation also plays a role in Navier--Stokes turbulence calculations in which vorticity tends to accumulate on `thin sets' -- quasi-two-dimensional sheets that roll up into one-dimensional tubes~\cite{Douady:1991aa}. The topology of vortex tubes can become highly complicated, but they are ubiquitous features of turbulent flows and have been dubbed `the sinews of turbulence' in~\cite{Moffatt:1994aa}. Extracting topological information from the underlying partial differential equation of the Navier--Stokes equations is an enduring problem: for a review, see~\cite{Paralta:2015aa}. Using the pull-back metric described above, we demonstrate, explicitly in the case of two-dimensional incompressible flows, that a topological invariant can be associated with the Lagrangian submanifold $L$ defined by~\eqref{eq:PEP}. Since $L$ is two-dimensional in this case, when the pull-back metric is Riemannian we can use the Gau{\ss}--Bonnet theorem to calculate an Euler number. We find that the curvature of $L$ is related to the physical properties of the flow in terms of gradients of vorticity and strain.

        We go on to show that certain solutions to the three-dimensional Navier--Stokes equations, such as Burgers' vortex, Hill's spherical vortex, and an integrable Arnol'd--Beltrami--Childress flow, possess symmetries that facilitate Hamiltonian reductions to two-dimensional problems. These results extend those presented in~\cite{Banos:2015aa}; in particular we show how such solutions can be studied from the point of view of the higher symplectic reduction of~\cite{Blacker:2021aa}. As somewhat of an aside for this paper, we note that helicity is described readily using the component parts of the Monge--Amp{\`e}re geometry developed herein. Further investigations using helicity or other invariants (e.g.~Maslov index) to study the topology of three-dimensional flows, within the framework of our Monge--Amp{\`e}re geometry, is a topic for future research.

        \paragraph{Organisation of the paper.}
        In \cref{sec:IFF}, we present the Navier--Stokes equations in a covariant framework, where the configuration space is an arbitrary Riemannian manifold $M$. Whilst this introduces some additional structure, the majority of our results are couched in a covariant language and it is therefore consistent to allow for arbitrary background geometries. Furthermore, the Poisson equation for the pressure involves additional curvature terms when written for an arbitrary $M$.

        Focusing on two-dimensional incompressible flows,  in \cref{sec:geometry2dFlows} we introduce the machinery of Monge--Amp{\`e}re geometry and Monge--Amp{\`e}re structures, following the geometric approach as described, for example, by~\cite{Kushner:2007aa}. This allows us to formulate the Monge--Amp{\`e}re equation arising in the Poisson equation for the pressure, revisiting some of the results of~\cite{Roulstone:2009bb}. However, in addition, we describe the role of the metric structure that arises on $T^*M$, as well as its pull-back to the Lagrangian submanifold $L$. We then use the pull-back metric on $L$ to show how the Gau{\ss}--Bonnet theorem, together with conditions on the projection $L\rightarrow M$, enable us to define the Euler number for `a vortex'. 

        Examples are then given in \cref{sec:examples2D}, in order to illustrate the application of the foregoing theory. These examples include a flow with topological bifurcations and the Taylor--Green vortex in two dimensions. Naturally, it is possible to find solutions to~\eqref{eq:PEP} that are not solutions to the full dynamical equations. However, our focus for now is to take the view that the geometry developed herein gives us the possibility to characterise the physical features of given solutions, rather than to explore how the Monge--Amp{\`e}re structure might facilitate the search for new solutions.

        In \cref{sec:geometricProperties3d}, we move on to consider~\eqref{eq:PEP} for flows in three dimensions. Concretely, we first revisit, in \cref{sec:2dr}, the two-dimensional case and note that another Monge--Amp{\`e}re structure can be defined using a different choice of symplectic structure. This choice is characterised by a duality (with respect to the metric on the configuration manifold, $M$), which in two-dimensions simply provides an alternative formulation to the one used before. However, in three dimensions, the duality leads to a higher Monge--Amp{\`e}re type structure defined in terms of a pair of differential three-forms, and this naturally encodes an equation such as~\eqref{eq:PEP}, even though there is no longer an underlying Monge--Amp{\`e}re equation in a single dependent variable. In \cref{sec:kpm}, we introduce the relevant concepts from higher symplectic geometry, in which the symplectic form is superseded by a closed, non-degenerate differential form of degree higher than two. 

        In \cref{sec:mag3df}, we explicitly set out the Monge--Amp{\`e}re geometry of three-dimensional flow, thereby extending the results of~\cite{Banos:2015aa}. We explain how the curvature of $L$ can be related, once again, to gradients of vorticity and strain. However, because the higher Lagrangian submanifold is now three-dimensional, we can no longer use the Gau{\ss}--Bonnet theorem to quantify the topology of vortices. Instead, we remark that helicity, a much-studied invariant of incompressible flows in three dimensions, can be formulated in terms of the geometric objects we have introduced. In principle, the resulting formulation can be applied to the Navier--Stokes equations in arbitrary dimensions. 

        We then illustrate the foregoing theory with examples in \cref{sec:eg3d}. We begin by writing out the form of the pull-back metric in terms of vorticity and rate-of-strain, via the velocity gradient matrix, and then discuss an example based on Burgers' vortex. This canonical model of a vortex tube has been studied from the point-of-view of Monge--Amp{\`e}re geometry in~\cite{Banos:2015aa}, where symplectic reduction was employed to illuminate the symmetry of the model, which is characterised physically by uniformity along the axis of rotation. 

        Inspired by this approach to solutions to the three-dimensional Navier--Stokes equations with symmetry, we introduce, in \cref{sec:kprs}, a higher symplectic reduction of the phase space and apply this to an integrable example of the Arnol'd--Beltrami--Childress flows, and to Hill's spherical vortex (a special case of the Hicks--Moffatt vortex), in \cref{sec:examplesReduction}.

        Finally, in \cref{sec:conclusions}, we summarise and draw conclusions.
        
        \section{Fluid flows and Monge--Amp{\`e}re geometry}\label{sec:FFMAG}

        \subsection{Incompressible fluid flows}\label{sec:IFF}

        To set the stage, let us summarise some facts about incompressible fluid flows. We shall be somewhat more general in that we allow the domain of the fluid flow to be a Riemannian manifold. This is because the equation~\eqref{eq:PEP} for the Laplacian of the pressure will be modified by a term depending on the Ricci curvature tensor of the underlying domain.

        \paragraph{Navier--Stokes equations.}
        Consider an $m$-dimensional oriented Riemannian manifold\footnote{We shall always assume that our manifolds are equipped with a \uline{good cover}, by which we mean a covering by open and contractible sets. Likewise, we shall always work in the induced topologies.} $M$ with metric $\bg{g}$. Let `$\rmd$' be the \uline{exterior derivative} on $M$ and `$\star_{\bg{g}}$' be the \uline{Hodge star} with respect to the metric $\bg{g}$. Furthermore, let the \uline{codifferential} acting on differential $p$-forms $\Omega^p(M)$ and the \uline{Hodge Laplacian} be given by
        \begin{equation}\label{eq: HodgeLaplace}
        	\bg{\delta}\ \coloneqq\ (-1)^{m(p-1)+1}{\star_{\bg{g}}\rmd\star_{\bg{g}}}
            \eand
            \bg{\lap}_{\rm H}\ \coloneqq\ \bg{\delta}\rmd+\rmd\bg{\delta}~,
        \end{equation}
        respectively. Set
        \begin{equation}
        	|\rho|^2\ \coloneqq\ \frac{\rho\wedge{\star_{\bg{g}}\rho}}{\vol{M}}
        \end{equation}
        for all $\rho\in\Omega^p(M)$, with `$\vol{M}$' the volume form on $M$ induced by $\bg{g}$.

        In the following, $M$ is taken to be the domain of the fluid flow in which we are interested, and the fluid flow is described by a one-parameter family of differential one-forms $v\in\Omega^1(M)$ on $M$, parametrised by $t\in\IR$, known as the \uline{velocity (co-)vector field}. The \uline{incompressible Navier--Stokes equations} is a system consisting of a flow equation for $v$,
        \begin{subequations}\label{eq:generalNavierStokes}
            \begin{equation}\label{eq:generalNavierStokesA}
                \parder[v]{t}\ =\ -(-1)^m{\star_{\bg{g}}(v\wedge{\star_{\bg{g}}\rmd v})}-\tfrac12\rmd|v|^2-\rmd p-\nu\bg{\lap}_{\rm H}v~,
            \end{equation}
            together with the \uline{divergence-free constraint}
            \begin{equation}\label{eq:generalNavierStokesB}
                \bg{\delta}v\ =\ 0~.
            \end{equation}
        \end{subequations}
        Here, $p\in\scC^\infty(M)$ is known as the \uline{pressure field} and $\nu\in\IR$ as the \uline{viscosity}, respectively.

        If we coordinatise $M$ by $x^i$ for $i,j,\ldots=1,\ldots,m$, then, with $v=v_i\rmd x^i$ and $v_i=v_i(t,x^1,\ldots,x^m)$, the equations~\eqref{eq:generalNavierStokes} become
         \begin{subequations}\label{eq:generalNavierStokesLocal}
            \begin{equation}\label{eq:generalNavierStokesLocalA}
                \parder[v^i]{t}\ =\ -v^j\bg{\nabla}_jv^i-\partial^ip+\nu(\bg{\lap}_{\rm B}v^i-\bg{R}^{ij}v_j)
            \end{equation}
            and
            \begin{equation}\label{eq:generalNavierStokesIncompressibilityLocal}
                \bg{\nabla}_iv^i\ =\ 0~,
            \end{equation}
        \end{subequations}
        where $\bg{\nabla}_i$ is the \uline{Levi-Civita connection} for the metric $\bg{g}_{ij}$ with Christoffel symbols denoted by $\bg{\Gamma}_{ij}{}^k$. Furthermore, the components of the associated \uline{Ricci tensor} and \uline{Riemann curvature tensor} are given by 
        \begin{equation}\label{eq:curvatureTensors}
        	\bg{R}_{ij}\ \coloneqq\ \bg{R}_{kij}{}^k \eand \bg{R}_{ijk}{}^l\ \coloneqq\ \partial_i\bg{\Gamma}_{jk}{}^l-\partial_j\bg{\Gamma}_{ik}{}^l-\bg{\Gamma}_{ik}{}^m\bg{\Gamma}_{jm}{}^l+\bg{\Gamma}_{jk}{}^m\bg{\Gamma}_{im}{}^l~,
        \end{equation}
        respectively. Letting $\bg{g}^{ij}$ be the inverse of $\bg{g}_{ij}$, the \uline{Beltrami Laplacian} is given by
        \begin{equation}\label{eq: BeltramiLaplace}
        	\bg{\lap}_{\rm B}\ \coloneqq\ \bg{g}^{ij}\bg{\nabla}_i\bg{\nabla}_j\ =\ \bg{\nabla}^i\bg{\nabla}_i~.
        \end{equation}
        Indices are raised and lowered using $\bg{g}^{ij}$ and $\bg{g}_{ij}$ respectively and we always use Einstein's summation convention. In deriving~\eqref{eq:generalNavierStokesLocal}, we have used the standard \uline{Weitzenb{\"o}ck formula}
        \begin{equation}\label{eq:Weitzenboeck}
            (\bg{\lap}_{\rm H}\rho)_{i_1\cdots i_p}\ =\ -\bg{\lap}_{\rm B}\rho_{i_1\cdots i_p}+p\bg{R}_{j[i_1}\rho^j{}_{i_2\cdots i_p]}+\tfrac12p(p-1)\bg{R}_{jk[i_1i_2}\rho^{jk}{}_{i_3\cdots i_p]}
        \end{equation}
        for a differential $p$-form $\rho=\frac{1}{p!}\rho_{i_1\cdots i_p}\rmd x^{i_1}\wedge\ldots\wedge\rmd x^{i_p}$ on $M$ that relates~\eqref{eq: HodgeLaplace} and~\eqref{eq: BeltramiLaplace}. Here and in the following, parentheses (respectively, square brackets) denote \uline{normalised symmetrisation} (respectively, \uline{anti-symmetrisation}) of the enclosed indices.

        Evidently, when $M$ is $\IR^m$ with the standard Euclidean metric $\delta_{ij}=1$ for $i=j$ and zero otherwise, the equations~\eqref{eq:generalNavierStokesLocal} reduce to the more familiar equations
        \begin{subequations}
            \begin{equation}\label{eq:standardNSEuclidean}
                \parder[v^i]{t}\ =\ -v^j\partial_j v^i-\partial^ip+\nu\lap v^i
            \end{equation}
            and
            \begin{equation}
                \partial_iv^i\ =\ 0~,
            \end{equation}
        \end{subequations}
        where now $\lap\coloneqq\partial^i\partial_i$ is the standard Euclidean Laplacian.

        \paragraph{Pressure constraint.}
        Upon applying the codifferential to~\eqref{eq:generalNavierStokesA} and using the divergence-free constraint~\eqref{eq:generalNavierStokesB}, we obtain the so-called \uline{pressure equation}
        \begin{equation}
            \bg{\lap}_{\rm H}p\ =\ -|\rmd v|^2+{\star_{\bg{g}}(v\wedge{\star_{\bg{g}}\bg{\lap}_{\rm H}v})}-\tfrac12\bg{\lap}_{\rm H}|v|^2~.
        \end{equation}
        In local coordinates, this becomes
        \begin{equation}\label{eq:pressureLaplacian}
            \bg{\lap}_{\rm B}p\ =\ -(\bg{\nabla}_iv_j)(\bg{\nabla}^jv^i)-v^iv^j\bg{R}_{ij}~,
        \end{equation}
        where we have again used the Weitzenb{\"o}ck formula~\eqref{eq:Weitzenboeck}. 

        Upon setting
        \begin{equation}\label{eq:vorticityAndStrain}
            \zeta_{ij}\ \coloneqq\ \bg{\nabla}_{[i}v_{j]}\ =\ \partial_{[i}v_{j]}
            \eand
            S_{ij}\ \coloneqq\ \bg{\nabla}_{(i}v_{j)}~,
        \end{equation}
        which are called the \uline{vorticity two-form} and the \uline{rate-of-strain tensor}, respectively, the pressure equation~\eqref{eq:pressureLaplacian} can be written in a more standard form as
        \begin{equation}\label{eq:pressureLapVorticityStrain}
            \bg{\lap}_{\rm B}p\ =\ \zeta_{ij}\zeta^{ij}-S_{ij}S^{ij}-v^iv^j\bg{R}_{ij}~.
        \end{equation}
        By definition, the rate-of-strain tensor vanishes if any only if the velocity vector field is a Killing vector field.

        Furthermore, because of the Poincar{\'e} lemma, on an open and contractible\footnote{In the following, we declare a \uline{neighbourhood} (of a point) to be an open and contractible set.} set $U\subseteq M$, the divergence-free constraint~\eqref{eq:generalNavierStokesB} can always be solved as
        \begin{equation}\label{eq:localVelocity}
            v\ =\ {\star_{\bg{g}}\rmd\psi}
            \efor
            \psi\ \in\ \Omega^{m-2}(U)
            \quad
            \Longleftrightarrow
            \quad
            v^i\ =\ \frac{\sqrt{\det(\bg{g})}}{(m-2)!}\eps^{i_1\cdots i_{m-1}i}\partial_{i_1}\psi_{i_2\cdots i_{m-1}}~,
        \end{equation}
        where $\eps_{i_1\cdots i_m}$ is the Levi-Civita symbol with $\eps_{1\cdots m}=1$; note that $\eps^{1\cdots m}=\frac{1}{\det(\bg{g})}\eps_{1\cdots m}$. Upon substituting this expression into~\eqref{eq:pressureLaplacian}, we obtain a Monge--Amp{\`e}re-type equation for $\psi$. For $m=2$, $\psi$ is known as the \uline{stream function}, and we obtain a Monge--Amp{\`e}re equation in a familiar setting. Generally, we may refer to $\psi\in\Omega^{m-2}(U)$ as the \uline{stream $(m-2)$-form}. 

        \begin{remark}
            The viscosity term in the Navier--Stokes equations~\eqref{eq:generalNavierStokesA} may be modified as
            \begin{equation}
                \parder[v]{t}\ =\ -(-1)^m{\star_{\bg{g}}(v\wedge{\star_{\bg{g}}\rmd v})}-\tfrac12\rmd|v|^2-\rmd p-\nu[\bg{\lap}_{\rm H}v+c\,\bg{\sfR\sfi\sfc}(v)]~,
            \end{equation}
            where $c\in\IR$ and $\bg{\sfR\sfi\sfc}(\rho)\coloneqq\bg{R}_i{}^j\rho_j\rmd x^i$ for any $\rho\in\Omega^1(M)$. Evidently, in the flat case, the extra term vanishes and this modified equation again reduces to the standard equation~\eqref{eq:standardNSEuclidean}. The coordinate version~\eqref{eq:generalNavierStokesLocalA} then becomes
            \begin{equation}
                \parder[v^i]{t}\ =\ -v^j\bg{\nabla}_jv^i-\partial^ip+\nu[\bg{\lap}_{\rm B}v^i-(c+1)\bg{R}^{ij}v_j]~.
            \end{equation}
            When $c=0$, we return to the situation given by~\eqref{eq:generalNavierStokesA}, and this version of the Navier--Stokes equations was perhaps first studied in the seminal work~\cite{Ebin:1970aa} and more recently in e.g.~\cite{Cao:1999aa,Kobayashi:2008aa}. The case $c=-1$ was discussed in e.g.~\cite{Pierfelice:2017aa} and the case $c=-2$ in e.g.~\cite{Ebin:1970aa,Chan:2017aa,Samavaki:2020aa} and the references therein. For instance, under the assumption of the divergence-free constraint~\eqref{eq:generalNavierStokesIncompressibilityLocal}, when $c=-2$ and with $S_{ij}$ the rate-of-strain tensor~\eqref{eq:vorticityAndStrain}, it can straightforwardly be seen that the viscosity term can be rewritten as
            \begin{equation}
                \bg{\lap}_{\rm B}v^i+\bg{R}^{ij}v_j\ =\ 2\bg{\nabla}_jS^{ij}~.
            \end{equation}
            Hence, in this case, for velocity fields that preserve the metric (i.e.~that are Killing), the viscosity term drops out from the Navier--Stokes equations. Generally, with the $c$-term switched on, the pressure equation~\eqref{eq:pressureLaplacian} takes the form
            \begin{equation}
                \bg{\lap}_{\rm B}p\ =\ -(\bg{\nabla}_iv_j)(\bg{\nabla}^jv^i)-v^iv^j\bg{R}_{ij}-\nu c\big(\bg{R}_{ij}S^{ij}+\tfrac12v^i\partial_i\bg{R}\big)\,,
            \end{equation}
            where $\bg{R}\coloneqq\bg{g}^{ij}\bg{R}_{ij}$ is the curvature scalar and we have used the well-known identity
            \begin{equation}
            	\bg{\nabla}^i\big(\bg{R}_{ij}-\tfrac12\bg{g}_{ij}\bg{R}\big)\ =\ 0~.
            \end{equation}
            Henceforth, we shall \uline{always} assume that $c=0$. Our results and conclusions remain unchanged, and all of the formul{\ae} can easily be adjusted to accommodate the $c$-term.
        \end{remark}

		\subsection{Geometric properties of fluids in two dimensions}\label{sec:geometry2dFlows}

		One of the enduring challenges of fluid mechanics is to understand the topology of vortices. However, no systematic method has been developed to extract topological information from the partial differential equations governing the flow. To this end, we now make some observations about how Monge--Amp{\`e}re geometry might provide new insights by focusing on incompressible flows in two spatial dimensions. We start off by recalling some of the key aspects of Monge--Amp{\`e}re geometry. See also \cref{app:MongeAmpereStructures}. We shall be rather brief and only list the relevant material for our discussion, and we refer the interested reader to the text book~\cite{Kushner:2007aa} for more details.

        \paragraph{Monge--Amp{\`e}re structures.}
        Let us consider the cotangent bundle $\pi:T^*M\rightarrow M$ of an $m$-dimensional manifold $M$, which we coordinatise by $(x^i,q_i)$ with $x^i$ local coordinates on $M$ as before and $q_i$ local fibre coordinates. Then, $T^*M$ comes with a canonical symplectic structure $\omega$ which in local coordinates is $\omega=\rmd q_i\wedge\rmd x^i$. Following~\cite{Lychagin:1993aa} (see also~\cite{Lychagin:1979aa}) a differential $m$-form $\alpha\in \Omega^m(T^*M)$ is called \uline{$\omega$-effective} whenever $\omega\wedge\alpha =0$. Furthermore, the pair $(\omega,\alpha)$, with $\alpha$ an $\omega$-effective $m$-form, will be called a \uline{Monge--Amp{\`e}re structure}~\cite{Banos:2002aa}. In this context, we shall refer to $\alpha$ as the \uline{Monge--Amp{\`e}re form}. We draw attention to the purpose of the requirement that the Monge--Amp{\`e}re form $\alpha$ is $\omega$-effective. This constraint removes redundancy that would occur if $\alpha$ were an arbitrary differential $m$-form, two of which produce the same Monge--Amp{\`e}re equation if and only if their difference is a differential form which is not $\omega$-effective~\cite{Lychagin:1993aa,Lychagin:1979aa}. \cref{thm:HodgeLepageLychagin} then tells us that the $\omega$-effective piece of a differential $m$-form uniquely determines the Monge--Amp{\`e}re equation.
        
        A \uline{generalised solution} for a Monge--Amp{\`e}re structure $(\omega,\alpha)$ is a Lagrangian submanifold $\iota:L\hookrightarrow T^*M$ with respect to $\omega$, that is, $\iota^*\omega=0$ and $\dim(L)=\dim(M)$, for which, in addition, we have $\iota^*\alpha=0$. In particular, the section $\rmd\psi:M\to T^*M$ associated with the function $\psi\in\scC^\infty(M)$ and locally given by $x^i\mapsto(x^i,q_i)=(x^i,\partial_i\psi)$ defines a Lagrangian submanifold $L_M\coloneqq\rmd\psi(M)$. Additionally, the requirement that a generalised solution satisfies $\iota^*\alpha =0$ then reads $(\rmd\psi)^*\alpha=0$, which in turn yields a \uline{Monge--Amp{\`e}re equation} for $\psi$. In this case, the functions $\psi\in\scC^\infty(M)$ satisfying the Monge--Amp{\`e}re equation and the corresponding generalised solutions described by $\rmd\psi$ are both referred to as \uline{classical solutions}. We shall call  Monge--Amp{\`e}re structures $(\omega,\alpha)$ and $(\omega,\alpha')$ \uline{symplectically equivalent} whenever there is a symplectomorphism $\Phi\in\scC^\infty(T^*M)$ such that $\alpha'=\Phi^*\alpha$.
   
        Moreover, as explained in \cref{app:LagrangianSubmanifolds}, a Lagrangian submanifold $L$ is locally a section $\rmd\psi:U\rightarrow T^*M$ for some $\psi\in\scC^\infty(U)$ and $U\subseteq M$ open and contractible if and only if the map $\pi|_L\coloneqq\pi\circ\iota:L\rightarrow M$ is a local diffeomorphism. In this case, we may take $x^i$ as local coordinates on $L$ and we have $\iota:x^i\mapsto(x^i,q_i)=(x^i,\partial_i\psi)$. However, an arbitrary generalised solution may exhibit singular behaviour where the projection $\pi|_L$ fails to be an immersion~\cite{Ishikawa:2006aa,Ishikawa:2015aa,Lychagin:1985aa} and $L$ is not locally described by the coordinates $x^i$. Recent work~\cite{Donofrio:2022aa} studying the semi-geostrophic equations has shown that such projection singularities may be related to the degeneracy of a specific metric on $L$. In order to isolate behaviour of $L$ which is due to the variation of vorticity and strain, hereafter we predominantly consider solutions which are (locally) described by a section.
           
        \paragraph{Monge--Amp{\`e}re geometry of two-dimensional fluid flows.}
        Let us now specialise to incompressible fluid flows in $m=2$ dimensions. In this case, the components~\eqref{eq:curvatureTensors} of the Ricci and Riemann curvature tensors simplify to 
        \begin{equation}
        	\bg{R}_{ij}\ =\ \tfrac{\bg{R}}{2}\bg{g}_{ij}
            \eand
            \bg{R}_{ijk}{}^l\ =\ \bg{R}\bg{g}_{k[j}\delta_{i]}{}^l
        \end{equation}
        respectively, where $\bg{R}$ is the curvature scalar. Furthermore, we denote by $\sfHess(\psi)$ the \uline{Hessian} of a function $\psi\in\scC^\infty(M)$. Explicitly, in local coordinates, it reads as $\sfHess(\psi)=(\bg{\nabla}_i\partial_j\psi)=(\bg{\nabla}_j\partial_i\psi)$. In this case,~\eqref{eq:localVelocity} yields
        \begin{equation}
        	v^i\ =\ -\sqrt{\det(\bg{g})}\eps^{ij}\partial_j\psi\,,
        \end{equation}
        and the pressure equation~\eqref{eq:pressureLaplacian}, on an open and contractible set $U\subseteq M$, becomes 
        \begin{equation}\label{eq:pressureLaplacian2d}
            \tfrac12\bg{\lap}_{\rm B}p\ =\ \det(\bg{g}^{-1}\sfHess(\psi))-\tfrac{\bg{R}}{4}|\rmd\psi|^2
            \quad
            \Longleftrightarrow
            \quad
            \tfrac12\bg{\nabla}^i\partial_ip\ =\ \det(\bg{\nabla}^i\partial_j\psi)-\tfrac{\bg{R}}{4}(\bg{\nabla}^i\psi)(\partial_i\psi)~.
        \end{equation}
        This can be understood as a Monge--Amp{\`e}re equation for the stream function and hence for the velocity field. The vorticity two-form~\eqref{eq:vorticityAndStrain} can be written as
        \begin{equation}\label{eq:vorticity2d}
            \zeta_{ij}\ =\ \tfrac12\sqrt{\det(\bg{g})}\,\eps_{ij}\,\zeta
            \ewith
            \zeta\ \coloneqq\ \bg{\lap}_{\rm B}\psi
            \quad\Longrightarrow\quad
            \zeta_{ij}\zeta^{ij}\ =\ \tfrac12\zeta^2~.
        \end{equation}

        Importantly, the pressure equation~\eqref{eq:pressureLaplacian2d} arises from a Monge--Amp{\`e}re structure on $T^*M$. Indeed, upon fixing the notation
        \begin{subequations}
			\begin{equation}\label{eq:PressureCurvature2d}
				\hat f\ \coloneqq\ \tfrac12 \bg{\lap}_{\rm B}p + \tfrac{\bg{R}}{4}|q|^2 \eand \bg{\nabla}q_i\ \coloneqq\ \rmd q_i-\rmd x^j\bg{\Gamma}_{ji}{}^kq_k~,
			\end{equation}       	
        	with $\bg{\Gamma}_{ij}{}^k$ the Christoffel symbols for $\bg{g}_{ij}$, it is readily checked that the differential forms\footnote{Note that $\bg{\nabla}q_i\wedge\rmd x^i=\rmd q_i\wedge\rmd x^i$.} 
        	\begin{equation}\label{eq:MongeAmpereStructure2d}
           		\begin{gathered}
                	\omega\ \coloneqq\ \bg{\nabla}q_i\wedge\rmd x^i~,
                	\\
                	\alpha\ \coloneqq\ \tfrac{\sqrt{\det(\bg{g})}}{2}\big[\eps^{ij}\bg{\nabla}q_i\wedge\bg{\nabla}q_j-\hat f\eps_{ij}\rmd x^i\wedge\rmd x^j\big]
            	\end{gathered}
        	\end{equation}
        \end{subequations}
        on $T^*M$ form a Monge--Amp{\`e}re structure on $T^*M$. For Lagrangian submanifolds $\iota:L\hookrightarrow T^*M$ which are locally $\rmd\psi:U\rightarrow T^*M$ for some $\psi\in\scC^\infty(U)$ with $U\subseteq M$ open and contractible, whilst $\iota^*\omega=0$ is automatic, the condition $\iota^*\alpha=0$ is equivalent to $\psi$ satisfying the Monge--Amp{\`e}re equation~\eqref{eq:pressureLaplacian2d}. In conclusion, the Monge--Amp{\`e}re equation~\eqref{eq:pressureLaplacian2d} arises from the Monge--Amp{\`e}re structure~\eqref{eq:MongeAmpereStructure2d}. 

        Note that $\alpha$ is non-degenerate if and only if $\hat f\neq 0$, and it is shown in \cref{sec:geometricProperties3d} that $\alpha$ is closed. We also note that pulling back $\alpha$ via $v = \star_{\bg{g}}\rmd\psi$ again yields~\eqref{eq:pressureLaplacian2d}; this observation shall inform the alternative Monge--Amp{\`e}re structure chosen in \cref{sec:geometricProperties3d}, which naturally generalises to higher dimensions.

        \paragraph{Almost (para-)Hermitian structure.}
        Next, following~\cite{Lychagin:1993aa}\footnote{see also~\eqref{eq:almostStructure4D}.}, we associate with the Monge--Amp{\`e}re structure~\eqref{eq:MongeAmpereStructure2d} an endomorphism $\hat J$ of the tangent bundle of $T^*M$ defined by
        \begin{equation}\label{eq:definitionJ2d}
            \frac{\alpha}{\sqrt{|\hat f|}}\ \eqqcolon\ \hat J\intprod\omega
        \end{equation}
        with $\hat f$ as defined in~\eqref{eq:MongeAmpereStructure2d} and under the assumption that $\hat f$ does not vanish. By virtue of the results of~\cite{Lychagin:1993aa}, $\hat J$ is an almost complex structure on $T^*M$ when $\hat f>0$ (in which case the Monge--Amp{\`e}re equation~\eqref{eq:pressureLaplacian2d} is elliptic) and an almost para-complex structure on $T^*M$ when $\hat f<0$ (in which case the Monge--Amp{\`e}re equation~\eqref{eq:pressureLaplacian2d} is hyperbolic). As can be checked following the arguments of~\cite{Lychagin:1993aa,Kruglikov:1999aa,Kushner:2007aa}, this structure is integrable if and only if $\hat f$ is constant.\footnote{This then necessarily means that the curvature of $M$ vanishes, i.e. $M$ is flat.} 

        Furthermore, as discussed in \cref{app:MongeAmpereStructures}, we can always find a differential two-form $\hat K$ which is of type $(1,1)$ with respect to $\hat J$, such that $\hat K\wedge\omega=0$, $\hat K\wedge(\hat J\intprod\omega)=0$, and $\hat K\wedge\hat K\neq 0$. Explicitly, we may take
        \begin{equation}\label{eq:almostKaehlerForm2d}
            \hat K\ \coloneqq\ -\sqrt{|\hat f|}\,\bg{\nabla}q_i\wedge{\star_{\bg{g}}\rmd x^i}~.
        \end{equation}
		
        Since $\hat K(\hat JX,Y)=-\hat K(X,\hat JY)$ for all $X,Y\in\frX(T^*M)$, we are naturally led to the almost (para-)Hermitian metric $\hat g(X,Y)\coloneqq\hat K(X,\hat JY)$ on $T^*M$ for all $X,Y\in\frX(T^*M)$, which is explicitly given by 
        \begin{equation}\label{eq:fluidMetric2d}
            \hat g\ =\ \tfrac12\hat f\bg{g}_{ij}\rmd x^i\odot\rmd x^j+\tfrac12\bg{g}^{ij}\bg{\nabla}q_i\odot\bg{\nabla}q_j~.
        \end{equation} 
        Evidently, in the elliptic case, when $\hat f>0$, the metric $\hat g$ is Riemannian, whilst in the hyperbolic case, when $\hat f<0$, the metric is Kleinian.

        \paragraph{Pull-back metric.}
        It is easily seen that the pull-back $g\coloneqq\iota^*\hat g$ of~\eqref{eq:fluidMetric2d} to the Lagrangian submanifold $\iota:L \hookrightarrow T^*M$ via $\rmd\psi$ is
        \begin{equation}\label{eq:fluidMetric2dPullBack}
            g\ =\ \tfrac12g_{ij}\rmd x^i\odot\rmd x^j
            \ewith
            g_{ij}\ \coloneqq\ \zeta\bg{\nabla}_i\partial_j\psi~,
        \end{equation}
        where we have used that 
        \begin{equation}\label{eq:pullBackHatf2D}
            f\ \coloneqq\ \iota^*\hat f\ =\ \tfrac12\bg{\lap}_{\rm B}p+\tfrac{\bg{R}}{4}|\rmd\psi|^2\ =\ \det(\bg{\nabla}^i\partial_j\psi)
        \end{equation}
        by~\eqref{eq:pressureLaplacian2d} and substituted~\eqref{eq:vorticity2d}. Clearly, in regions where the vorticity vanishes, this metric vanishes as well. When both $\tr(\bg{g}^{ik}g_{kj})>0$ and $\det(\bg{g}^{ik}g_{kj})>0$, it follows that $g$ is Riemannian. The former condition is always satisfied since $\tr(\bg{g}^{ik}g_{kj})=\zeta^2$, and the latter is satisfied if and only if $f>0$. Similarly, when $f<0$, $g$ is Kleinian. Hence, the signature of $g$ is independent of the sign of the vorticity~\eqref{eq:vorticity2d} and only depends on the sign of $f$.
        
        Upon comparing~\eqref{eq:pressureLapVorticityStrain} and~\eqref{eq:pullBackHatf2D}, we find that $f=\frac12(\zeta_{ij}\zeta^{ij}-S_{ij}S^{ij})$ with the indices on $\zeta_{ij}$ and $S_{ij}$ raised with the background metric. Hence, when $f>0$ and the metric $g$ is Riemannian, vorticity dominates, yet when $f<0$ and $g$ is Kleinian, strain dominates. This statement covariantly extends the pressure criterion for a vortex, as given in~\cite{Larcheveque:1990aa,Larcheveque:1993aa}, to an arbitrary Riemannian background manifold, while accounting for the underlying curvature. The standard criterion are recovered on a flat background.
        
		Now that we have criterion for testing the dominance of vorticity and strain of a flow on a Riemannian manifold, we discuss how to obtain topological information about the flow.

		\paragraph{Local Gau{\ss}--Bonnet theorem.}
		Let $L$ be a Lagrangian submanifold of $T^*M$, which is locally described by sections $\rmd\psi_U:U\rightarrow T^*M$, with $U\subseteq M$ open and contractible, and $\psi_U\in\scC^\infty(U)$ the stream function on $U$. Furthermore, let $\Sigma\subseteq U$ be a compact region in $U$, on which $f>0$. We can then define the compact region $L_\Sigma\subseteq L$ by $L_\Sigma\coloneqq\rmd\psi_U(\Sigma)$. It is now natural to consider the question of how we might use the \uline{local Gau{\ss}--Bonnet theorem} to relate the geometry of $L_\Sigma$ to its topology, as given by the Euler characteristic $\chi(L_\Sigma)$. For the reader's convenience, let us state this theorem, see e.g.~\cite[Theorem 4.2]{Rotskoff:2010aa} for details. 

        \begin{theorem}
            Let $\Sigma$ be a two-dimensional, compact, oriented Riemannian manifold with metric $g$. Suppose that $\Sigma$ has a boundary composed of disjoint, simple, closed, piecewise regular, piecewise arc-length parametrised curves $\gamma_\alpha$, that is, $\partial\Sigma=\bigcup_\alpha\gamma_\alpha$. Let $R$ be the curvature scalar of the Levi-Civita connection of $g$, $\vol{\Sigma}$ the volume form on $\Sigma$, and $\kappa$ the geodesic curvature. Furthermore, let $\varphi_\beta$ be the exterior angles at the non-smooth points of the boundary $\partial\Sigma$. Then, the Euler number $\chi(\Sigma)$ of $\Sigma$ is given by
            \begin{equation}\label{eq:localGaussBonnet}
                \tfrac12\int_\Sigma\vol{\Sigma}\,R+\sum_\alpha\int_{\gamma_\alpha}\rmd s\,\kappa(\gamma_\alpha(s))+\sum_\beta\varphi_\beta\ =\ 2\pi\chi(\Sigma)~.
            \end{equation}
        \end{theorem}

        \noindent
        Let $\Sigma\subseteq U\subseteq M$ be as above. Then, $\chi(\Sigma)=\chi(L_\Sigma)$, since $\pi|_L$ is now a diffeomorphism. For instance, if such $\Sigma$ is bounded by a simple, closed curve such as closed, isovortical contour, or a closed stream-line, then $\Sigma$ is homeomorphic to a disc and so, $\chi(L_\Sigma)=1$. 

        \paragraph{Christoffel symbols and curvatures.}
        Let us now give some of the formul{\ae} needed when evaluating~\eqref{eq:localGaussBonnet}. In particular, we introduce the notation
        \begin{equation}\label{eq:psiIndexNotation}
            \psi_{i_1\cdots i_n}\ \coloneqq\ \bg{\nabla}_{(i_1}\cdots\bg{\nabla}_{i_{n-1}}\partial_{i_n)}\psi~.
        \end{equation}
        A quick calculation shows that $\psi_{i_1\cdots i_n}$ can be expressed in terms of the components of the rate-of-strain tensor and the vorticity two-form, see~\eqref{eq:vorticityAndStrain} and~\eqref{eq:vorticity2d}, as
        \begin{equation}\label{eq:gradientsVorticityStrain}
            \psi_{i_1\cdots i_n}\ =\ -\sqrt{\det(\bg{g})}\bg{g}^{jk}\eps_{j(i_1}\bg{\nabla}_{i_2}\cdots\bg{\nabla}_{i_{n-1}}S_{i_n)k}+\tfrac12\bg{g}_{(i_1i_2}\bg{\nabla}_{i_3}\cdots\bg{\nabla}_{i_{n-1}}\partial_{i_n)}\zeta
        \end{equation}
        for $n>1$. Then, using~\eqref{eq:vorticity2d}, we can write the metric~\eqref{eq:fluidMetric2dPullBack} as $g_{ij}=\zeta\tilde g_{ij}$ with $\tilde g_{ij}\coloneqq\psi_{ij}$. Hence, due to its conformal nature, the Christoffel symbols $\Gamma_{ij}{}^k$ of $g_{ij}$ take the form
        \begin{subequations}
            \begin{equation}\label{eq:pullback2DChristoffel}
                \Gamma_{ij}{}^k\ =\ \tilde\Gamma_{ij}{}^k+\partial_{(i}\delta_{j)}{}^k\log(|\zeta|)-\tfrac12\tilde g_{ij}\tilde g^{kl}\partial_l\log(|\zeta|)~,
            \end{equation}
            where $\tilde g^{ij}$ denotes the inverse of $\tilde g_{ij}$, and the $\tilde\Gamma_{ij}{}^k$ are the Christoffel symbols of the Hessian metric $\tilde g_{ij}$, 
            \begin{equation}\label{eq:hessianChristoffel}
                \tilde\Gamma_{ij}{}^k\ =\ \bg{\Gamma}_{ij}{}^k+\tfrac12\Upsilon_{ijl}\tilde g^{lk}
                \ewith
                \Upsilon_{ijk}\ \coloneqq\ \psi_{ijk}+\tfrac43\psi_l\bg{R}_{k(ij)}{}^l~.
            \end{equation}
        \end{subequations}
        Consequently, the curvature scalar $R$ of $g_{ij}$ is given by
        \begin{subequations}\label{eq:curvature2D}
            \begin{equation}\label{eq:scalarCurvature2D}
                R\ =\ \frac{1}{\zeta}\left\{\tilde R-\tfrac{1}{\sqrt{|\det(\tilde g)|}}\partial_i\Big[\sqrt{|\det(\tilde g)|}\,\tilde g^{ij}\partial_j\log(|\zeta|)\Big]\right\},
            \end{equation}
            where $\tilde R$ is the curvature scalar for $\tilde g_{ij}$,
            \begin{equation}\label{eq:HessianScalarCurvature}
                \begin{aligned}
                    \tilde R\ &=\ \tfrac12\tilde g^{ij}\bg{g}_{ij}\bg{R}-\tfrac14\tilde g^{ij}\tilde g^{kl}\tilde g^{mn}(\Upsilon_{ijm}\Upsilon_{kln}-\Upsilon_{ikm}\Upsilon_{jln})
                    \\
                    &\kern1cm+\tfrac23\tilde g^{ij}\tilde g^{kl}\big[\psi_{mn}\big(\delta^m_i\bg{R}_{j(kl)}{}^n-\delta^m_j\bg{R}_{l(ik)}{}^n\big)+\psi_m\big(\bg{\nabla}_i\bg{R}_{j(kl)}{}^m-\bg{\nabla}_j\bg{R}_{l(ik)}{}^m\big)\big]\,,
                \end{aligned}
            \end{equation}
        \end{subequations}
        see \cref{app:pullbackCurvature} for details. Importantly, no fourth-order derivatives of the stream function appear, and in that sense, the curvature scalar of the pull-back metric~\eqref{eq:fluidMetric2dPullBack} is generated by gradients of vorticity and strain, see~\eqref{eq:gradientsVorticityStrain}. In addition, $\psi_i$ occurs without any further derivatives, hence the curvature scalar depends also on the components of velocity directly. 
        
        Furthermore, given an arc-length parametrised curve $\gamma:s\rightarrow(y^1(s),y^2(s))$ in two dimensions, we may use Beltrami's formula 
        \begin{equation}\label{eq:geodesicCurvature}
            \kappa(\gamma(s))\ =\ \sqrt{|\det(g(y(s)))|}\,\eps_{ij}\,\dot y^i(s)\big[\ddot y^j(s)+\Gamma_{kl}{}^j(y(s))\dot y^k(s)\dot y^l(s)\big]~,
        \end{equation}
        for the geodesic curvature $\kappa$ at a point $\gamma(s)$ of the curve. Here, the superposed dots indicate derivatives with respect to the arc-length parameter $s$. 
         
        Let us return to our previous example, where $\Sigma\subseteq U\subseteq M$ with $f>0$ and a boundary given by a simple, regular, closed curve $c:\IR\rightarrow U$. As $\rmd\psi_U$ is a diffeomorphism on $U$, it follows that the boundary of $L_\Sigma\coloneqq\rmd\psi_{U}(\Sigma)$ is given by $\gamma\coloneqq\rmd\psi_U\circ c:\IR\rightarrow L$, which is also a simple, regular, closed curve and may be assumed to be arclength-parametrised without loss of generality. Consequently,~\eqref{eq:localGaussBonnet} evaluates to
        \begin{equation}\label{eq:localGaussBonnetExample}
            \int_{\gamma}\rmd s\,\kappa(\gamma(s))\ =\ 2\pi-\tfrac12\int_{L_\Sigma}\vol{L_{\Sigma}}\,R~,
        \end{equation}
        on $L_\Sigma$, where $R$ is given by~\eqref{eq:scalarCurvature2D}. That is, the mean curvature of the boundary is determined by the average curvature of the interior. Noting~\eqref{eq:scalarCurvature2D},~\eqref{eq:HessianScalarCurvature}, and~\eqref{eq:geodesicCurvature}, we remark that at a formal qualitative level, the local Gau{\ss}--Bonnet relation~\eqref{eq:localGaussBonnetExample} is a statement to the effect that\footnote{Recall here that the boundary of $L_\Sigma$ is given by the image of the boundary of $\Sigma$, that is, the image of the closed stream-line bounding a candidate vortex in $M$, under $\rmd\psi_U$.}
        \begin{equation}
            \begin{aligned}
                &\mbox{mean curvature of the boundary of } L_\Sigma \ =\
                \\
                &\kern1cm=\ 2\pi-\mbox{mean gradients of vorticity and strain}~.
            \end{aligned}
        \end{equation}

		In this sense, we can use Monge--Amp{\`e}re geometry, when $f>0$, to assign a topological invariant to a `vortex' described by $L_\Sigma$ --- the image of the graph of the gradient of the stream function, over a compact region of $M$ bounded by some closed stream-line. Whilst the framework described here is an elaborate mechanism for determining the Euler number of a vortex patch, it illuminates a relationship between vortex topology and the physical phenomena, such as the gradients of vorticity and strain, that determine certain topological properties of the flow via the topology of $L$ and the diffeomorphic nature of the projection $\pi|_L$. When $\pi|_L$ fails to be a diffeomorphism, then singular behaviour may be anticipated (note the recent work by some of the authors in~\cite{Donofrio:2022aa} focuses on this problem in the context of the semi-geostrophic equations of meteorological flows). When the pull-back metric~\eqref{eq:fluidMetric2dPullBack} is Kleinian, the Gau{\ss}--Bonnet theorem can be extended to such cases, under certain conditions pertaining to the boundary $\partial L_\Sigma$ --- e.g.~it should have no null segments --- however, the link between topology as quantified by the Euler characteristic and the Gau{\ss}--Bonnet theorem becomes tenuous~\cite{Birman:1984aa,Steller:2006aa}.

        \subsection{Examples in two dimensions}\label{sec:examples2D}

        We now consider some classical examples of flows on $\IR^2$ with metric $\bg{g}_{ij}=\delta_{ij}$. As noted in \cref{sec:intro}, a solution for~\eqref{eq:PEP} is not necessarily one for the Navier--Stokes equations, and our first two examples fall into this category. However, they illustrate how the topology of flows can change with time and it is useful to view such phenomena from the point of view of Monge--Amp{\`e}re geometry. The second example can be turned into a solution to the Navier--Stokes equations by adding higher-order terms~\cite{Moffatt:2001aa}, yet the basic topological features on which we focus (as did~\cite{Moffatt:2001aa}) are most clearly illustrated in the form presented below. Our final example, the Taylor--Green vortex, is a solution to Navier--Stokes equations.
        
        \paragraph{Preliminaries.}
        For convenience, let us summarise the relevant simplified formul{\ae} first, adopting the notation $x\coloneqq x^1$ and $y\coloneqq x^2$. Working with a flat background metric, $\bg{R}=0$, and so we find for $\hat f$ given in~\eqref{eq:MongeAmpereStructure2d} and $f$ given below~\eqref{eq:fluidMetric2dPullBack} that
        \begin{equation}
			\hat f\ =\ \tfrac12\lap p\ =\ \partial_x^2\psi \partial_y^2\psi - (\partial_x\partial_y\psi)^2 =\ f
            \ewith
            \lap\ \coloneqq\ \partial_x^2+\partial_y^2~.
        \end{equation}
        Hence, the metric~\eqref{eq:fluidMetric2d} on $T^*\IR^2$ takes the form
        \begin{equation}\label{eq:flatBack2DMetric}
            \hat g\ =\ 
            \begin{pmatrix}
                f\unit_2 & 0
                \\
                0 & \unit_2
            \end{pmatrix}
        \end{equation}
        with its signature dictated by the sign of $f$. This is singular if and only if $f=0$, and the corresponding curvature scalar~\eqref{eq:curvatureScalarFluidMetric} becomes
        \begin{equation}\label{eq:flatBackCurvatureLR2D}
            \hat R\ =\ \tfrac{1}{f^3}\big(\partial_xf\partial_xf+\partial_yf\partial_yf-f\lap f\big)\,.
        \end{equation}
        Thus, at a stationary point of $f$, the sign of $\hat R$ is determined by the sign of $\lap f$. Consequently, when $f$ accumulates and has a local maximum, $\lap f<0$ and $\hat R>0$. 
        
        The vorticity~\eqref{eq:vorticity2d} is simply $\zeta=\lap\psi$ for the stream function $\psi=\psi(x,y)$, so the pull-back metric~\eqref{eq:fluidMetric2dPullBack} becomes
        \begin{equation}\label{eq:flatBack2DPullback}
        	g\ =\ \zeta
    		\begin{pmatrix}
    			\partial_x^2\psi & \partial_x\partial_y\psi 
    			\\
    			\partial_x\partial_y\psi & \partial_y^2\psi
    		\end{pmatrix} 
            \ =\ \frac\zeta2
    		\begin{pmatrix}
    			\zeta+2S_{xy} & -2S_{xx}
    			\\
    			-2S_{xx} & \zeta-2S_{xy}
    		\end{pmatrix},
        \end{equation}
        where $S_{xx}=-S_{yy}$ and $S_{xy}$ are the components of rate-of-strain tensor~\eqref{eq:vorticityAndStrain}, describing a \uline{shearing deformation} at an angle of $\frac12\arctan\Big(\frac{S_{xy}}{S_{xx}}\Big)$, without overall dilation, since our flow is divergence-free~\cite{Roulstone:2013aa,Landau:1987aa}. We note that $g$ is singular when the vorticity vanishes, in addition to when the Hessian part of the metric is singular, that is, where $f=0$. We shall discuss these points in due course. We also note that when $f$ depends on time $t$, then the metric~\eqref{eq:flatBack2DMetric} will depend on $t$ as a parameter. The same is true for~\eqref{eq:flatBack2DPullback} via the time-dependence of vorticity and rate-of-strain. The one-parameter family of metrics~\eqref{eq:flatBack2DMetric} and~\eqref{eq:flatBack2DPullback} will thus evolve according to either the Euler or the Navier--Stokes equations. 

        Another rotational invariant of the \uline{velocity-gradient matrix} 
        \begin{equation}
            A\ \coloneqq\ 
            \begin{pmatrix}
                -\partial_x\partial_y\psi & -\partial^2_y\psi
                \\
                \partial^2_x\psi & \partial_x\partial_y\psi
            \end{pmatrix},
        \end{equation}
        the \uline{resultant deformation} $D_{\rm R}$~\cite{Roulstone:2013aa}, occurs in the expression for the eigenvalues of~\eqref{eq:flatBack2DPullback},
        \begin{equation}\label{eq:flatBackPullbackEigen}
        	E_\pm\ =\ \tfrac12\big(\zeta^2\pm|\zeta|D_{\rm R}\big)
            \ewith
            D_{\rm R}^2\ \coloneqq\ 4(\partial_x\partial_y\psi)^2+\big(\partial_x^2\psi-\partial_y^2\psi\big)^2\,.
        \end{equation}
        Note that $D_{\rm R}^2=\zeta^2-4f$, so the eigenvalues take the same sign for $f>0$ and opposite sign for $f<0$, provided they are both non-zero. Finally, the curvature scalars~\eqref{eq:curvature2D} reduce to
        \begin{subequations}\label{eq:flatBackCurvature2D}
            \begin{equation}
                R\ =\ \frac1\zeta\left\{\tilde R-\tfrac{1}{\sqrt{|\det(\tilde g)|}}\partial_i\big[\sqrt{|\det(\tilde g)|}\,\tilde g^{ij}\partial_j\log(|\zeta|)\big]\right\},
            \end{equation}
    		where
            \begin{equation}
                \tilde g^{-1}\ =\ \frac1f
                \begin{pmatrix}
                    \partial_y^2\psi & -\partial_x\partial_y\psi 
                    \\
                    -\partial_x\partial_y\psi & \partial_x^2\psi
                \end{pmatrix}
            \end{equation}
            and
            \begin{equation}
                \tilde R\ =\ -\tfrac14\tilde g^{ij}\tilde g^{kl}\tilde g^{mn}(\partial_i\partial_j\partial_m\psi\,\partial_k\partial_l\partial_n\psi-\partial_i\partial_k\partial_m\psi\,\partial_j\partial_l\partial_n\psi)~.
            \end{equation}
        \end{subequations}
        As shown in~\eqref{eq:flatBack2DPullback}, the pull-back metric can be considered a function of vorticity and rate-of-strain, and the curvature of that metric therefore involves derivatives of these quantities. In turbulent flows, fine-scale structure (such as vortex filaments) could imply large gradients of vorticity and rate-of-strain, which in turn could present challenges in calculating such gradients in numerical simulations. However, as we shall illustrate in the following section, when the metric structure degenerates and/or the scalar curvature becomes singular, then these geometric features are associated with topological changes in the fluid flow.

        \paragraph{Larchev{\^e}que's criterion and uniform vorticity and strain.}
        In~\cite{Larcheveque:1990aa,Larcheveque:1993aa} it is noted that the stream function is uniquely defined on a simply connected domain $\Sigma$ bounded by a closed streamline when $\lap p>0$ and $\psi|_{\partial\Sigma}$ is known. For example, consider the stream function 
        \begin{equation}\label{eq:streamFunctionLarcheveque}
            \psi(t,x,y)\ \coloneqq\ \tfrac12\big[a(t)x^2+b(t)y^2\big]\,,
        \end{equation} 
        where $a$ and $b$ are functions of time $t$ alone. The Laplacian of the pressure is $\lap p=2f=2ab$, hence when $a$ and $b$ have the same sign, $f>0$, vorticity dominates, and the metric $\hat g$ given by~\eqref{eq:flatBack2DMetric} is Riemannian. Similarly, when $a$ and $b$ have different signs, $f<0$, strain dominates, and the metric $\hat g$ is Kleinian. Additionally, the metric is globally singular when $a$ or $b$ vanish, that is, when $f=0$.
        
        The vorticity is simply $\zeta=a+b$ and the pull-back metric~\eqref{eq:flatBack2DPullback} becomes
        \begin{equation}\label{eq:lrpbLarcheveque}
            g\ =\ (a+b)
            \begin{pmatrix}
                a & 0
                \\
                0 & b
            \end{pmatrix}.
        \end{equation}
        Like $\hat g$, its pull-back is Riemannian when $f>0$, Kleinian when $f<0$, and singular when $f=0$, with the following exception: the pull-back metric is also singular when the vorticity vanishes, that is, where $a=-b$. In line with Larchev{\^e}que, this additional singularity falls outside of Riemannian regions, hence the sign of vorticity remains constant where $f>0$.
		
        \paragraph{Flow with bifurcations.}
        In the following example, discussed in~\cite{Moffatt:2001aa} in connection with topological fluid dynamics, the occurrence of singularities in the Monge--Amp{\`e}re geometry can be associated with important features, such as bifurcations, in the fluid flow.

        \begin{figure}[ht]
            \vspace{15pt}
            \begin{center}
                \includegraphics[scale=0.201]{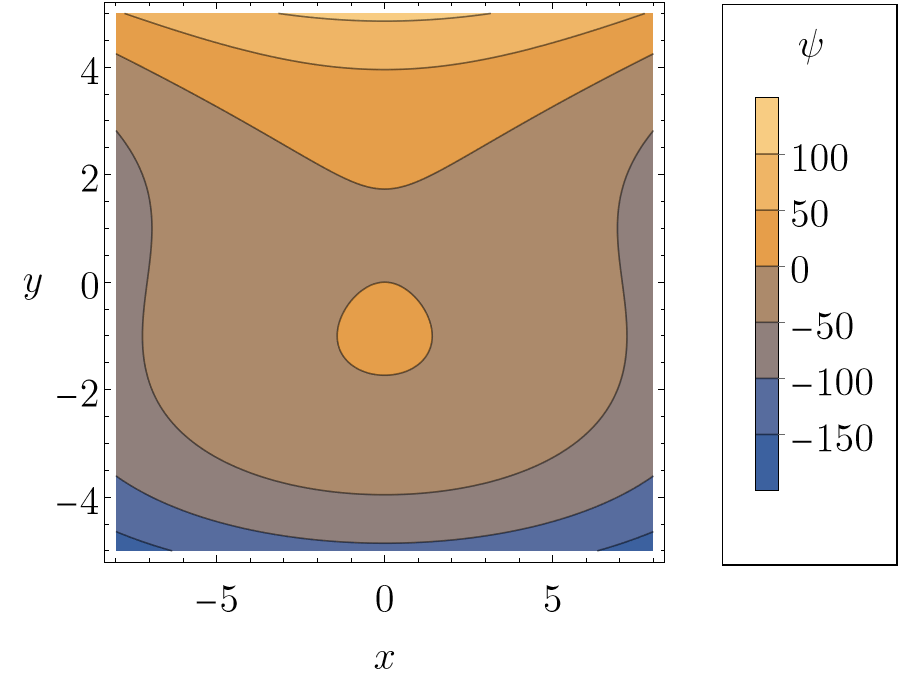}
                \caption{Plot of the streamlines for stream function~\eqref{eq:streamFunctionMoffat} with $t=-1$. The streamlines around the elliptic point $(0,-1)$ form closed contours, whilst those near the hyperbolic point $(0,1)$ diverge.}
                \label{fig:streamLinesMoffatt}
            \end{center}
        \end{figure}

        In particular, following~\cite{Moffatt:2001aa}, we consider the stream function
        \begin{equation}\label{eq:streamFunctionMoffat}
            \psi(t,x,y)\ \coloneqq\ -x^2+3yt+y^3~,
        \end{equation}  
        shown in \cref{fig:streamLinesMoffatt}, where time, $t$, is a parameter. Using~\eqref{eq:streamFunctionMoffat}, the Poisson equation for the pressure takes the explicit form
        \begin{equation}\label{eq:MAMoffatt}
            \partial_x^2\,\psi\partial_y^2\psi-(\partial_x\partial_y\psi)^2\ =\ -12y\ =\ f~,
        \end{equation}
        which we can view as a variable-coefficient Monge--Amp{\`e}re equation for $\psi$. This equation is elliptic when $y<0$, hyperbolic when $y>0$, and it degenerates to parabolic form on $y=0$. 

        The equations for the flow are
        \begin{equation}\label{eq:flowMoffat} 
            \dot x\ =\ -\partial_y\psi\ =\ -3\big(t+y^2\big)
            \eand
            \dot y\ =\ \partial_x\psi\ =\ -2x~,
        \end{equation}
        where the superposed dot indicates the derivative with respect to $t$. When $t>0$ there are no real fixed points\footnote{The stability matrix $A^i{}_j$ with $\delta\dot x^i=A^i{}_j\delta x^j$ is the velocity-gradient matrix $A^i{}_j\coloneqq\partial_jv^i$ which, upon recalling~\eqref{eq:localVelocity}, is related to the Hessian of the stream function by means of $A^i{}_j=\eps^{ik}\partial_j\partial_k\psi$. It has eigenvalues $\pm\sqrt{-f}$ which are thus purely imaginary when $f>0$.} (defined by $\dot x=\dot y=0$) and when $t=0$ the level set $\psi=0$ has a cusp singularity at the origin. When $t<0$, the fixed points are located at $(0,\pm\sqrt{-t})$; the fixed point at $(0,\sqrt{-t})$ is hyperbolic, hence streamlines in a neighbourhood of this fixed point tend to diverge, whilst the fixed point at $(0,-\sqrt{-t})$ is elliptic, indicating that the flow in a neighbourhood around this fixed point tends to swirl. This shows that for values of $t$ at which there are fixed points, the elliptic fixed point lies in the region where vorticity dominates, whilst the hyperbolic fixed point resides where strain dominates. Note that $f$ is time-independent, so these regions remain coherent in time, regardless of the fixed points.

        In terms of the Monge--Amp{\`e}re geometry introduced thus far, the metric $\hat g$ on $T^*\IR^2$ is given by~\eqref{eq:flatBack2DMetric}. The corresponding curvature scalar~\eqref{eq:flatBackCurvatureLR2D} is given by
        \begin{equation}\label{eq:MoffattCurvatureLR}
            \hat R\ =\ -\frac{1}{12y^3}~.
        \end{equation}
        Note that $f$ vanishes at $y=0$ and hence the metric $\hat g$ is singular. Furthermore, the signs of $\hat R$ and $f$ coincide. More generally, for $y<0$, we have $f>0$ and vorticity dominates in this region, where the metric~\eqref{eq:flatBack2DMetric} is Riemannian with positive curvature scalar. For $y>0$, it follows that $f<0$ and strain dominates, with~\eqref{eq:flatBack2DMetric} becoming Kleinian with negative curvature scalar. Furthermore, the vorticity is $\zeta=2(3y-1)$, hence the pull-back metric~\eqref{eq:flatBack2DPullback} is
        \begin{equation}\label{eq:pullbackMetricMoffatt}
            g\ =\ 4(1-3y)
            \begin{pmatrix}
                1 & 0
                \\
                0 & -3y
            \end{pmatrix}.
        \end{equation}
        Evidently, the metric is Riemannian for $y<0$ and singular when $y=0$ or $y=\frac13$. The former singularity corresponds to where $f=0$, with the latter occurring precisely where the vorticity vanishes. Using~\eqref{eq:flatBack2DPullback}, the components of strain are given by $S_{xx}=S_{yy}=0$ and $S_{xy}=-3y-1$, describing shearing at an angle $\frac{\pi}{4}$ to the coordinate axes, near hyperbolic fixed points, in regions where strain dominates. The corresponding curvature scalars~\eqref{eq:flatBackCurvature2D} are
        \begin{equation}\label{eq:curvatureMoffatt2D}
            R\ =\ \frac{1-9y}{8y^2(1-3y)^3}
            \eand
            \tilde R\ =\ 0~.
        \end{equation}
        This, in turn, shows that the metric singularities $y=0$ and $y=\frac13$ are, in fact, singularities of the scalar curvature, which is invariant under changes of coordinates on $M$.

        The picture emerging here has some interesting features. Recall the definitions of classical and generalised solutions from \cref{sec:geometry2dFlows}. Then, commencing with~\eqref{eq:streamFunctionMoffat}, we note that $L$ is a classical solution. However, as just indicated, the metric and curvature of $L$ have singularities, which are related to the points at which the flow changes from elliptic to hyperbolic, and where vorticity vanishes. We shall show next that we can describe this singular behaviour in terms of a generalised solution to the \uline{Legendre-dual problem}.
         
        \paragraph{Legendre duals.}   
        Consider a domain $\Sigma$ of the flow, with $L_\Sigma\coloneqq\rmd\psi(\Sigma)=\{(x,y;\partial_x\psi,\partial_y\psi)\,|\,(x,y)\in\Sigma\}$ the corresponding region in the Lagrangian submanifold $L$ described locally by $\rmd\psi$. Then, locally on this domain, $\pi|_L$ is the identity and is hence non-singular. In~\cite{Sewell:1987aa} it is shown that
        \begin{equation}\label{eq:legendreLocalInversion}
      		x'(t)\ =\ \parder[\psi(t,x,y)]{x}\ =\ v
      		\eand
      		y'(t)\ =\ \parder[\psi(t,x,y)]{y}\ =\ -u~.
      	\end{equation}
      	is a local inversion and one can define the \uline{Legendre transformation}~\cite{Sewell:1994aa}
      	\begin{equation}\label{eq:legendreStream}
      		\psi'(t,x',y')\ \coloneqq\ x'x+y'y-\psi(t,x,y)
      	\end{equation}
      	when $f\not\in\{0,\infty\}$, where finiteness of $f$ follows from the non-singular nature of the projection $\pi|_L$. Here, $x'$ and $y'$ are the local coordinates on the Legendre-dual space and $\psi'$ is known as the \uline{Legendre-dual (stream-)function}. Furthermore, in this setting we may also define the map $\tilde\pi|_L:(x,y)\mapsto(\partial_x\psi,\partial_y\psi)=(x',y')$,\footnote{Note that $\tilde \pi|_L=\tilde\pi\circ\iota$ with $\tilde\pi:(x,y,p,q)\mapsto(p,q)$ defined on $T^*\IR^2$. Furthermore, we use that $\pi|_L$ is locally a diffeomorphism to coordinatise $L$ by $(x,y)$ from $M$.} with determinant $f$. Hence, $f\neq 0$ precisely when $\tilde \pi|_L$ is non-singular, corresponding to when~\eqref{eq:legendreLocalInversion} is a local inversion. 
        
        It follows that the map $\tilde\pi|_L$ is singular precisely when $f$ vanishes and~\eqref{eq:legendreLocalInversion} is not invertible, in which case $\psi$ and $\psi'$ have different regularity. In particular, if $\psi$ is a classical solution to~\eqref{eq:pressureLaplacian2d}, the Legendre-dual $\psi'$ generates a generalised solution to the dual Monge--Amp{\`e}re equation, with singular behaviour where $f=0$. The dual Monge--Amp{\`e}re equation is given by
      	\begin{equation}\label{eq:MongeLegendre}
      		f(t,x(x',y'),y(x',y'))\ =\ \frac{1}{\det(\sfHess(\psi'(t,x',y')))}~.
      	\end{equation}
        Consequently, vanishing $f$ corresponds to $\det(\sfHess(\psi'(t,x',y')))$ blowing up. As $f$ is finite, it follows that $\det(\sfHess(\psi'(x',y')))\neq 0$. The Legendrian dual to the Hessian part of the metric~\eqref{eq:fluidMetric2dPullBack} is
      	\begin{equation}\label{eq:hessMetLegendre}
            \tilde g'\ =\ \tfrac12\psi'_{ij}\rmd x'^i\odot\rmd x'^j~,
        \end{equation}
      	where $(x'^i)=(x',y')$ and $\det(\tilde g')=\frac1f$. It follows that $\tilde g'$ is non-degenerate, however, it does blow up when $f=0$. The accompanying vorticity conformal factor $\zeta$ has the Legendrian dual
      	\begin{equation}
      		\zeta'\ =\ f\lap'\psi'~,
      	\end{equation}
      	with $\lap'$ the Laplacian with respect to $(x',y')$. From this, it follows that $g'=\zeta'\tilde g'$ may in-fact be singular when $f=0$ or $\lap'\psi'=0$. The curvature scalar associated to $g'$ is given by~\eqref{eq:flatBackCurvature2D}, with objects replaced by their primed Legendre dual as appropriate and 
        \begin{equation}\label{eq:lagMetLegendre}
            g'^{-1}\ =\ \frac{f}{(\lap'\psi')^2}
            \begin{pmatrix}
                \partial_{y'}^2\psi' & -\partial_{x'}\partial_{y'}\psi' 
                \\
                -\partial_{x'}\partial_{y'}\psi' & \partial_{x'}^2\psi'
            \end{pmatrix}.
        \end{equation} 
        Thus, singularities of $\tilde g$ in the $(x,y)$ coordinates do not occur in $\tilde g'$ in the $(x',y')$ coordinates and are instead transferred to the projection $\tilde\pi|_L$ and the dual solution $\psi'$, via the Legendre transformation. By restricting our domain $\Sigma$ such that $f$ has constant sign, we can impose that the Legendre transformation is well defined and both $\tilde\pi|_L$ and $\tilde g$ are non-singular. 

        \begin{figure}[ht]
            \vspace{15pt}
            \begin{center}
                \begin{subfigure}[b]{0.425\textwidth}
                    \includegraphics[scale=0.201]{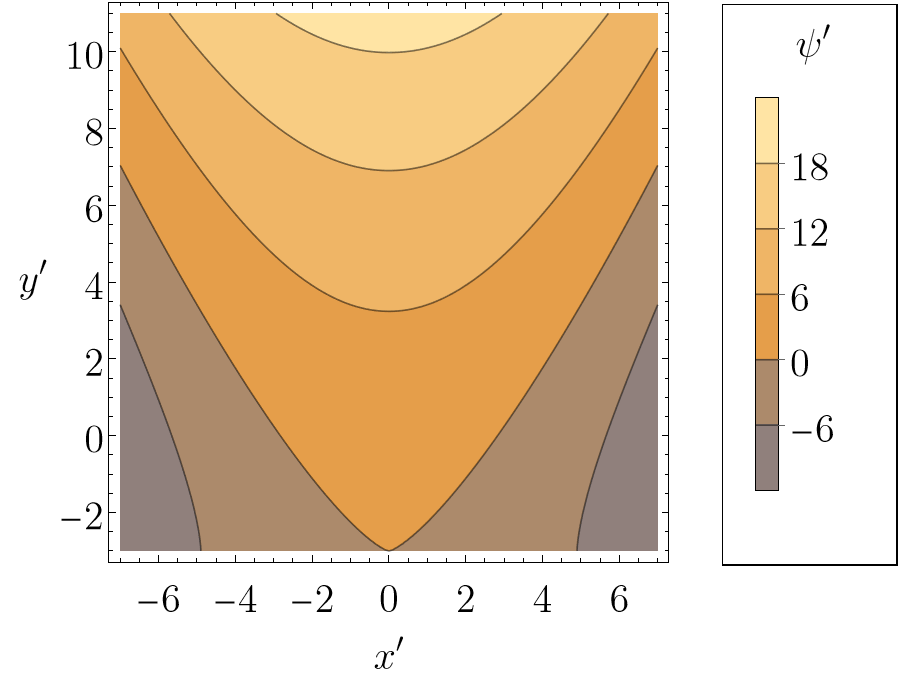}
                    \caption{Contour plot for the upper sheet of the function~\eqref{eq:moffattDualStream}. There is a hyperbolic point at $(0,3t)$.}
                \end{subfigure}
                \hspace{25pt}
                \begin{subfigure}[b]{0.425\textwidth}
                    \includegraphics[scale=0.201]{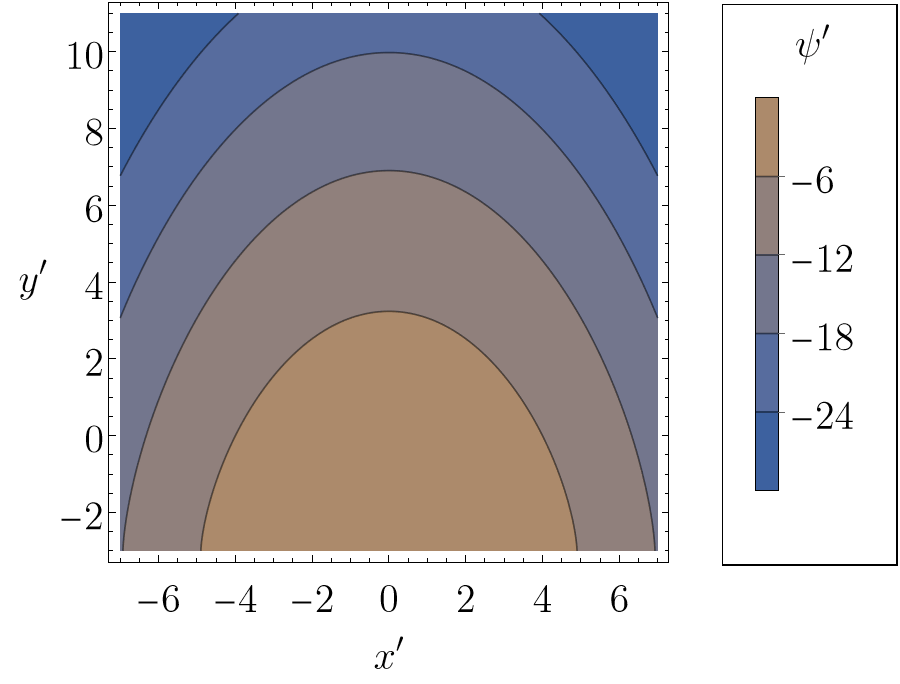}
                    \caption{Contour plot for the lower sheet of the function~\eqref{eq:moffattDualStream}. There is an elliptic point at $(0,3t)$.}
                \end{subfigure}
                \\[15pt]
                \begin{subfigure}[b]{\textwidth}
                    \begin{center}
                        \includegraphics[scale=0.32]{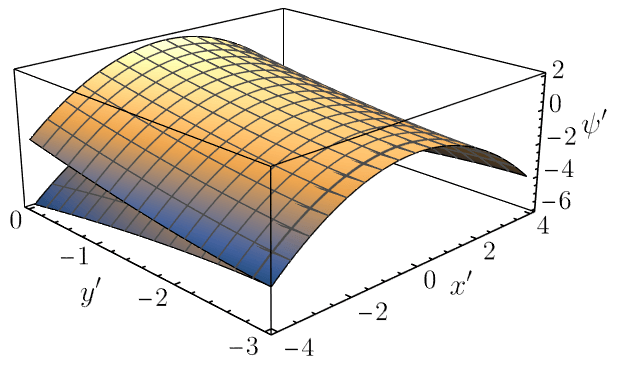}
                    \end{center}
                    \caption{A plot of the two sheets of the function~\eqref{eq:moffattDualStream}, with a fold singularity appearing along the line $y'=-3$, where the two sheets meet.}
                \end{subfigure}
                \caption{A selection of plots of the Legendre-dual stream function~\eqref{eq:moffattDualStream}, at time $t=-1$. The multivalued behaviour of $\psi'$ is associated with the corresponding Lagrangian submanifold $\iota:L\hookrightarrow T^*\IR^2$ being a generalised solution to~\eqref{eq:MongeLegendre}.}
                \label{fig:moffattDualStream}
            \end{center}
        \end{figure}  

        \paragraph{Flow with bifurcations and Legendre duality.}
        Returning to the stream function~\eqref{eq:streamFunctionMoffat}, we obtain for~\eqref{eq:legendreStream}
		\begin{subequations}
			\begin{equation}\label{eq:moffattDualStream}
    	   		\psi'(t,x',y')\ =\ -\tfrac14x'^2\pm\tfrac{2}{3\sqrt{3}}(y'-3t)^{\frac32}~,
        	\end{equation}
        	with 
        	\begin{equation}
        		x'\ =\ \dot y\ =\ -2x
        		\eand 
        		y'\ =\ -\dot x\ =\ 3y^2 + 3t~,
        	\end{equation}
        	and
        	\begin{equation}
        		x\ =\ -\tfrac12x'\ =\ \parder[\psi']{x'}
        		\eand
        		y\ =\ \pm\sqrt{\tfrac13(y'-3t)}\ =\ \parder[\psi']{y'}~,
        	\end{equation}
		\end{subequations}
        respectively. As $\psi$ is a classical solution to our Monge--Amp{\`e}re equation, it follows that $\pi|_L$ is the identity on $\IR^2$. However, $\psi'$ is only defined for $y'-3t\geq 0$ and is multivalued on momentum space, except where $y'-3t=0$; plots of the two sheets are shown in~\cref{fig:moffattDualStream}.
		
		As the Jacobian of the projection $\tilde\pi|_L:(x,y)\mapsto(\partial_x\psi,\partial_y\psi)$ is precisely the Hessian, it follows that both $\tilde\pi|_L$ and the local inversion~\eqref{eq:legendreLocalInversion} are singular where the Hessian is degenerate, that is, where $f = 0 =y'-3t$. Restricting to a domain $\Sigma$ on which the sign of $f$ is constant then amounts to choosing a sheet of $\psi'$.
		
        \begin{figure}[ht]
            \vspace{15pt}
            \begin{center}
                \begin{subfigure}[b]{0.425\textwidth}
                    \includegraphics[scale=0.201]{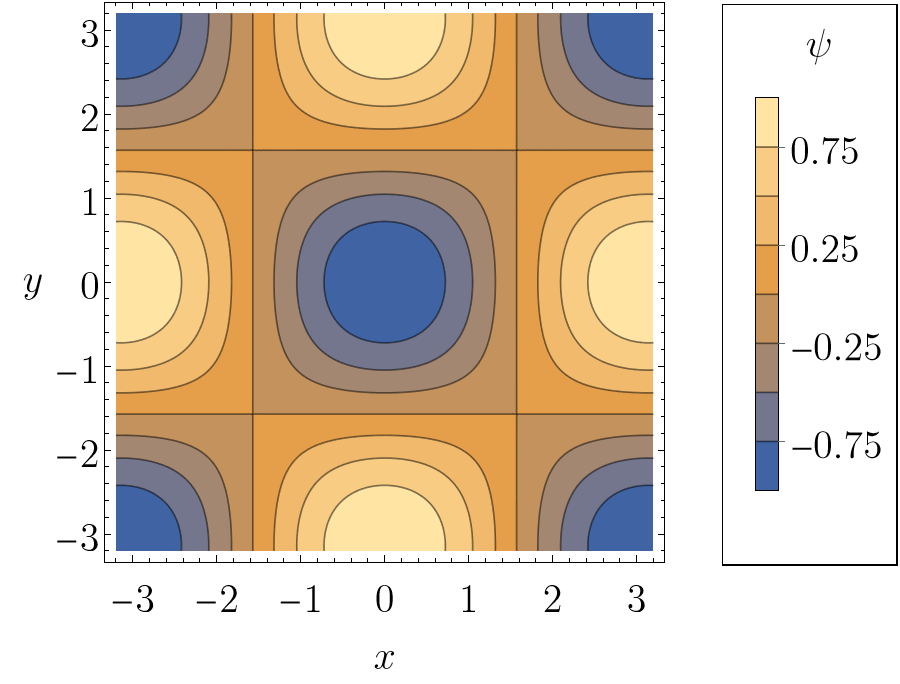}
                    \caption{Plot of the streamlines of the stream function~\eqref{eq:streamFunctionTGV}. The domain is partitioned the domain into squares of side length $\pi$, across which the sign of the stream function alternates.}
                    \label{fig:streamFunctionTGV}                        
                \end{subfigure}
                \hspace{25pt}
                \begin{subfigure}[b]{0.425\textwidth}
                    \includegraphics[scale=0.201]{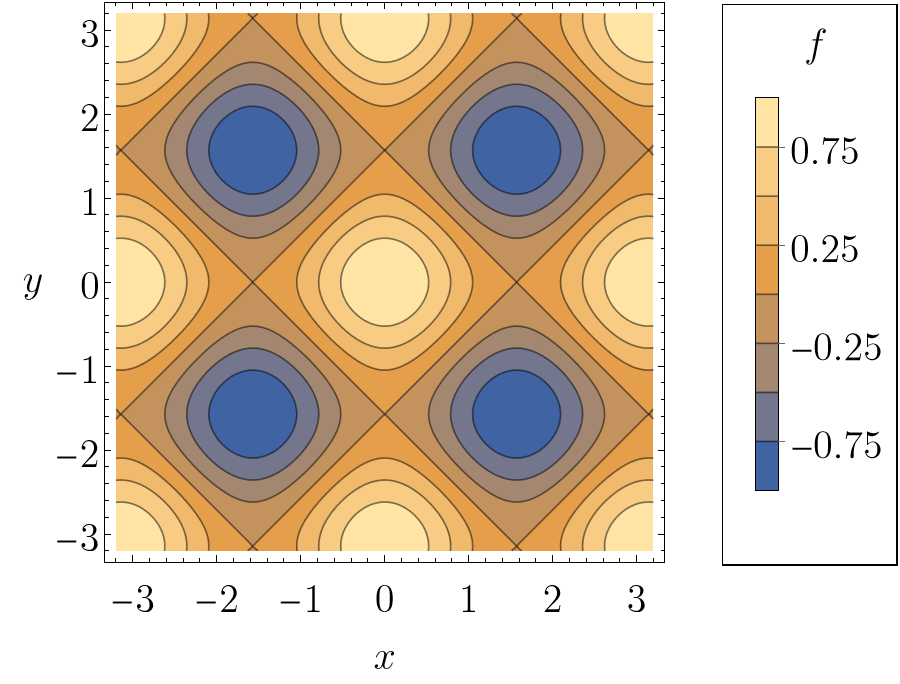}
                    \caption{Contour plot of $f=\frac12a^2b^2F^2[\cos(2ax)\linebreak[4]+\cos(2by)]$. The domain is partitioned into rhombi, with positive/negative regions around elliptic/hyperbolic fixed points.}   
                \end{subfigure}
                \caption{Plots of the iso-lines of the stream function and half the Laplacian of pressure for the Taylor--Green vortex with parameters $a=b=1$ and $F(t)\equiv 1$, which shall be used for the remainder of the plots for this example. Streamlines corresponding to values of sufficiently large magnitude are closed contours contained in regions of positive $f$, where vorticity dominates. The vorticity is proportional to the stream function, $\zeta=-(a^2+b^2)\psi$.}
            \end{center}
            \vspace{-10pt}
        \end{figure}

        \paragraph{Taylor--Green vortex.}
        In two dimensions, the stream function of the \uline{Taylor--Green vortex}~\cite{Taylor:1937aa} takes the form 
        \begin{equation}\label{eq:streamFunctionTGV}
            \psi(t,x,y)\ \coloneqq\ -F(t)\cos(ax)\cos(by)
        \end{equation}
        where $F$ is a function of time $t$ alone and $a,b\in\IR$ are some parameters. See \cref{fig:streamFunctionTGV}.

        Hence, for~\eqref{eq:streamFunctionTGV}, we have $f=\frac12a^2b^2F^2[\cos(2ax)+\cos(2by)]$, so the metric is again~\eqref{eq:flatBack2DMetric}, and the curvature scalar~\eqref{eq:flatBackCurvatureLR2D} is simply given by
        \begin{equation}\label{eq:curvatureLR_TGV}
            \hat R\ =\ \frac{8(a^2+b^2)[1+\cos(2ax)\cos(2by)]}{a^2b^2F^2[\cos(2ax)+\cos(2by)]^3}~.
        \end{equation}
        See \cref{fig:curvature4DTGV}. Consequently, when $\cos(2ax)+\cos(2by)>0$, the metric is Riemannian with a positive curvature scalar and vorticity dominates. When $\cos(2ax)+\cos(2by)<0$ the metric is Kleinian with negative curvature scalar, and strain dominates. The signs of $f$ and $\hat R$ coincide. Both the metric and curvature scalar are singular when $abF=0$ and along the lines $y=\frac{a}{b}x+\frac{\pi}{2b}(2n+1)$ for all $n\in\IZ$ (when $\cos(2ax)+\cos(2by)=0$), corresponding to where $f=0$. 
        
        Furthermore, the vorticity is given by $\zeta=(a^2+b^2)F\cos(ax)\cos(by)$ so that the pull-back metric~\eqref{eq:flatBack2DPullback} becomes
        \begin{equation}\label{eq:pullbackMetricTGV}
            g\ =\ \frac{(a^2+b^2)F^2}{4}
            \begin{pmatrix}
                a^2[1+\cos(2ax)][1+\cos(2by)] & -ab\sin(2ax)\sin(2by)
                \\
                -ab\sin(2ax)\sin(2by) & b^2[1+\cos(2ax)][1+\cos(2by)]
            \end{pmatrix}.
        \end{equation}
        Its eigenvalues~\eqref{eq:flatBackPullbackEigen} are
        \begin{subequations}\label{eq:eigenvaluesTGV}
            \begin{equation}
                E_\pm\ =\ \tfrac{F^{2}(a^2+b^2)}{4}\left[2\big(a^2+b^2\big)\cos^2(ax)\cos^2(by)\pm|\cos(ax)\cos(by)|\sqrt{\tilde E}\,\right]
            \end{equation}
            with
            \begin{equation}
                \tilde E\ \coloneqq\ \big(a^4-6a^2b^2+b^4\big)[\cos(2ax)+\cos(2by)]+\big(a^2+b^2\big)^2[1+\cos(2ax)\cos(2by)]~.
            \end{equation}
        \end{subequations}
        See \cref{fig:eigenvaluesTGV}. The corresponding curvature scalars~\eqref{eq:flatBackCurvature2D} are
        \begin{equation}\label{eq:curvaturePullbackTGV}
            R\ =\ \frac{8}{F^2(a^2+b^2)[\cos(2ax)+\cos(2by)]^2}
            \eand
            \tilde R\ =\ 0~.
        \end{equation}
        Evidently, $E_+$ is everywhere non-negative, so the signature of the metric~\eqref{eq:pullbackMetricTGV} is determined by the sign of $E_-$. It is clear from~\cref{fig:eigenvalueMinusTGV} that, when $\cos(2ax)+\cos(2by)\gtrless 0$, we have $E_-\gtrless 0$ and the metric $g$ is Riemannian/Kleinian with vorticity/strain dominating. Also, $E_-=0$ when $\cos(2ax)+\cos(2by)=0$ and we note from~\eqref{eq:curvaturePullbackTGV} that the scalar curvature is singular at these points too. The vorticity changes sign as the contours $x=\frac{\pi}{2a}(2n+1)$ or $y=\frac{\pi}{2b}(2n+1)$ are crossed, and the metric~\eqref{eq:pullbackMetricTGV} is Kleinian on both sides; this is consistent with the observations that the vorticity has constant sign in Riemannian regions, where it also dominates~\cite{Larcheveque:1993aa}. 

        \begin{figure}[ht]
            \vspace{15pt}
            \begin{center}
                \begin{subfigure}[b]{0.425\textwidth}
                    \includegraphics[scale=0.201]{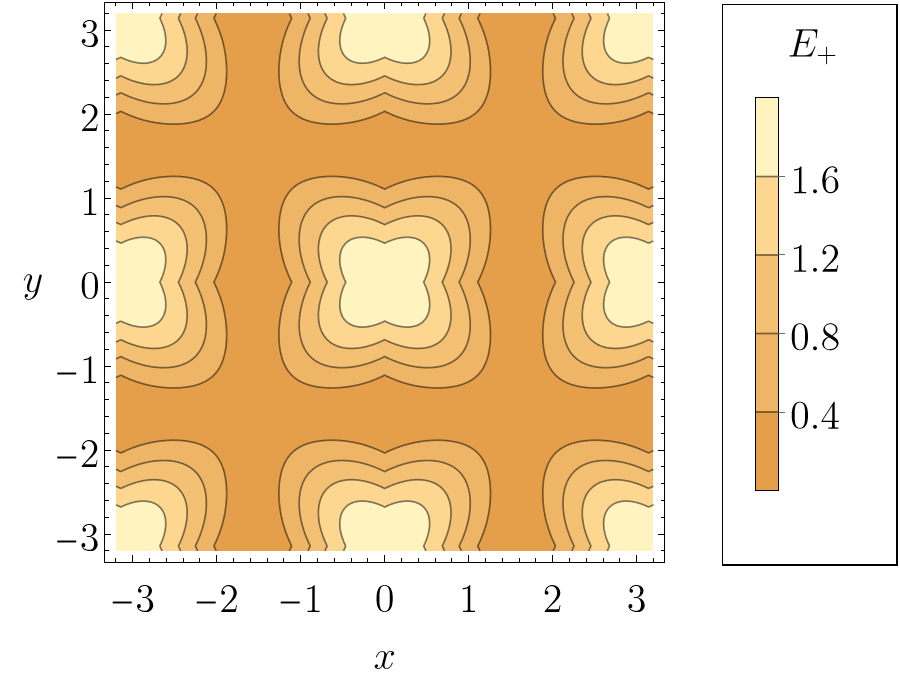}
                    \caption{Contour plot for the eigenvalue $E_+$. This eigenvalue is non-negative on the domain and vanishes along $x=\frac\pi2(2n+1)$ or $y=\frac\pi2(2n+1)$ for all $n\in\IZ$.}
                \end{subfigure}
                \hspace{25pt}
                \begin{subfigure}[b]{0.425\textwidth}
                    \includegraphics[scale=0.201]{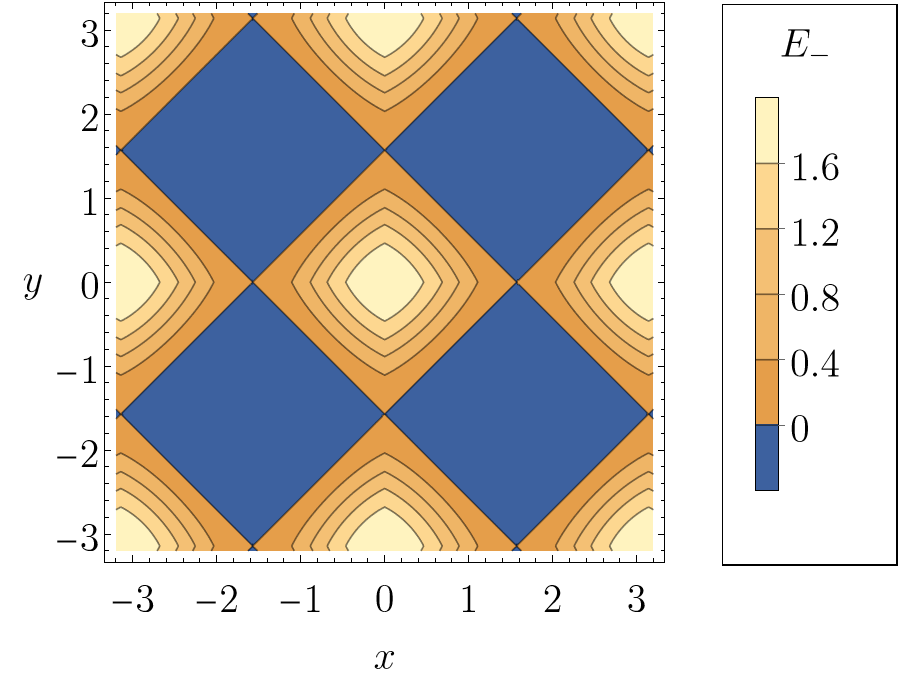}
                    \caption{Contour plot for the eigenvalue $E_-$. This eigenvalue is non-positive within the dark blue regions, but also vanishes at points within these domains.}
                    \label{fig:eigenvalueMinusTGV}
                \end{subfigure}
                \caption{Plots of the eigenvalues~\eqref{eq:eigenvaluesTGV} of the pull-back metric~\eqref{eq:pullbackMetricTGV} for the Taylor--Green vortex with parameters $a=b=1$ and $F(t)\equiv1$.}
                \label{fig:eigenvaluesTGV}
            \end{center}
            \vspace{-15pt}
        \end{figure}

        \begin{figure}[ht]
            \begin{center}
                \begin{subfigure}[b]{0.425\textwidth}
                    \includegraphics[scale=0.201]{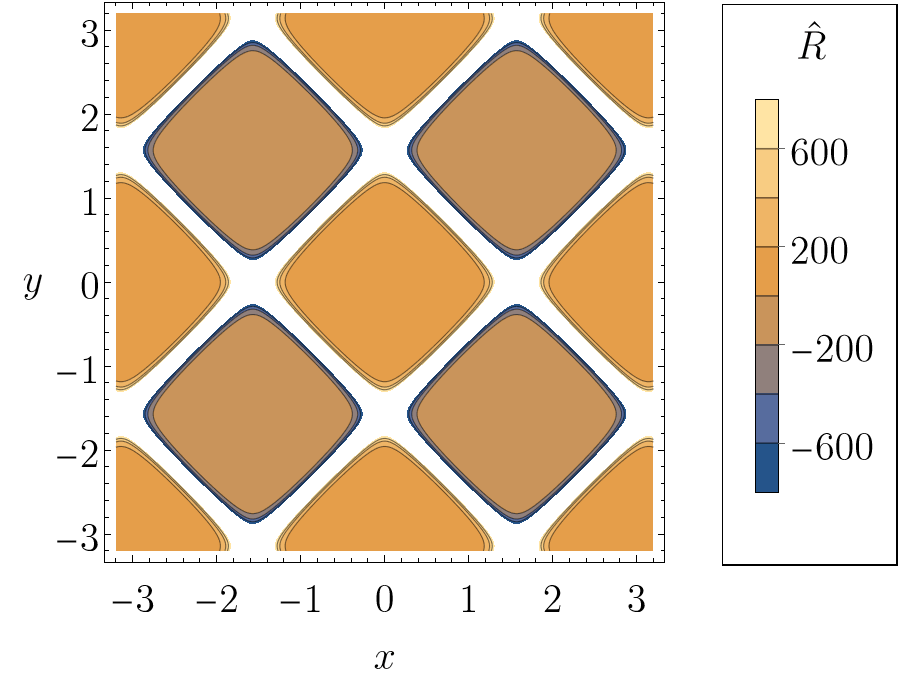}
                    \caption{Contour plot for the curvature scalar $\hat R$ on $T^*\IR^2$. This is singular at $y=x+\frac\pi2(2n+1)$ for all $n\in\IZ$, which is where $f$ vanishes.}
                    \label{fig:curvature4DTGV}
                \end{subfigure}
                \hspace{25pt}
                \begin{subfigure}[b]{0.425\textwidth}
                    \includegraphics[scale=0.201]{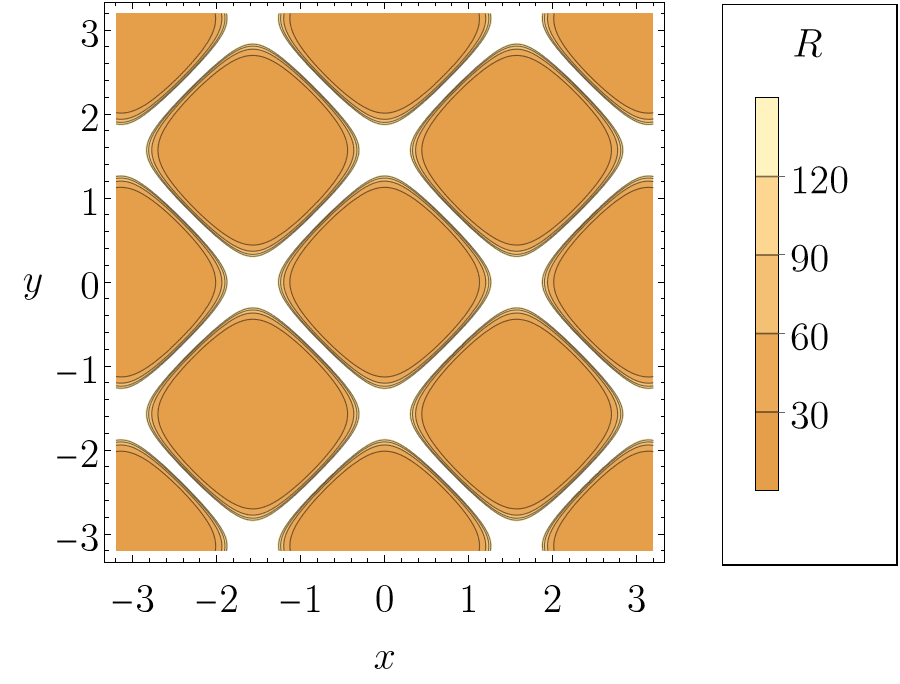}
                    \caption{Contour plot for the curvature scalar $R$. This is everywhere positive yet small away from the singularities given by $y=x+\frac{\pi}{2}(2n+1)$ for all $n\in\IZ$.}
                \end{subfigure}
                \caption{Contour plots of the curvatures~\eqref{eq:curvatureLR_TGV} and~\eqref{eq:curvaturePullbackTGV} respectively, for the Taylor Green vortex with parameters $a=b=1$ and $F(t)\equiv 1$.}
            \end{center}
            \vspace{-25pt}
        \end{figure}

        \section{Geometric properties of fluids in three dimensions}\label{sec:geometricProperties3d}

        Having discussed two-dimensional fluid flows, we now consider flows in three (or even higher) dimensions. Here, the situation is much more involved, since in the general case the flow is not described by a stream function.\footnote{We shall consider some examples of three-dimensional flows with symmetries, which can be described in terms of a stream function, such as Hill's spherical vortex.} Hence, the pressure equation~\eqref{eq:pressureLaplacian} cannot be converted into a Monge--Amp{\`e}re equation. Nevertheless, one can formulate the flow in terms of differential forms, as we shall now explain. To do this, we first revisit the formulation we have just used for two-dimensional flows, and note that an alternative description naturally presents itself. Whilst this alternative view makes little difference to the geometric picture in two dimensions, we show that it provides a mechanism to unify what could otherwise be quite different geometric descriptions of two-dimensional and three-dimensional flows, as was described in~\cite{Banos:2015aa}.

        \subsection{Two-dimensional case revisited}\label{sec:2dr}

        \paragraph{Monge--Amp{\`e}re structure.}
        In \cref{sec:geometry2dFlows}, we have seen that the Monge--Amp{\`e}re structure~\eqref{eq:MongeAmpereStructure2d} encodes incompressible fluids on a two-dimensional Riemannian manifold $(M,\bg{g})$. As before, let $(x^i,q_i)$ be local coordinates on $T^*M$. Instead of using the standard symplectic structure~\eqref{eq:MongeAmpereStructure2d} on $T^*M$, we now propose taking the `dual' form
        \begin{equation}\label{eq:alternativeSymplecticStructure}
            \varpi\ \coloneqq\ \bg{\nabla}q_i\wedge{\star_{\bg{g}}}\rmd x^i~,
        \end{equation}
        that is,~\eqref{eq:almostKaehlerForm2d} without the pre-factor. Evidently, $\varpi$ is non-degenerate\footnote{We have already seen this in our discussion around~\eqref{eq:almostKaehlerForm2d}.} and it is also closed as a consequence of
        \begin{subequations}\label{eq:closureOmega}
            \begin{equation}\label{eq:dNablaq}
                \rmd(\bg{\nabla}q_i)\ =\ \tfrac12\rmd x^l\wedge\rmd x^k\bg{R}_{kli}{}^jq_j+\rmd x^j\bg{\Gamma}_{ji}{}^k\wedge\bg{\nabla}q_k
            \end{equation}
            and
            \begin{equation}\label{eq:dstardx}
                \rmd{\star_{\bg{g}}}\rmd x^i\ =\ -\bg{g}^{jk}\bg{\Gamma}_{jk}{}^i\,\vol{M}
                \ewith
                \vol{M}\ \coloneqq\ \tfrac{\sqrt{\det(\bg{g})}}{2}\eps_{ij}\rmd x^i\wedge\rmd x^j~.
            \end{equation}
        \end{subequations}
        Hence,~\eqref{eq:alternativeSymplecticStructure} defines a symplectic structure. It is then easily seen that, when $\iota:L\hookrightarrow T^*M$ is given by
        \begin{equation}\label{eq:incompressibilityFromLagrangian}
            \iota\,:\,x^i\ \mapsto\ (x^i,q_i)\ \coloneqq\ (x^i,v_i(x))~,
        \end{equation}
        where $v_i=v_i(x)$ are the components of the velocity (co-)vector field, the condition $\iota^*\varpi=0$ is equivalent to requiring the divergence-free constraint~\eqref{eq:generalNavierStokesIncompressibilityLocal}. Thus, we again obtain a Lagrangian submanifold $L$ of $T^*M$, and this time $L$ encodes the divergence-free constraint. 

        Moreover, using ${\star_{\bg{g}}}(\rmd x^i\wedge\rmd x^j)=\sqrt{\det(\bg{g})}\eps^{ij}$ and the volume form~\eqref{eq:dstardx} on $M$, we may rewrite the Monge--Amp{\`e}re form $\alpha$ defined in~\eqref{eq:MongeAmpereStructure2d} as
        \begin{equation}\label{eq:alternativeMAForm}
            \alpha\ =\ \tfrac12\bg{\nabla}q_i\wedge\bg{\nabla}q_j\wedge{\star_{\bg{g}}}(\rmd x^i\wedge\rmd x^j)-\hat f\,\vol{M}~.
        \end{equation}
        Again, using~\eqref{eq:dNablaq} together with
        \begin{subequations}\label{eq:closureAlpha}
            \begin{equation}
                \rmd{\star_{\bg{g}}(\rmd x^i\wedge\rmd x^j)}\ =\ 2\bg{g}^{kl}\bg{\Gamma}_{kl}{}^{[i}{\star_{\bg{g}}\rmd x^{j]}}+2\bg{g}^{k[i}\bg{\Gamma}_{kl}{}^{j]}{\star_{\bg{g}}\rmd x^l}
            \end{equation}
            and
            \begin{equation}
                \rmd x^k\wedge{\star_{\bg{g}}(\rmd x^i\wedge\rmd x^j)}\ =\ -2\bg{g}^{k[i}{\star_{\bg{g}}\rmd x^{j]}}
            \end{equation}
        \end{subequations}
        it is not too difficult to see that $\alpha$ is closed. In addition, the requirement that the pull-back of $\alpha$ under~\eqref{eq:incompressibilityFromLagrangian} vanishes is directly equivalent to the pressure equation~\eqref{eq:pressureLaplacian}, provided that we simultaneously demand that $\iota^*\varpi=0$. Notice that we also have $\alpha\wedge\varpi=0$ so that the pair $(\varpi,\alpha)$ is again a Monge--Amp{\`e}re structure. 

        \paragraph{Almost (para-)Hermitian structure.}
        We may now follow our discussion in \cref{sec:geometry2dFlows} and define an endomorphism $\hat\caJ$ of the tangent bundle of $T^*M$ by
        \begin{equation}\label{eq:alternativeDefinitionJ2d}
            \frac{\alpha}{\sqrt{|\hat f|}}\ \eqqcolon\ \hat\caJ\intprod\varpi
        \end{equation}
        under the assumption that $\hat f$ does not vanish. As before, $\hat\caJ$ is an almost complex structure when $\hat f>0$, an almost para-complex structure when $\hat f<0$, and integrable if and only if $\hat f$ is constant. As in \cref{sec:geometry2dFlows}, we can always find a differential two-form $\,\hat\caK$ of type $(1,1)$ with respect to $\hat\caJ$ such that $\,\hat\caK\wedge\varpi=0$, $\,\hat\caK\wedge(\hat\caJ\intprod\varpi)=0$, and $\,\hat\caK\wedge\,\hat\caK\neq0$. In particular, we choose
        \begin{equation}\label{eq:alternativeAlmostKaehlerForm2d}
            \,\hat\caK\ \coloneqq\ \sqrt{|\hat f|}\,\bg{\nabla}q_i\wedge\rmd x^i~,
        \end{equation}
        that is, the standard symplectic structure~\eqref{eq:MongeAmpereStructure2d} times the same function as in~\eqref{eq:almostKaehlerForm2d}. Importantly, the compatibility of $\,\hat\caK$ and $\hat\caJ$ again yields the metric~\eqref{eq:fluidMetric2d}. It should be noted that the pull-back of the standard symplectic structure $\omega=\bg{\nabla}q_i\wedge\rmd x^i$ from~\eqref{eq:MongeAmpereStructure2d} under~\eqref{eq:incompressibilityFromLagrangian} is $\iota^*\omega=\rmd v$ and thus, this vanishes if and only if the vorticity~\eqref{eq:vorticityAndStrain} is zero.

        \begin{remark}
            In conclusion, the Monge--Amp{\`e}re structure $(\varpi,\alpha)$, with $\varpi$ defined by~\eqref{eq:alternativeSymplecticStructure} and $\alpha$ written as~\eqref{eq:alternativeMAForm}, represents alternative means to describe two-dimensional incompressible fluids. Whilst the Monge--Amp{\`e}re structure~\eqref{eq:MongeAmpereStructure2d} yields manifestly the description of the fluid flow in terms of a stream function, the advantage of this alternative Monge--Amp{\`e}re structure is that with this choice,\footnote{The triple of differential two-forms $\alpha$, $\bg{\nabla}q_i\wedge\rmd x^i$, and $\bg{\nabla}q_i\wedge{\star_{\bg{g}}\rmd x^i}$ define for $\hat f>0$ what is known as an almost quaternionic Hermitian structure on $T^*M$ and for $\hat f<0$ an almost quaternionic para-Hermitian structure, respectively, with the two choices of Monge--Amp{\`ere} structure we have presented corresponding to picking specific points in the moduli space of such structures.} we can straightforwardly generalise our treatment to fluid flows in any dimension. Essentially, this is due to the fact that the conditions $\iota^*\varpi=0$ and $\iota^*\alpha=0$ with $\iota$ given by~\eqref{eq:incompressibilityFromLagrangian} (with $i=1,\ldots,m$) are equivalent to the divergence-free constraint and the pressure equation in any dimension. However, in $m>2$ dimensions, we leave the realm of symplectic geometry as we shall explain shortly.
        \end{remark}
        
        \begin{remark}
        	At this stage, it is worth noting how our above choices deviate from constructions used in previous works. It is clear that~\eqref{eq:alternativeSymplecticStructure} and~\eqref{eq:alternativeMAForm} are a covariantisation of the Monge--Amp{\`e}re structure in~\cite{Banos:2015aa,Roulstone:2009aa}, with~\eqref{eq:alternativeDefinitionJ2d} the corresponding almost (para-)complex structure. However, we are free to make a choice of differential two-form in~\eqref{eq:alternativeAlmostKaehlerForm2d}, which corresponds to a choice of almost (para)-Hermitian metric on $T^*M$. In particular,~\cite{Banos:2015aa} works with the non-degenerate bilinear form 
        	\begin{equation}\label{eq:oldLRMetric}
        		g_\alpha(X,Y)\ \coloneqq\ \frac{[(X\intprod\alpha)\wedge(Y\intprod\varpi)+(Y\intprod\alpha)\wedge(X\intprod\varpi)]\wedge\vol{M}}{\frac12\varpi^2}
        	\end{equation}
            for all $X,Y\in\frX(T^*M)$. As discussed in~\cite{Roulstone:2001aa,Kossowski:1992aa}, the third differential two-form may be defined by\footnote{Whilst we present these expressions in our notation, the literature only treats the Euclidean case.} $\sqrt{|\hat f|}\,g_\alpha(\hat\caJ X,Y)$ for all $X,Y\in\frX(T^*M)$ in this case. The pull-back of $g_\alpha$ via~\eqref{eq:incompressibilityFromLagrangian} is then simply the Hessian of $\psi$ without vorticity as a conformal factor, in contrast to~\eqref{eq:fluidMetric2dPullBack}, where the vorticity is made manifest. Whilst the presence of the vorticity prefactor is clearly significant in our context, our choice is far from ad-hoc, as it arises perhaps even more naturally from the underlying geometry than~\eqref{eq:oldLRMetric}. Note also that the metric~\eqref{eq:oldLRMetric} has been linked in~\cite{Donofrio:2023aa} to a metric occurring in the theory of optimal mass transport in which optimal maps are characterised by volume-maximising Lagrangian submanifolds.  
        \end{remark}

        \subsection{Higher symplectic manifolds}\label{sec:kpm}

        The appropriate notion for our purposes is that of higher symplectic geometry. Here, we shall be rather brief and merely summarise some of the key facts that are needed for our subsequent discussion. For more details, we refer the interested reader to~\cite{Cantrijn:1999aa,Baez:2010,Rogers:2011}.

        \paragraph{Higher symplectic vector spaces.}
        To begin with, let $V$ be a real vector space and $\varpi\in\bigwedge^{k+1}V^*$ a $(k+1)$-form. Then, $\varpi$ is called \uline{non-degenerate} if and only if the contraction map $V\rightarrow\bigwedge^kV^*$, given by $v\mapsto v\intprod\varpi$ for all $v\in V$, is injective. Generally, the contraction map is not surjective; for $k=1$, however, injectivity implies surjectivity by the rank--nullity theorem, and we obtain the identification $V\cong V^*$ in this case. We call the pair $(V,\varpi)$ with $\varpi\in\bigwedge^{k+1}V^*$ non-degenerate a \uline{$k$-plectic vector space}. When $k=1$, we recover the standard case of a symplectic vector space. 

        Furthermore, for $U\subseteq V$ a vector subspace of $V$, we define the $\ell$-th \uline{orthogonal complement} $U^{\perp,\ell}$ for $\ell=1,\ldots,k$ with respect to $\varpi$ by
        \begin{equation}
            U^{\perp,\ell}\ \coloneqq\ \{v\in V\,|\,v\intprod u_1\ldots\intprod u_\ell\intprod\varpi=0\mbox{ for all }u_1,\ldots,u_\ell\in U\}~.
        \end{equation}
        Whenever $U=U^{\perp,\ell}$ for some $\ell=1,\ldots,k$, we call the vector subspace $U$ an \uline{$\ell$-Lagrangian subspace} of $V$. For $k=1$, there are only $1$-Lagrangian subspaces (or simply Lagrangian subspaces), and they all have the same dimension $\frac12\dim(V)$. For $k>1$, $\ell$-Lagrangian subspaces may have different dimensions.
        
        \paragraph{Higher symplectic manifolds.}
        Let $M$ be a manifold and $\varpi\in\Omega^{k+1}(M)$ a differential $(k+1)$-form which is point-wise non-degenerate. Suppose also that $\varpi$ is closed. Such a manifold is called a \uline{$k$-plectic manifold}, and in this case $\varpi$ is referred to as a \uline{$k$-plectic structure}. A diffeomorphism on $M$ that preserves $\varpi$ is called a \uline{$k$-plectomorphism}. Furthermore, a submanifold $L\hookrightarrow M$ is called an $\ell$-Lagrangian submanifold of $M$ if and only if $TL=TL^{\perp,\ell}$ for some $\ell=1,\ldots,k$. Here, we have used the obvious notation
        \begin{equation}
            TL^{\perp,\ell}\ \coloneqq\ \bigcup_{p\in L}\{(p,X_p)\,|\,X_p\in(T_pL)^{\perp,\ell}\}~.
        \end{equation}

        \subsection{Higher Monge--Amp{\`e}re geometry of three-dimensional fluid flows}\label{sec:mag3df}

        Having introduced the notion of $k$-plectic manifolds, we can now make precise the description of higher-dimensional incompressible fluid flows. 

        \paragraph{Higher Monge--Amp{\`e}re structure.}
        In $m=3$ dimensions, the components~\eqref{eq:curvatureTensors} of the Riemann and Ricci curvature tensors are related by the identity
        \begin{equation}
        	\bg{R}_{ijk}{}^l\ =\ 2\bg{R}^l{}_{[i}\bg{g}_{j]k}-2\big(\bg{R}_{k[i}-\tfrac12\bg{R}\bg{g}_{k[i}\big)\delta_{j]}{}^l~,
        \end{equation}
        where $\bg{R}$ is the curvature scalar. Upon following our above discussion and setting
        \begin{subequations}
        	\begin{equation}\label{eq:PressureCurvature3d}
        		\hat f\ \coloneqq\ \tfrac12\big(\bg{\lap}_{\rm B}p+\bg{R}^{ij}q_iq_j\big)\,,
        	\end{equation}
        	we consider the pair of differential three-forms
        	\begin{equation}\label{eq:MongeAmpereStructure3d}
            	\begin{gathered}
                	\varpi\ \coloneqq\ \bg{\nabla}q_i\wedge{\star_{\bg{g}}}\rmd x^i~,
                	\\
                	\alpha\ \coloneqq\ \tfrac12\bg{\nabla}q_i\wedge\bg{\nabla}q_j\wedge{\star_{\bg{g}}}(\rmd x^i\wedge\rmd x^j)-\hat f\vol{M}
            	\end{gathered}
        	\end{equation}
            on $T^*M$, where the volume form on $M$ is now given by
            \begin{equation}
        	   \vol{M}\ \coloneqq\ \tfrac{\sqrt{\det(\bg{g})}}{3!}\eps_{ijk}\rmd x^i\wedge\rmd x^j\wedge\rmd x^k~.
            \end{equation} 
        \end{subequations}
        Again, $\varpi$ is non-degenerate and closed by virtue of~\eqref{eq:closureOmega} and so, $\varpi$ defines a 2-plectic structure on $T^*M$. The submanifold $\iota:L\hookrightarrow T^*M$ defined by $\iota^*\varpi=0$ with $\iota$ given by~\eqref{eq:incompressibilityFromLagrangian} with $i=1,2,3$ is a three-dimensional $2$-Lagrangian submanifold of the $2$-plectic manifold $(T^*M,\varpi)$. As discussed above, the conditions $\iota^*\varpi=0$ and $\iota^*\alpha=0$ are equivalent to the divergence-free constraint~\eqref{eq:generalNavierStokesIncompressibilityLocal} and the pressure equation~\eqref{eq:pressureLaplacian}, respectively. Furthermore, by virtue of~\eqref{eq:dNablaq} and~\eqref{eq:closureAlpha}, $\alpha$ is closed. It is also non-degenerate, so $(T^*M,\alpha)$ defines a $2$-plectic manifold. Note that~\eqref{eq:MongeAmpereStructure3d} can be understood as a covariantisation of what was given previously in~\cite{Roulstone:2009aa,Banos:2015aa}. Note also that, for
        \begin{equation}\label{eq:standardSymp3D}
        	\omega\ =\ \bg{\nabla}q_i\wedge\rmd x^i
        \end{equation}
        the standard symplectic structure on $T^*M$, $\varpi\wedge\omega=0$ and $\alpha\wedge\omega=0$, so $\varpi$ and $\alpha$ are both Monge--Amp{\`e}re forms for $\omega$.

        Importantly, the formulation~\eqref{eq:MongeAmpereStructure3d} makes it transparent that this construction works in any dimension $m>1$. Indeed, we simply need to take the appropriate volume form in $\alpha$ and the function $\hat f$ is the same in any dimension. The pull-backs $\iota^*\varpi=0$ and $\iota^*\alpha=0$  then yield the divergence-free constraint and the pressure equation. Furthermore, $\varpi$ is $(m-1)$-plectic, and it defines an $m$-dimensional $(m-1)$-Lagrangian submanifold for general $m$. However, in general, whilst $\alpha$ is $(m-1)$-plectic, it may not define an $m$-dimensional $(m-1)$-Lagrangian submanifold. It should also be noted that $\alpha\wedge\varpi$ vanishes if and only if $m\neq 3$.

        \paragraph{Almost (para-)Hermitian structure.}
        Next, we wish to generalise the relation~\eqref{eq:alternativeDefinitionJ2d}. To this end, we use the results of~\cite{Hitchin:2000jd}. In particular, we note that the there is a isomorphism $\Omega^5(T^*M)\cong\frX(T^*M)\otimes\Omega^6(T^*M)$ that is induced by the natural exterior product pairing $\Omega^1(T^*M)\otimes\Omega^5(T^*M)\rightarrow\Omega^6(T^*M)$.\footnote{Explicitly, $\phi:\Omega^5(T^*M)\rightarrow\frX(T^*M)\otimes\Omega^6(T^*M)$ is given by $\phi(\rho)(\lambda,X_1,\ldots,X_6)\coloneqq X_1\intprod\ldots\intprod X_6\intprod(\rho\wedge\lambda)$ for all $\rho\in\Omega^5(T^*M)$, $\lambda\in\Omega^1(T^*M)$, and $X_1,\ldots,X_6\in\frX(T^*M)$.} Consequently, upon letting $\omega$ be the standard symplectic structure on $T^*M$ as in~\eqref{eq:standardSymp3D} and $\eps$ be the poly-vector field dual to the \uline{Liouville volume form} $\frac{1}{3!}\omega^3$ on $T^*M$, that is, $\eps\intprod\frac{1}{3!}\omega^3=1$, we may associate with the differential three-form $\alpha$ defined in~\eqref{eq:MongeAmpereStructure3d} the endomorphism
        \begin{equation}\label{eq:definitionJ3d}
            \hat\caJ X\ \coloneqq\ -\tfrac{1}{2\sqrt{|\hat f|}}\,\eps\intprod(\alpha\wedge X\intprod\alpha)
            \eforall
            X\ \in\ \frX(T^*M)
        \end{equation}
        under the assumption that $\hat f$ does not vanish. It then follows that $\hat\caJ$ is an almost complex structure on $T^*M$ when $\hat f>0$ and an almost para-complex structure when $\hat f<0$. 

        Furthermore, the differential two-form $\hat\caK$ defined in~\eqref{eq:alternativeAlmostKaehlerForm2d}, now with $i$ running from one to three, together with~\eqref{eq:definitionJ3d}, satisfies $\hat\caK(\hat\caJ X,Y)=-\hat\caK(X,\hat\caJ Y)$ for all $X,Y\in\frX(T^*M)$ and so, $\,\hat\caK$ of type $(1,1)$ with respect to~\eqref{eq:definitionJ3d}. Consequently, we can define an almost (para-)Hermitian metric $\hat g$ on $T^*M$ with respect to~\eqref{eq:definitionJ3d} by setting $\hat g(X,Y)\coloneqq\hat\caK(X,\hat\caJ Y)$ for all $X,Y\in\frX(T^*M)$. Explicitly,
        \begin{equation}\label{eq:fluidMetric3d}
            \hat g\ =\ \tfrac12\hat f\bg{g}_{ij}\rmd x^i\odot\rmd x^j+\tfrac12\bg{g}^{ij}\bg{\nabla}q_i\odot\bg{\nabla}q_j~.
        \end{equation} 
        Evidently, this metric is the direct generalisation of~\eqref{eq:fluidMetric2d}, and it is essentially a covariantisation that follows from a bilinear form introduced in~\cite{Lychagin:1993aa} (see also~\cite{Roulstone:2009aa}). This also justifies using the same letter $\hat\caJ$ in the definition~\eqref{eq:definitionJ3d} as it is a direct generalisation of~\eqref{eq:alternativeDefinitionJ2d}.

        \begin{remark}
            Recall that $\alpha$ defined in~\eqref{eq:MongeAmpereStructure3d} is closed. Furthermore, it can be taken as the imaginary part (under a quaternionic structure) of a differential (3,0)-form with respect to the almost (para-)complex structure $\hat\caJ$ defined in~\eqref{eq:definitionJ3d}.\footnote{Recall that a \uline{quaternionic structure} is an anti-linear endomorphism that squares to minus the identity. When $\hat\caJ$ is an almost complex structure, this quaternionic structure is simply complex conjugation.} In addition, since $\hat\caK$ defined in~\eqref{eq:alternativeAlmostKaehlerForm2d} is simply a rescaling by $\sqrt{|\hat f|}$ of the standard symplectic structure~\eqref{eq:standardSymp3D}, we conclude that also the standard symplectic structure is a differential $(1,1)$-form with respect to $\hat\caJ$. Hence, the tuple 
            \begin{equation}
                \left(\tfrac{1}{\sqrt{|\hat f|}}\hat\caK,\hat\caJ,\tfrac{1}{\sqrt{|\hat f|}}\hat g,\alpha\right)
            \end{equation}
            defines what is known as a \uline{nearly (para-)Calabi--Yau structure}~\cite{Xu:2006aa,Xu:2008aa}.
        \end{remark}

        \begin{remark}\label{rmk:dimensionalReduction}
            We can make the relationship between~\eqref{eq:alternativeDefinitionJ2d} and~\eqref{eq:definitionJ3d} more explicit. To make a notational distinction between the dimensions $m=2$ and $m=3$, we shall write $(M^m,\bg{g}_m)$ as well as $\varpi_m$ and $\alpha_m$ for~\eqref{eq:MongeAmpereStructure3d}, $\hat\caJ_m$ for~\eqref{eq:alternativeDefinitionJ2d} and~\eqref{eq:definitionJ3d}, and $\omega_m$ for the standard symplectic structure. In addition, we let $\eps_m$ be the poly-vector field dual to Liouville volume form on $T^*M^m$ with respect to $\omega_m$. 

            Firstly, it is not too difficult to see that~\eqref{eq:alternativeDefinitionJ2d} can be rewritten as
            \begin{equation}\label{eq:definitionJ2dRewritten}
                \hat\caJ_2X\ =\ \tfrac{1}{\sqrt{|\hat f|}}\,\eps_2\intprod(\varpi_2\wedge X\intprod\alpha_2)
                \eforall
                X\ \in\ \frX(T^*M^2)~.
            \end{equation}
            Note that $\omega_2\wedge\omega_2=\varpi_2\wedge\varpi_2$ and so, $\eps_2$ is also dual to the Liouville volume form with respect to $\varpi_2$. Next, let us assume that $M^3$ factorises as $M^3=M^2\times N$ with $N$ a one-dimensional manifold, and we take 
            \begin{equation}
                \bg{g}_3\ =\ \bg{g}_2+\rmd x^3\otimes\rmd x^3~,
            \end{equation}
            as the metric on $M^3$ with $\bg{g}_2$ a metric on $M^2$ and $x^3$ a local coordinate coordinate on $N$. Assuming that $p\in\scC^\infty(M^2)$, a short calculation then reveals that
            \begin{equation}
                \varpi_3\ =\ \varpi_2\wedge\rmd x^3+\vol{M^2}\wedge\rmd q_3
                \eand
                \alpha_3\ =\ \alpha_2\wedge\rmd x^3+\varpi_2\wedge\rmd q_3~.
            \end{equation}
            The decomposition for $\alpha_3$ and the effectiveness $\alpha_2\wedge\varpi_2=0$ imply that 
            \begin{equation}
            	\alpha_3\wedge(X\intprod\alpha_3)\ =\ -2(\varpi_2\wedge X\intprod\alpha_2)\,\rmd q_3\wedge\rmd x^3
            \end{equation}
            for all $X\in\frX(T^*M^2)$.\footnote{The horizontal lift of $X$ to $T^*M^3$ is trivial because of the assumed form of the metric $\bg{g}_3$.} Since $\eps_3=\eps_2\wedge\parder{x^3}\wedge\parder{q_3}$, this then yields 
			\begin{equation}            
            	\eps_3\intprod(\alpha_3\wedge X\intprod\alpha_3)\ =\ -2\eps_2\intprod(\varpi_2\wedge X\intprod\alpha_2)~.
            \end{equation}
            Consequently, combining this result with~\eqref{eq:definitionJ3d} and~\eqref{eq:definitionJ2dRewritten}, we finally obtain
            \begin{equation}
                \hat\caJ_3|_{M^2}\ =\ \hat\caJ_2~.
            \end{equation}
        \end{remark}

        \begin{remark}
            The metrics~\eqref{eq:fluidMetric2d} and~\eqref{eq:fluidMetric3d} on $T^*M$ are in spirit of the \uline{rescaled Sasaki metrics} studied e.g.~in~\cite{Wang:2011aa}. The main difference here is that our $\hat f$ is a function on $T^*M$ rather than on $M$. This results in a metric~\eqref{eq:fluidMetric3d} which is allowed to change type across $T^*M$. Furthermore, earlier work~\cite{Gezer:2014aa} has focused on constructing almost para-Nordenian manifolds in the case $\hat f>0$, preferentially selecting a structure which is almost para-complex~\cite{Cruceanu:1993aa}, as opposed to almost complex.
        \end{remark}

        Before moving on, we conclude with some remarks concerning the curvatures of the metric~\eqref{eq:fluidMetric3d} on $T^*M$ and its pull-back to $L$, as well as a connection with helicity.

        \paragraph{Curvature.}
        In view of our later applications, let us state the curvature scalar for the metric~\eqref{eq:fluidMetric3d}. The following is  derived in \cref{app:phaseMetricCurvature} and holds in any dimension. In particular, we have
        \begin{equation}\label{eq:curvatureScalarFluidMetric}
            \begin{aligned}
                \hat R\ &=\ \frac{1}{\hat f}\bg{R}-\frac{1}{4\hat f^2}\bg{R}_{ijk}{}^l\bg{R}^{ijkm}q_kq_m-(m-1)\hat\lap_{\rm B}\log(|\hat f|)-\bg{g}_{ij}\parder[^2]{q_i\partial q_j}\log(|\hat f|)
                \\
                &\kern.6cm+\frac{1}{4\hat f}(m-1)(m-2)\bg{g}^{ij}\left(\parder{x^i}+\bg{\Gamma}_{ik}{}^lq_l\parder{q_k}\right)\log(|\hat f|)\left(\parder{x^j}+\bg{\Gamma}_{jm}{}^nq_n\parder{q_m}\right)\log(|\hat f|)
                \\
                &\kern1.1cm+\frac14m(m-3)\bg{g}_{ij}\parder{q_i}\log(|\hat f|)\parder{q_j}\log(|\hat f|)~,
            \end{aligned}
        \end{equation}
        where $\hat\lap_{\rm B}$ is the Beltrami Laplacian for $\hat g$. The occurrence of the term $\hat\lap_{\rm B}\log(|\hat f|)$ again suggests that the accumumlation of $\hat f$ will determine the sign of the scalar curvature, as it does in the two-dimensional case. 

        \paragraph{Pull-back metric.}
        It is a straightforward exercise to check that the pull-back $g\coloneqq\iota^*\hat g$ of~\eqref{eq:fluidMetric3d} to the 2-Lagrangian submanifold $L$ via $\iota$ given by~\eqref{eq:incompressibilityFromLagrangian} with $i=1,2,3$ is
        \begin{equation}\label{eq:fluidMetric3dPullBack}
            g\ =\ \tfrac12g_{ij}\rmd x^i\odot\rmd x^j
            \ewith
            g_{ij}\ \coloneqq\ A^k{}_iA_{kj}-\tfrac12{\bg g}_{ij}A_{kl}A^{lk}~.
        \end{equation}
        Here, we have made use of the \uline{velocity-gradient tensor} $A_{ij}\coloneqq\bg\nabla_j v_i$ and noted that 
        \begin{equation}\label{eq:pullBackHatf3D}
            f\ \coloneqq\ \iota^*\hat f\ =\ \tfrac12(\zeta_{ij}\zeta^{ij}-S_{ij}S^{ij})\ =\ -\tfrac12A_{ij}A^{ji}
        \end{equation}
        with $\zeta_{ij}$ the vorticity two-form and $S_{ij}$ the rate-of-strain tensor introduced in~\eqref{eq:vorticityAndStrain}, and the indices on $\zeta_{ij}$, $S_{ij}$, and $A_{ij}$ raised with the background metric. Again, as in the two-dimensional case, the pull-back metric is a quadratic function of the velocity gradient tensor and curvature will be generated by gradients of vorticity and rate-of-strain. 

        \paragraph{Helicity.}
		In two dimensions, we utilised the local Gau{\ss}--Bonnet theorem~\eqref{eq:localGaussBonnet} in order to relate the geometry of fluid flows, as described by the curvature scalar~\eqref{eq:curvature2D}, to a topological invariant, namely the Euler characteristic of a given compact region. In three dimensions, it quickly becomes apparent that this is not a suitable approach and that we require an alternative topological quantity.

		Recall that the pull-back of the standard symplectic form $\omega$ seen in~\eqref{eq:standardSymp3D}, under~\eqref{eq:incompressibilityFromLagrangian}, is $\iota^*\omega=\rmd v=\zeta_{ij}\rmd x^i\wedge\rmd x^j$ with $\zeta_{ij}$ the vorticity. The pull-back of the associated tautological one-form $\theta\coloneqq q_i\rmd x^i$ is simply $v=v_i\rmd x^i$. It then follows that
        \begin{equation}\label{eq:helicity}
            \iota^*(\theta\wedge\rmd\theta)\ =\ v_i\zeta^i\,\vol{M^3}
            \ewith
            \zeta^i\ \coloneqq\ \sqrt{\det(\bg{g}_3)}\eps^{ijk}\zeta_{jk}
        \end{equation}
        the vorticity in three dimensions, derived from~\eqref{eq:vorticityAndStrain}. Integrals of quantities of the form~\eqref{eq:helicity}, over a compact region $U\subseteq M^3$, are referred to as \uline{helicity}~\cite{Moffatt:1969aa,Woltjer:1958aa}. Hence, in our context, $v_i\zeta^i$ may be referred to as the \uline{helicity per volume}.
		
		Consider an inviscid, incompressible fluid, with kinematics described by the Euler equations, on a compact region $U\subseteq M^3$. Suppose also that $U$ describes the volume contained inside a closed orientable surface, which is moving with the fluid and has (outward) unit normal $n$ with components denoted $n_i$. It is shown\footnote{In the context of magneto-hydrodynamics, the analogous result was presented in~\cite{Woltjer:1958aa}.} in~\cite{Moffatt:1969aa} that, provided the distribution of vorticity is local and continuous, and $n_i\zeta^i=0$, then the integral of~\eqref{eq:helicity} is an invariant of the Euler equations and the vorticity field within the volume is conserved. Furthermore, it is shown that for discrete vortex filaments, this quantity can be associated\footnote{We point the interested reader towards~\cite{Moffatt:1992aa,Ricca:1992aa} for elaboration on these associations.} with the topological invariants given by the \uline{Gau{\ss} linking number} and \uline{C\u{a}lug\u{a}reanu invariant}~\cite{Calugareanu:1959aa,Calugareanu:1961aa}. In~\cite{Whitehead:1947aa} it was also shown that helicities are isotopy invariants of their volume. Perhaps more significantly, a recent work~\cite{Liu:2012aa} has managed to demonstrate that, in ideal conditions, helicity-type quantities can be reinterpreted as Abelian Chern--Simons actions and hence can be related to the Jones polynomial.
		
		In addition to the interpretation of the pull-backs of~\eqref{eq:MongeAmpereStructure3d} under~\eqref{eq:incompressibilityFromLagrangian} as the divergence-free constraint and the pressure equation, we now also have that the corresponding pull-back of the standard symplectic form encodes the helicity. Additionally, previous work relating helicity to various topological invariants suggests that, as in two dimensions, one can relate the topology of fluid flows to our geometric constructions.

        \subsection{Examples in three dimensions}\label{sec:eg3d}

        In this section, we adopt the notation $x\coloneqq x^1$, $y\coloneqq x^2$, and $z\coloneqq x^3$ and consider some classical examples of flows in $\IR^3$ with background metric $\bg{g}_{ij}=\delta_{ij}$.

        \paragraph{Preliminaries.}
        Recall the expressions of the vorticity two-form and the rate-of-strain tensor defined in~\eqref{eq:vorticityAndStrain}. In view of our discussion below, we set 
    	\begin{subequations}
            \begin{equation}
                \zeta\ \eqqcolon\
                \begin{pmatrix}
                    0 & \zeta_3 & -\zeta_2
                    \\
                    -\zeta_3 & 0 & \zeta_1
                    \\
                    \zeta_2 & -\zeta_1 & 0
                \end{pmatrix}
                \eand
                S\ \eqqcolon\
                \begin{pmatrix}
        			\alpha & \sigma_3 & \sigma_2
                    \\
        			\sigma_3 & \beta & \sigma_1
                    \\
        			\sigma_2 & \sigma_1 & \gamma		
        		\end{pmatrix}
            \end{equation}
            and introduce the velocity-gradient matrix
            \begin{equation}\label{eq:VGT}
                A\ \coloneqq\ S-\zeta\ =\
                \begin{pmatrix}
                    \alpha & \sigma_3 - \zeta_3 & \sigma_2+\zeta_2
                    \\
                    \sigma_3+\zeta_3 & \beta & \sigma_1-\zeta_1
                    \\
                    \sigma_2-\zeta_2 & \sigma_1 + \zeta_1 & \gamma
                \end{pmatrix}.
        	\end{equation}
        \end{subequations}
        Furthermore, the metric~\eqref{eq:fluidMetric3d} in the then takes the form
        \begin{equation}
        	\hat g\ =\ 
            \begin{pmatrix}
                \frac12\lap p\unit_3 & 0
                \\
                0 & \unit_3
            \end{pmatrix}
            \ewith
            \lap p\ =\ -\tr(A^2)~.
        \end{equation}
        It now follows that the pull-back metric~\eqref{eq:fluidMetric3dPullBack} is 
        \begin{equation}\label{eq:fluidMetric3dPullBackFlat}
        	g\ =\ A^\sfT A-\tfrac12\tr(A^2)\unit_3~.
        \end{equation}
        Whilst it is now possible to substitute~\eqref{eq:VGT} into~\eqref{eq:fluidMetric3dPullBackFlat}, the result would not be particularly helpful. To see the structure of the pull-back metric a little more clearly, we next consider Burgers' canonical model of the vortex, for which the velocity-gradient matrix takes a relatively simple form, and which in turn motivates our work on higher symplectic reduction. Studying this example will show how the signature of the metric depends on relationships between vorticity and rate-of-strain.
    
        \paragraph{Burgers' vortex.}
        Earlier works~\cite{Roulstone:2009bb,Banos:2015aa} considered a class of solutions to the three-dimensional incompressible Euler and Navier--Stokes equations, with Euclidean background metric, that can be reduced to solutions to the incompressible equations in two dimensions via the \uline{Lundgren transformation}~\cite{Lundgren:1982aa}. Such solutions, which take the form~\cite{Ohkitani:2000aa}
		\begin{equation}\label{eq:flow3d}
            (\dot x,\dot y,\dot z)\ \coloneqq\ \big(v_x(t,x,y),v_y(t,x,y),z\phi(t,x,y) + W(t,x,y)\big)
		\end{equation}
        for some functions $\phi$ and $W$, where the superposed dot refers to the derivative with respect to the time $t$, are often referred to as \uline{two-and-a-half-dimensional flows}~\cite{Gibbon:1999aa}. In particular, a geometric description of Burgers' vortex~\cite{Burgers:1948aa} has been presented through this lens. 

        Consider the following idealised \uline{Burgers' vortex}~\cite{Burgers:1948aa} with the velocity components be given by
        \begin{equation}
        	u\ =\ \alpha x+(\sigma_3-\zeta_3)y~,
            \quad
        	v\ =\ \beta y+(\sigma_3+\zeta_3)x~,
            \eand
        	w\ =\ \gamma z
        \end{equation}
        with $\alpha$, $\beta$, $\gamma$, $\sigma_3$, and $\zeta_3$ constant in space. Then, the divergence-free constraint is given by $\alpha+\beta+\gamma=0$. In particular, we have chosen $\phi=\gamma(t)$ and $W\equiv 0$ in~\eqref{eq:flow3d}.
       
        Next, the velocity-gradient matrix is simply
        \begin{equation}\label{eq:BurgersVGTForm}
            A\ =\
            \begin{pmatrix}
                \alpha & \sigma_3-\zeta_3 & 0
                \\
                \sigma_3+\zeta_3 & \beta & 0
                \\
                0 & 0 & \gamma
            \end{pmatrix},
        \end{equation}
        and using the divergence-free constraint, it follows that
        \begin{equation}\label{eq:BurgersPressure}
	       \tfrac12\lap p\ =\ -\tfrac12\tr(A^2)\ =\ \alpha\beta+\gamma(\alpha+\beta)+\zeta_3^2-\sigma_3^2~.
        \end{equation}
        We can now deduce that the pull-back metric~\eqref{eq:fluidMetric3dPullBackFlat} is 
        \begin{equation}\label{eq:BurgersMetric}
        	g\ =\
            \begin{pmatrix}
                \gamma\beta+2\zeta_3(\sigma_3+\zeta_3) & \alpha(\sigma_3-\zeta_3)+\beta(\sigma_3+\zeta_3) & 0
                \\
                \alpha(\sigma_3-\zeta_3)+\beta(\sigma_3+\zeta_3) & \gamma\alpha+2\zeta_3(\zeta_3-\sigma_3) & 0
                \\
                0 & 0 & \alpha\beta+(\zeta_3-\sigma_3)(\zeta_3+\sigma_3)
        	\end{pmatrix}
        \end{equation}
        which has eigenvalues 
        \begin{equation}
        	\begin{gathered}
        		E_{\pm}\ =\ \tfrac12\Big\{4\zeta_3^2-\gamma^2\pm\sqrt{\gamma^2(\alpha-\beta)^2+4\big[4\sigma_3^2\zeta_3^2+(\alpha+\beta)^2\sigma_3^2 +(\alpha-\beta)^2\zeta_3^2\big]}\Big\}\,, 
                \\
        		E_3\ =\ \alpha\beta-\sigma_3^2+\zeta_3^2~.
        	\end{gathered}
        \end{equation}
        The top left $(2\times 2)$-block of~\eqref{eq:BurgersMetric}, with $\gamma=0$ is precisely the pull-back metric of an incompressible two-dimensional flow with velocity-gradient matrix given by the top left $(2\times 2)$-block of~\eqref{eq:BurgersVGTForm}, where $\gamma=-(\alpha+\beta)=0$. Hence, setting $\gamma\neq 0$ produces compressible two-dimensional flows, for example Burgers' vortex after reduction as described in~\cite{Banos:2015aa}. It follows that when $\lap p>0$, $E_3>0$. Furthermore, if we assume axi-symmetry by setting $\alpha=\beta=-\tfrac12 \gamma$, then with $\lap p>0$ and $E_3>0$, we have $E_+>0$, while $E_-$ is bounded below by $-\gamma^2$. Further investigation of such criteria might facilitate a classification of conditions under which accumulations of vorticity could constitute `a vortex'. 

        With~\cite{Banos:2015aa} in mind, we now show how higher symplectic geometry and reductions thereof, provide a mechanism for formulating the Monge--Amp{\`e}re geometry of certain exact solutions to the incompressible Navier--Stokes equations in three dimensions. 

        \subsection{Higher symplectic reductions}\label{sec:kprs}

        In the following, we wish to study dimensional reductions from three to two dimensions. In particular, we shall focus on symplectic and higher symplectic reductions. This will enable us to study fluid flows in three dimensions with symmetries that eventually can be analysed as effective two-dimensional problems. As we shall explain, symplectic and higher symplectic reductions yield, to the extent in which we are interested in this paper, essentially the same geometric information in two dimensions; however, the higher symplectic reduction will yield the two-dimensional problem directly in terms of a stream function, thus resolving the lower-dimensional would-be divergence-free constraint automatically. Before analysing examples in \cref{sec:examplesReduction}, including the Arnol'd--Beltrami--Childress flow and Hicks--Moffatt-type vortices, let us set the stage. In particular, we first recall the \uline{Marsden--Weinstein reduction process}~\cite{Marsden:1974aa,Meyer:1973aa}, a well-known tool from symplectic geometry, for reducing spaces with symmetries. Concretely, this reduction process can be summarised as follows.
        
        \begin{theorem}\label{thm:marsdenWeinstein}
            Let $(M,\omega)$ be a symplectic manifold. Suppose that $\sfG$ is a Lie group acting by symplectomorphisms on $(M,\omega)$. Let $\mu:M\rightarrow\frg^*$ be the moment map for this action with $\frg$ the Lie algebra of $\sfG$. Furthermore, let $c\in\frg^*$ be a regular value of $\mu$ and $\sfG_c\subseteq\sfG$ the (coadjoint) stabiliser group of $c$. We assume that $\sfG_c$ acts freely and properly on $\mu^{-1}(\{c\})$. Set $M_c\coloneqq\mu^{-1}(\{c\})/\sfG_c$ and consider,
            \begin{equation}
                \begin{tikzcd}
                    \mu^{-1}(\{c\}) \arrow[r, hook, "\fri"]\arrow[d, two heads, "\frp"] & M
                    \\
                    M_c &
                \end{tikzcd}
            \end{equation}
            Then, there exists a unique symplectic structure $\omega_c$ on $M_c$ such that $\frp^*\omega_c=\fri^*\omega$.
        \end{theorem}

        To discuss symmetry reductions of higher-dimensional fluid flows which are described directly in terms of higher Monge--Amp{\`e}re structures, we would like to generalise this result to the higher symplectic geometry summarised in \cref{sec:geometricProperties3d}. Fortunately for us, \cref{thm:marsdenWeinstein} has been generalised to the $k$-plectic case rather recently as follows~\cite{Blacker:2021aa}.
        
        \begin{theorem}\label{thm:generalisedMarsdenWeinstein}
            Let $(M,\varpi)$ be a $k$-plectic manifold. Suppose that $\sfG$ is a Lie group acting by $k$-plectomorphisms on $(M,\varpi)$. Let $\mu:M\rightarrow\bigwedge^{k-1}T^*M\otimes\frg^*$ be the moment map for this action with $\frg$ the Lie algebra of $\sfG$. Furthermore, let $c\in\Omega^{k-1}(M,\frg^*)$ be closed and define
            \begin{equation}
                \begin{aligned}
                    \mu^{-1}(\{c\})\ &\coloneqq\ \{x\in M\,|\,\mu(x)=c_x\}~,
                    \\
                    \sfG_c\ &\coloneqq\ \big\{g\in\sfG\,\big|\,g^{-1}_*X_1\intprod\ldots\intprod g^{-1}_*X_{k-1}\intprod \Ad^*_g c_{g^{-1}x}=X_1\intprod\ldots\intprod X_{k-1}\intprod c_x
                    \\
                    &\kern1.5cm\text{ for all }x\in M\text{ and for all }X_1,\ldots,X_{k-1}\in T_xM\big\}~.
                \end{aligned}
            \end{equation}
            Suppose that $\mu^{-1}(\{c\})$ is an embedded submanifold of $M$ and that $\sfG_c$ acts freely and properly on $\mu^{-1}(\{c\})$. Set $M_c\coloneqq\mu^{-1}(\{c\})/\sfG_c$ and consider,
            \begin{equation}
                \begin{tikzcd}
                    \mu^{-1}(\{c\}) \arrow[r, hook, "\fri"]\arrow[d, two heads, "\frp"] & M
                    \\
                    M_c &
                \end{tikzcd}
            \end{equation}
            Then, there exists a unique closed differential form $\varpi_c\in\Omega^{k+1}(M_c)$ on $M_c$ such that $\frp^*\varpi_c=\fri^*\varpi$.
        \end{theorem}

        \noindent
        Evidently, for $k=1$ this result reduces to \cref{thm:marsdenWeinstein}. It is important to stress that for $k>1$, the differential form $\varpi_c\in\Omega^{k+1}(M_c)$ might be degenerate. For full details of the above, see~\cite{Blacker:2021aa}.

        \paragraph{Setting for dimensional reduction.}
        In \cref{rmk:dimensionalReduction}, we have already discussed a simple dimensional reduction of the Monge--Amp{\`e}re structure~\eqref{eq:MongeAmpereStructure3d} by assuming that the three-dimensional background manifold $M^3$ is a direct product of a two-dimensional manifold $M^2$ and a one-dimensional manifold $N$. Let us now assume it is of warped-product form instead, that is, we take
        \begin{equation}\label{eq:warpedProduct}
            \bg{g}_3\ =\ \bg{g}_2+\rme^{2\varphi}\,\rmd x^3\otimes\rmd x^3
        \end{equation}
        with $\varphi\in\scC^\infty(M^2)$ as the metric on $M^3$ where $\bg{g}_2$ is a metric on $M^2$ and $x^3$ local coordinates on $N$, respectively. Put differently, we assume that there is a one-parameter family of isometries, and we choose adapted coordinates. Now let $i,j,\ldots=1,2$, such that
        \begin{equation}
        	\bg{g}_2\ =\ \tfrac12\bg{g}_{ij}\rmd x^i\odot\rmd x^j~.
        \end{equation}
        Hence, the only non-vanishing Christoffel symbols for the metric $\bg{g}_3$ are $\bg{\Gamma}_{ij}{}^k$, which are precisely the Christoffel symbols for $\bg{g}_2$, alongside
        \begin{equation}
        	\bg{\Gamma}_{33}{}^i\ =\ -\rme^{2\varphi}\bg{g}^{ij}\partial_j\varphi
            \eand
            \bg{\Gamma}_{i3}{}^3\ =\ \partial_i\varphi~.
        \end{equation}

        Next, consider the differential forms~\eqref{eq:MongeAmpereStructure3d} on $M^m$ for $m=2,3$. As in \cref{rmk:dimensionalReduction}, let us denote them by $\varpi_m$ and $\alpha_m$, and let us also use a similar notation for other quantities. Then, under the assumption that $p\in\scC^\infty(M^2)$, some algebra reveals that 
        \begin{subequations}\label{eq:warpedReduction}
            \begin{equation}
                \begin{aligned}
                    \varpi_3\ &=\ \rme^\varphi\,\varpi_2\wedge\rmd x^3+\rme^{-\varphi}\,\vol{M^2}\wedge\bg{\nabla}q_3~,
                    \\
                    \alpha_3\ &=\ \rme^\varphi(\alpha_2-\hat h_+\vol{M^2})\wedge\rmd x^3+\rme^{-\varphi}\,(\varpi_2-q_3\rmd x^3\wedge{\star_{\bg{g}_2}\rmd\varphi})\wedge\bg{\nabla}q_3
                \end{aligned}
            \end{equation}
            with 
            \begin{equation}\label{eq:warpedReductionhath}
                \hat h_\pm\ \coloneqq\ \tfrac12[\bg{\nabla}^i\varphi\partial_ip-\big(\bg{\nabla}^i\bg{\nabla}^j\varphi\pm\bg{\nabla}^i\varphi\bg{\nabla}^j\varphi\big)q_iq_j-\rme^{-2\varphi}(\bg{\lap}_{\rm B}\varphi\pm\bg{\nabla}^i\varphi\partial_i\varphi)q_3^2\big]\,,
            \end{equation}
        \end{subequations}
        where again all differential operators in $\hat h_\pm$ are with respect to the metric $\bg{g}_2$. 

        Furthermore, we obtain
        \begin{equation}
            \begin{aligned}\label{eq:reducedMA}
                \varpi'_2\ &\coloneqq\ \parder{x^3}\intprod\varpi_3
                \\ 
                &\kern2.5pt=\ \rme^\varphi\big(\varpi_2+q_i\bg{\nabla}^i\varphi\,\vol{M^2}\big)\,,
                \\
                \alpha'_2\ &\coloneqq\ \parder{x^3}\intprod\alpha_3
                \\
                &\kern2.5pt=\ \rme^\varphi\big[\alpha_2-\big(\hat h_++\rme^{-2\varphi}\bg{\nabla}^i\varphi\partial_i\varphi\,q_3^2\big)\,\vol{M^2}+q_i\bg{\nabla}^i\varphi\,\varpi_2\big]+\rme^{-\varphi}q_3\rmd q_3\wedge{\star_{\bg{g}_2}\rmd\varphi}~.
            \end{aligned}
        \end{equation}
        A short calculation then shows that both $\varpi'_2$ and $\alpha'_2$ are closed. In fact, using~\eqref{eq:dstardx}, we also have that
        \begin{equation}\label{eq:exactnessReducedOmega3}
           \varpi'_2\ =\ \rmd({\star_{\bg{g}_2}\rme^\varphi q_i\rmd x^i})~.
        \end{equation}

        \paragraph{Symplectic reduction of the higher Monge--Amp{\`e}re structure.}
        Given that $\varpi'_2$ and $\alpha'_2$ are closed, Cartan's formula for Lie derivatives then immediately yields that $\caL_{\parder{x^3}}\varpi_3=0=\caL_{\parder{x^3}}\alpha_3$. Consequently, we can consider a dimensional reduction following \cref{thm:marsdenWeinstein}. In particular, we take the symplectic form
        \begin{equation}
            \omega_3\ \coloneqq\ \rmd q_i\wedge\rmd x^i-\rmd(\lambda q_3)\wedge\rmd x^3~,
        \end{equation}
        where $\lambda\in\scC^\infty(M^2)$ is assumed to be non-vanishing. Evidently, $\parder{x^3}\intprod\omega_3=\rmd(\lambda q_3)$, so the corresponding moment map is $\mu(x,q)=\lambda q_3$. Hence, 
        \begin{equation}
        	\mu^{-1}(\{c\})\ =\ \{(x,q)\,|\,q_3=c\,\lambda^{-1}\}
        \end{equation}
        for any regular value $c\in\IR$. Consequently, $\mu^{-1}(\{c\})/\sfG_c$ is locally given by $(x^i,x^3,q_i,q_3)=(x^i,\text{const},q_i,q_3=q_3(x^i))$. Next, by virtue of \cref{thm:marsdenWeinstein}, we obtain the symplectic form
        \begin{equation}
        	\omega_c\ \coloneqq\ \rmd q_i\wedge\rmd x^i
        \end{equation}
        on $\mu^{-1}(\{c\})/\sfG_c\cong T^*M^2$ which satisfies $\frp^*\omega_c=\fri^*\omega_3$, as well as two closed differential two-forms given by
        \begin{equation}\label{eq:reducedMA2}
            \begin{aligned}
                \tilde\varpi_2\ &\coloneqq\ \rme^\varphi\big(\varpi_2+q_i\bg{\nabla}^i\varphi\,\vol{M^2}\big)\,,
                \\
                \tilde\alpha_2\ &\coloneqq\ \rme^\varphi\big\{\alpha_2-\big[\hat h_++\rme^{-2\varphi}\big(\bg{\nabla}^i\varphi\partial_i\varphi\,q_3^2-q_3\bg{\nabla}^i\varphi\partial_iq_3\big)\big]\,\vol{M^2}+q_i\bg{\nabla}^i\varphi\,\varpi_2\big\}\,,
            \end{aligned}
        \end{equation}
        which are simply the differential two-forms from~\eqref{eq:reducedMA} with $q_3$ understood as a function of $x^1$ and $x^2$. Upon requiring the vanishing of the pull-back of $\tilde\varpi_2$ and $\tilde\alpha_2$ along~\eqref{eq:incompressibilityFromLagrangian} together with the relabelling the function $q_3$ by $v_3$, we obtain
        \begin{equation}\label{eq:reducedKinematicSystem}
            \begin{aligned}
                \bg{\nabla}_iv^i\ &=\ -v^i\partial_i\varphi~,
                \\
                \bg{\lap}_{\rm B}p+\bg{\nabla}_iv^j\bg{\nabla}_jv^i+\tfrac12|v|^2\bg{R}\ &=\ -\bg{g}^{ij}\partial_i\varphi\partial_j p+v^iv^j\bg{\nabla}_i\partial_j\varphi
                \\
                &\kern1cm+\rme^{-2\varphi}\big[\big(\bg{\lap}_{\rm B}\varphi-\bg{g}^{ij}\partial_i\varphi\partial_j\varphi\big)v_3^2+2v_3\bg{g}^{ij}\partial_i\varphi\partial_jv_3\big]~.
            \end{aligned}
        \end{equation}
        These are precisely the divergence-free constraint~\eqref{eq:generalNavierStokesIncompressibilityLocal} and the pressure equation~\eqref{eq:pressureLaplacian} when adapted to the warped product metric~\eqref{eq:warpedProduct}, under the assumption that $p$ is independent of $x^3$. Evidently, when $\varphi=0$, we obtain the standard situation of an incompressible fluid flow in two dimensions from \cref{sec:geometry2dFlows}, and $v_3$ is not constrained by~\eqref{eq:reducedKinematicSystem}.

        Next, let $X$ be a vector field on $\mu^{-1}(\{c\})/\sfG_c\cong T^*M^2$ and consider its horizontal lift $\tilde X$ to $T^*M^3$ using the Levi-Civita connection for the metric~\eqref{eq:warpedProduct},
        \begin{equation}
            \tilde X\ \coloneqq\ X+X\intprod\rmd x^i\bg{\Gamma}_{i3}{}^3q_3\parder{q_3}\ =\ X+X\intprod\rmd\varphi\,q_3\parder{q_3}~.
        \end{equation}
        Using that
        \begin{equation}
        	\tilde X\intprod\bg{\nabla}q_3\ =\ 0
            \eand
            \varpi_2\wedge(\alpha_2-\hat h_+\,\vol{M^2})\ =\ 0~,
        \end{equation}
        we obtain 
        \begin{equation}
        	\alpha_3\wedge(\tilde X\intprod\alpha_3)\ =\ -2\varpi_2\wedge X\intprod(\alpha_2-\hat h_+\,\vol{M^2})\wedge\bg{\nabla}q_3\wedge\rmd x^3~.
        \end{equation}
        Consequently, the endomorphism~\eqref{eq:definitionJ3d} becomes
        \begin{equation}
            \hat\caJ_3\tilde X\ =\ \frac{1}{\sqrt{|\hat f_2+\hat h_+|}}\,\eps_2\intprod\big[\varpi_2\wedge X\intprod\big(\alpha_2-\hat h_+\,\vol{M^2}\big)\big]\,,
        \end{equation}
        where $\eps_2$ is the dual to the Liouville volume form on $T^*M^2$; see also \cref{rmk:dimensionalReduction}. Hence, we obtain an endomorphism $\hat\caJ_2$ on $\mu^{-1}(\{c\})/\sfG_c$ that is precisely of the form~\eqref{eq:definitionJ2dRewritten} (or, equivalently of the form~\eqref{eq:alternativeDefinitionJ2d}) when using the Monge--Amp{\`e}re structure $\big(\varpi_2,\alpha_2-\hat h_+\,\vol{M^2}\big)$. Here, $\hat h_+$ is considered to be a function of $(x^i,q_i)$ only, since $\mu^{-1}(\{c\})/\sfG_c\cong T^*M^2$, with $q_3=v_3(x^i)$. Note $\alpha_2-\hat h_+\,\vol{M^2}$ is simply $\alpha_2$ with $\hat f_2$ replaced with $\hat f_2+\hat h_+$. Also, whilst $\varpi_2$ is a symplectic form on $T^*M^2$, $\alpha_2-\hat h_+\,\vol{M^2}$ fails to be closed and is degenerate when $\hat f_2+\hat h_+=0$. Next, we consider~\eqref{eq:alternativeAlmostKaehlerForm2d} and set
        \begin{equation}
            \hat\caK_2\ \coloneqq\ \sqrt{|\hat f_2+\hat h_+|}\,\bg{\nabla}q_i\wedge\rmd x^i~.
        \end{equation}
        Then, as before, $\hat\caK_2(\hat\caJ_2X,Y)=-\hat\caK_2(X,\hat\caJ_2Y)$ for all vector fields $X$ and $Y$ on $\mu^{-1}(\{c\})/\sfG_c$ so that $\hat g_2(X,Y)\coloneqq\hat\caK_2(X,\hat\caJ_2Y)$ is an almost (para-)Hermitian metric on $\mu^{-1}(\{c\})/\sfG_c$. Explicitly,
        \begin{equation}\label{eq:fluidMetric3dReduced}
            \hat g_2\ =\ \tfrac12(\hat f_2+\hat h_+)\bg{g}_{ij}\rmd x^i\odot\rmd x^j+\tfrac12\bg{g}^{ij}\bg{\nabla}q_i\odot\bg{\nabla}q_j~.
        \end{equation}

        \paragraph{Higher symplectic reduction of the higher Monge--Amp{\`e}re structure.}
        Let us now discuss the $2$-plectic reduction of the Monge--Amp{\`e}re structure~\eqref{eq:MongeAmpereStructure3d} following \cref{thm:generalisedMarsdenWeinstein}. In particular, by virtue of exactness~\eqref{eq:exactnessReducedOmega3}, we can take
        \begin{equation}\label{eq:2PlecticMoment}
            \mu(x,q)\ =\ {\star_{\bg{g}_2}\rme^\varphi q_i\rmd x^i}
        \end{equation}
        as the moment map which, of course, is defined up to a shift by an exact form. Then, for $\psi\in\scC^\infty(M^2)$, it follows the $\mu^{-1}(\{-\rmd\psi\})$ is non-empty and given by
        \begin{equation}\label{eq:2plecticInverseImageMomentMap}
            \mu^{-1}(\{-\rmd\psi\})\ =\ \big\{(x,q)\,\big|\,q_i=-\sqrt{\det(\bg{g}_2)}\,\rme^{-\varphi}\eps_{ij}\bg{g}^{jk}\partial_k\psi\big\}~.
        \end{equation}
        Consequently, the quotient $\mu^{-1}(\{-\rmd\psi\})/\sfG_{-\rmd\psi}$ is locally given by $(x^i,x^3,q_i,q_3)=\big(x^i,\text{const},-\sqrt{\det(\bg{g}_2)}\,\rme^{-\varphi}\eps_{ij}\bg{g}^{jk}\partial_k\psi,q_3\big)$. Furthermore, we obtain
        \begin{equation}
        	\varpi_{-\rmd\psi}\ \coloneqq\ \rme^{-\varphi}\vol{M^2}\wedge\rmd q_3
        \end{equation}
        on $\mu^{-1}(\{-\rmd\psi\})/\sfG_{-\rmd\psi}$, which satisfies $\frp^*\varpi_{-\rmd\psi}=\fri^*\varpi_3$. In addition, whilst the pull-back of $\varpi'_2$ given in~\eqref{eq:reducedMA} to $\mu^{-1}(\{-\rmd\psi\})$ vanishes identically, there is a closed differential two-form $\alpha_{-\rmd\psi}$ on $\mu^{-1}(\{-\rmd\psi\})/\sfG_{-\rmd\psi}$ given by
        \begin{equation}\label{eq:2plecticReducedForm}
            \alpha_{-\rmd\psi}\ \coloneqq\ \rme^\varphi\big[\det(\bg{\nabla}^iq_j)-\big(\hat f_2+\hat h_-\big)\big]\Big|_{q_i=-\sqrt{\det(\bg{g}_2)}\,\rme^{-\varphi}\eps_{ij}\bg{g}^{jk}\partial_k\psi}\vol{M^2}+\rme^{-\varphi}q_3\rmd q_3\wedge{\star_{\bg{g}_2}\rmd\varphi}    
        \end{equation}
        and which satisfies $\frp^*\alpha_{-\rmd\psi}=\fri^*\alpha'_2$. The function $\hat h_-$ used here was defined in~\eqref{eq:warpedReductionhath}. Finally, upon requiring the vanishing of the pull-back of $\alpha_{-\rmd\psi}$ along 
        \begin{equation}\label{eq:pullback2Plectic}
            \iota\,:\,x^i\ \mapsto\ (x^i,q_3)\ \coloneqq\ (x^i,v_3(x^i))~,
        \end{equation}
        we obtain the system~\eqref{eq:reducedKinematicSystem} with $v_i$ given by
        \begin{equation}\label{eq:velocity2Plectic}
        	v_i\ =\ -\sqrt{\det(\bg{g}_2)}\,\rme^{-\varphi}\eps_{ij}\bg{g}^{jk}\partial_k\psi\,.
        \end{equation}
		Evidently, the first equation of~\eqref{eq:reducedKinematicSystem} can be rewritten as $\bg{\nabla^i}(\rme^\varphi v_i)=0$, and by the Poincar\'e lemma, any solution to $\bg{\nabla^i}(\rme^\varphi v_i)=0$ is locally of the form~\eqref{eq:velocity2Plectic} for some $\psi\in\scC^\infty(M^2)$. Hence, the $2$-plectic reduction of the Monge--Amp{\`e}re structure directly yields the two-dimensional fluid flow in terms of the stream function. Indeed, as already indicated, the symplectic reduction provides all of the geometric information that $2$-plectic reduction does, at least to the extent in which we are interested in this paper, thus enabling the analysis of singularities and curvature scalars as in two-dimensions. However, should one only require a description of the reduced kinematics, $k$-plectic reduction is certainly a more elegant, compact tool. 

        Before discussing specific examples, let us close this section by stating that the pull-back of the metric~\eqref{eq:fluidMetric3dReduced} along 
        \begin{equation}\label{eq:reductionPullback}
        	\tilde\iota\,:\,x^i\mapsto(x^i,q_i,q_3)\ \coloneqq\ \big(x^i,-\sqrt{\det(\bg{g}_2)}\,\rme^{-\varphi}\eps_{ij}\bg{g}^{jk}\partial_k\psi,v_3(x^i)\big)
        \end{equation}
        given by~\eqref{eq:2plecticInverseImageMomentMap} and~\eqref{eq:pullback2Plectic}, is
        \begin{subequations}\label{eq:pullbackFluidMetric3dReduced}
        	\begin{equation}
                g_2\ =\ \tfrac12\big(\bg\lap_{\rm B}\psi\bg{\nabla}_i\partial_j\psi+T_{ij}\big)\rme^{-2\varphi}\rmd x^i\odot\rmd x^j
        	\end{equation}
            with
        	\begin{equation}
        		\begin{aligned}
        			T_{ij}\ &\coloneqq\ \bg{g}_{ij}\big\{ \bg{\nabla}^l\varphi\,\partial_l\psi\big(\bg{\nabla}^k\varphi\, \partial_k\psi-\bg\lap_{\rm B}\psi\big)-\big(\bg{\nabla}^k\varphi\,\partial_k\varphi\big)\big(\bg{\nabla}^l\psi\,\partial_l\psi\big)
                    \\
                    &\kern1.5cm+\bg{\nabla}^k\varphi\big[\bg{\nabla}^l\psi\,\bg{\nabla}_k\partial_l\psi+v_3\big( \partial_kv_3 - v_3 \partial_k \varphi \big) \big]\big\}
                    \\
        			&\kern2cm+\partial_i\varphi\,\partial_j\varphi\big(\bg{\nabla}^k\psi\,\partial_k\psi\big)-\bg{\nabla}^k\psi\big[\partial_i\varphi\,\bg{\nabla}_j\partial_k\psi+\partial_j\varphi\,\bg{\nabla}_i\partial_k\psi\big]\,.
        		\end{aligned}
        	\end{equation}
        \end{subequations}
        Evidently, $T_{ij}=0$ when $\varphi=0$, in which case we recover the metric~\eqref{eq:fluidMetric2dPullBack}.
        
        Alternatively, we may write the above formula in such a way that the term $\hat f_2+\hat h_+$ remains explicit
        \begin{subequations}\label{eq:pullbackFluidMetric3dReducedV2}
        	\begin{equation}
        		g_2\ =\ \tfrac12g_{ij}\rmd x^i\odot\rmd x^j  
        	\end{equation}
        	with
        	\begin{equation}
        		\begin{aligned}
        			g_{ij}\ &=\ \tilde\iota^*(\hat f_2+\hat h_+)\bg g_{ij} 
        			\\
        			&\kern1cm+\rme^{-2\varphi}\big\{(\bg\nabla^k\partial_i\psi)(\bg\nabla_k\partial_j\psi)+(\partial_i\varphi)(\partial_j\varphi)(\bg\nabla^k\psi)(\partial_k\psi)-2(\partial_k\psi)[(\bg\nabla^k\partial_{(i}\psi)(\partial_{j)}\varphi)]\big\}\,.
        		\end{aligned}
        	\end{equation}        
        \end{subequations} 
        It follows from~\eqref{eq:reductionPullback} that $\tilde\iota^*\hat f_2=f_2$ if and only if $\varphi=0$. Again, we note from~\eqref{eq:pullbackFluidMetric3dReduced} and~\eqref{eq:pullbackFluidMetric3dReducedV2} that the pull-back metric is a quadratic function of the velocity gradients.
        
        \begin{remark}\label{rmk:ReducedPressure}
        	Recall that we define
        	\begin{equation}
        			\hat f_m\ \coloneqq\ \tfrac12(\bg{\lap}_{\rm B}p+\bg{R}_{ij}q^i q^j)~,     	
        	\end{equation}
        	where the differential operators are taken with respect to the metric $\bg{g}_m$ and $i,j=1,2,\ldots,m$, as in~\eqref{eq:PressureCurvature2d} and~\eqref{eq:PressureCurvature3d}. For fluid flows on three-dimensional background manifolds with warped-product metric~\eqref{eq:warpedProduct}, assuming both $p,v_3\in\scC^\infty(M^2)$, it follows that
        	\begin{equation}\label{eq:PressureCurvatureReduced}
        		\hat f_3\ =\ \hat f_2+\hat h_+~,
        	\end{equation}	 
        	with $\hat h_+$ as defined in~\eqref{eq:warpedReductionhath}. Hence, for three-dimensional flows with symmetry $\parder{x^3}$, the function $\hat f_2+\hat h_+$ should be interpreted as the diagnostic quantity $\hat f_3$, where the former representation highlights the deviation from the diagnostic quantity $\hat f_2$ for two-dimensional incompressible fluid flows.			
			
			In computing the pull-back of~\eqref{eq:PressureCurvatureReduced} along~\eqref{eq:reductionPullback},\footnote{This is equivalent to computing~\eqref{eq:pullBackHatf3D} for flows with symmetry $\parder{x^3}$.} similar representations of the traces of the squares of the vorticity two-form and the rate-of-strain tensor~\eqref{eq:vorticityAndStrain} are also enlightening. Let $\zeta^m_{ij}$ and $S^m_{ij}$ respectively denote the vorticity two-form and the rate-of-strain tensor in $m$ dimensions, where $i,j=1,2,\ldots m$ and the covariant derivatives occurring in~\eqref{eq:vorticityAndStrain} are understood to be with respect to $\bg{g}_m$. Then,			
        	\begin{subequations}
        		\begin{equation}\label{eq:ReducedVorticity}
        			\zeta_{IJ}^{3}\zeta^{IJ}_{3}\ =\ \zeta_{ij}^{2}\zeta^{ij}_{2}+\tfrac12(\partial_i v_3)(\bg \nabla^i v_3)\rme^{-2\varphi}
        		\end{equation}
                and
        		\begin{equation}\label{eq:ReducedStrain}
        			S_{IJ}^{3} S^{IJ}_{3}\ =\ S_{ij}^{2} S^{ij}_{2}+\rme^{-2\varphi}\big[\tfrac12(\partial_i v_3)(\bg{\nabla}^i v_3)-(\partial_i v_3)(\bg{\nabla}^i\varphi)v_3+(\partial_i\varphi)(\bg{\nabla}^i\varphi)v_3^2\big]+(v_i\bg{\nabla}^i\varphi)^2~,
        		\end{equation}
        	\end{subequations}
        	where the indices $I,J=1,2,3$ and $i,j=1,2$. Like $\hat f_2+\hat h_+$, these expressions do not depend on the coordinate $x^3$, however unlike $\hat f_2$, the quantities $\zeta^2_{ij}$ and $S^2_{ij}$ retain some dependence on the three-dimensional geometry via $\varphi$, since $v_i$ must satisfy~\eqref{eq:reducedKinematicSystem}.     	
        \end{remark}
        
        \begin{remark}
            The reduction presented explicitly in this work assumes a one-dimensional symmetry $\parder{x^3}$ of the underlying manifold $M^3$, dictated by complete $x^3$ independence of the velocity components. In contrast, by applying the above approach to cases where the symmetry lies in $\frX(T^*M^3)$, flows where $v_3(x)$ depends linearly on $x^3$ may also be considered. In particular, it was shown in~\cite{Banos:2015aa} that Burgers' vortex, a flow of the form~\eqref{eq:flow3d} with $W\equiv0$ and $\phi=\phi(t)$, has a symmetry generated by $\parder{x^3}+\phi\parder{q_3}$, hence admits a Hamiltonian reduction.     
        \end{remark}
        
        \subsection{Examples of higher symplectic reductions}\label{sec:examplesReduction}

        Let us now discuss a few examples of the reduction processes as outlined in \cref{sec:kprs}. 
        
        \paragraph{Arnol'd--Beltrami--Childress flow.}
		Let us consider flows on $M_3\coloneqq\IR^3$ equipped with the standard Euclidean metric	
		\begin{equation}
			\bg{g}_3\ \coloneqq\ \bg{g}_2+\rmd z\otimes\rmd z 
			\ewith
			\bg{g}_2\ \coloneqq\ \rmd x\otimes\rmd x+\rmd y\otimes\rmd y~,
		\end{equation}	
		which corresponds to the case when $\varphi=0$. Then, $\hat h_\pm=0$ and our symplectic reduction yields an incompressible fluid flow in two dimensions, on an Euclidean background. In summary, the equations~\eqref{eq:reducedMA2} reduce to $\tilde\varpi_2=\varpi_2$, and $\tilde\alpha_2=\alpha_2$, with the divergence-free constraint and the pressure equation~\eqref{eq:reducedKinematicSystem} respectively given by
		\begin{subequations}\label{eq:reducedEqABC}
        	\begin{equation}\label{eq:incompABC}
				\partial_xv_x+\partial_yv_y\ =\ 0
			\end{equation}
			and
			\begin{equation}\label{eq:pressureABC}
				\lap p\ =\ 2\big(\partial_xv_x\partial_yv_y-\partial_xv_y\partial_yv_x \big)
                \ewith
                \lap\ \coloneqq\ \partial_x^2+\partial_y^2~,
        	\end{equation}
        \end{subequations}
        where $v_x$ and $v_y$ are functions of $x$ and $y$ only.
        
        \begin{figure}[ht]
            \vspace{15pt}
            \begin{center}
                \begin{subfigure}[b]{0.425\textwidth}
                    \includegraphics[scale=0.201]{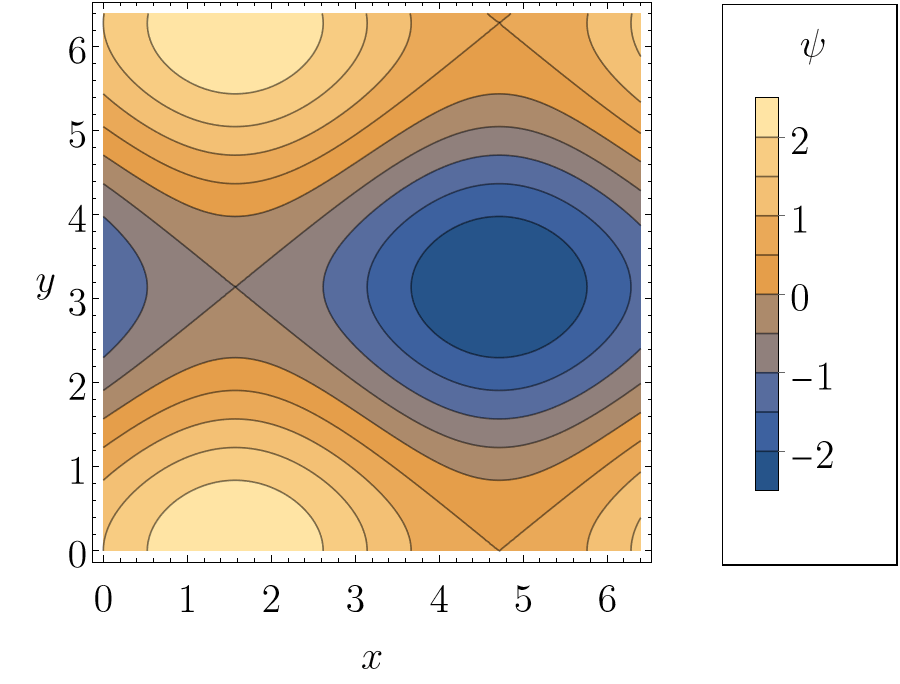}
                    \caption{Plot of the streamlines for $\psi$. The locus $\psi=0$ defines a shear layer between two homoclinic orbits, corresponding to vanishing vorticity $\zeta\coloneqq\lap\psi = -\psi$.}
                    \label{fig:plotStreamFunctionABC}
                \end{subfigure}
                \hspace{25pt}
                \begin{subfigure}[b]{0.425\textwidth}
                    \includegraphics[scale=0.201]{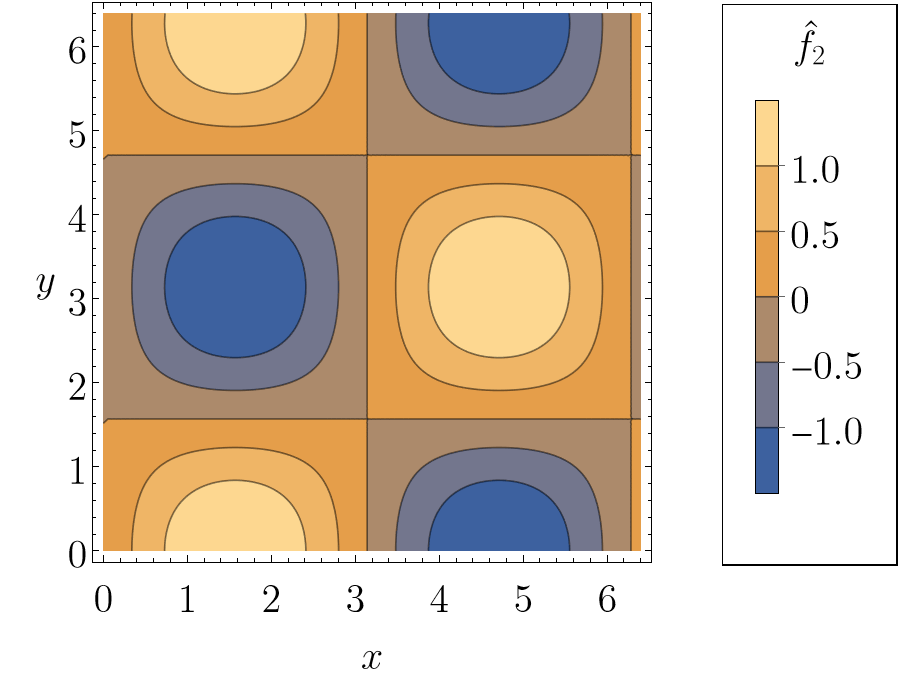}
                    \caption{Contour plot for $\hat f_2$. The domain is partitioned into squares of side length $\pi$, across which the sign of $\lap p$ alternates. \\}
                    \label{fig:pressureLapABC}
                \end{subfigure}
                \caption{Plots of the iso-lines of the stream function~\eqref{eq:streamFunctionABC} and reduced Laplacian of pressure~\eqref{eq:pressureLapABC} for an integrable Arnol'd--Beltrami--Childress flow with parameters $A=1.5$ and $B=1$.}
            \end{center}
            \vspace{-15pt}
        \end{figure}
		
		Additionally, performing the $2$-plectic reduction to obtain velocity components $v_x$ and $v_y$ satisfying~\eqref{eq:incompABC}, in terms of a stream function in two dimensions, yields the same result as applying the Poincar\'e lemma to~\eqref{eq:incompABC} itself, that is, 
		\begin{equation}\label{eq:velocitiesABC}
        	q_x\ \coloneqq\ v_x\ =\ -\partial_y\psi
            \eand
        	q_y\ \coloneqq\ v_y\ =\ \partial_x\psi
        \end{equation}	
        for some stream function $\psi=\psi(x,y)$. The corresponding differential form~\eqref{eq:2plecticReducedForm} is 
        \begin{equation}
        	\alpha_{-\rmd\psi}\ =\ \big[\partial_x^2\psi\,\partial_y^2\psi-(\partial_x\partial_y\psi)^2-\tfrac12\lap p\big]\,\rmd x\wedge\rmd y~.
        \end{equation}
		This is unchanged when pulled back along $(x,y)\mapsto(x,y,q_z)\coloneqq(x,y,v_z(x,y))$, so imposing a vanishing pull-back condition is equivalent to the Monge--Amp{\`e}re equation
		\begin{equation}\label{eq:pressureStreamABC}
			\tfrac12\lap p\ =\ \partial_x^2\psi\,\partial_y^2\psi-(\partial_x\partial_y\psi)^2~,
		\end{equation}
		which is, in turn, precisely~\eqref{eq:pressureABC} with $v_x$ and $v_y$ evaluated as per~\eqref{eq:velocitiesABC}. Hence, one is free to choose a pair of $z$-independent functions $\psi$ and $v_z$ in order to recover an incompressible fluid flow in $\IR^3$ that reduces to an incompressible flow on the $(x,y)$-plane.
		
        \begin{figure}[ht]
            \begin{center}
                \begin{subfigure}[b]{0.425\textwidth}
                    \includegraphics[scale=0.201]{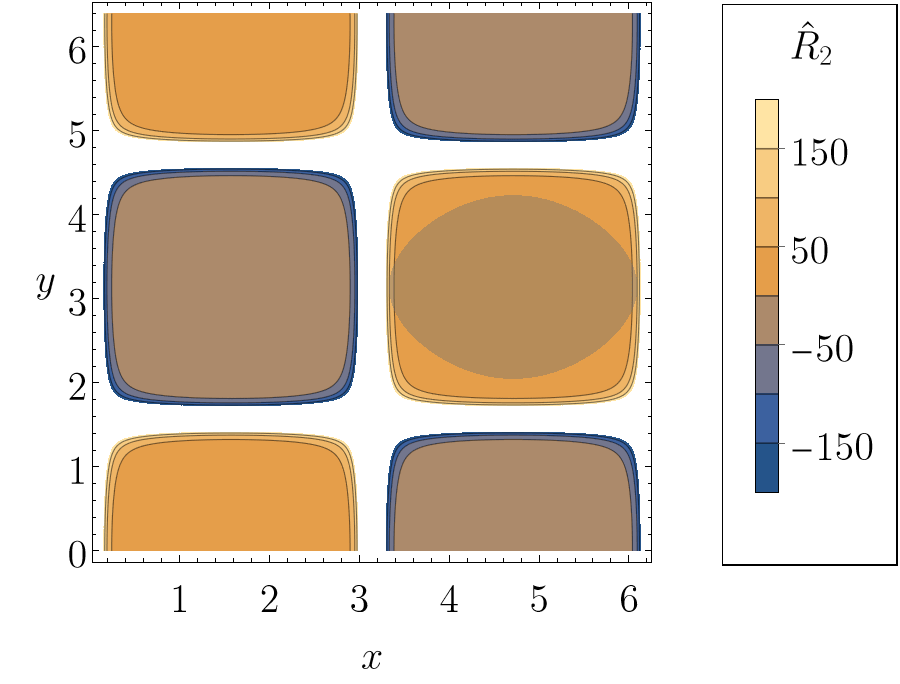}
                    \caption{Contour plot for the curvature scalar $\hat R_2$. Both $\hat R_2$ and $\hat f_2$ have the same signs and $\hat R_2$ blows up as $\hat f_2$ tends to zero.}
                    \label{fig:curvatureLRvortexABC}
                \end{subfigure}
                \hspace{25pt}
                \begin{subfigure}[b]{0.425\textwidth}
                    \includegraphics[scale=0.201]{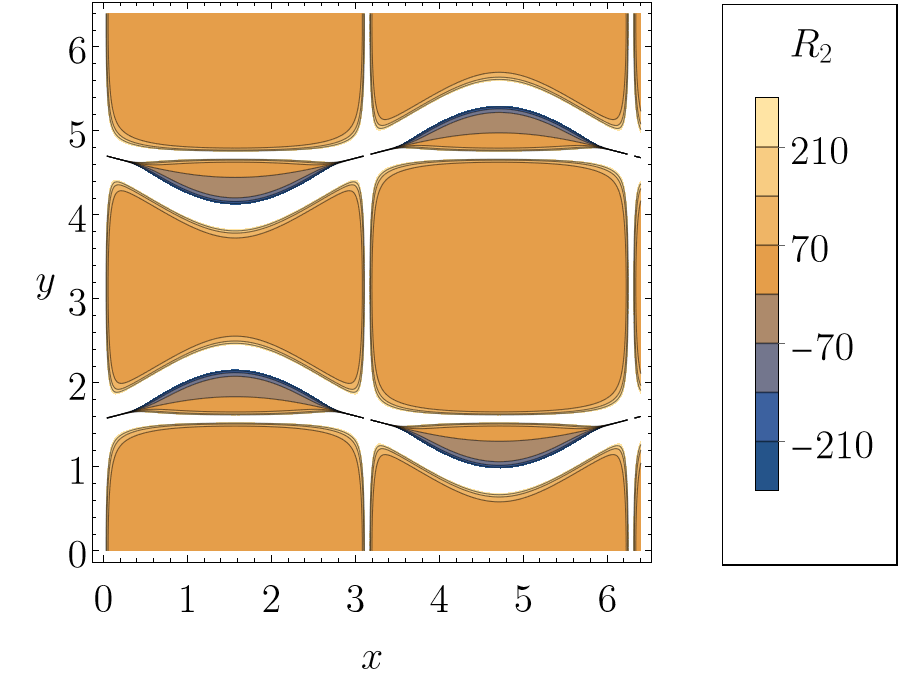}
                    \caption{Contour plot for the curvature scalar $R_2$. Note that curvature singularities occur when $\hat f_2=0$. }
                    \label{fig:curvature2DABC}
                \end{subfigure}
                \caption{Contour plots of the curvatures~\eqref{eq:4DCurvatureABC} (left) and~\eqref{eq:curvature2DABC} (right) respectively, for the Arnol'd--Beltrami--Childress flow with parameters $A=1.5$ and $B=1$. The ellipse highlighted on the left is the domain bounded by the closed streamline $\psi=-\frac{27}{16}$, which  is contained in a region on which the metrics $\hat g_2$ and $g$ are Riemannian, and $\hat f_2>0$.}
            \end{center}
        \end{figure}  		
		
		Making the choice
		\begin{equation}\label{eq:streamFunctionABC}
			v_z(x,y)\ =\ \psi(x,y)\ \coloneqq\ A\cos(y)+B\sin(x)
		\end{equation}
        for $A,B\in\IR$ some constants, see \cref{fig:plotStreamFunctionABC}, and computing~\eqref{eq:velocitiesABC}, we recover the velocity field for the integrable case of \uline{Arnol'd--Beltrami--Childress flow}~\cite{Dombre:1986aa},
		\begin{equation}\label{eq:velocityABC}
			(v_x,v_y,v_z)\ =\ (\dot x,\dot y,\dot z)\ =\ \big(A\sin(y),B\cos(x),A\cos(y)+B\sin(x)\big)\,.
		\end{equation}

        Next, following~\cite{Dombre:1986aa}, upon taking the quotient of $v_x$ and $v_y$, this system integrates to $v_z=A\cos(y)+B\sin(x)=\text{const}$. Furthermore,~\eqref{eq:pressureStreamABC} becomes
        \begin{equation}\label{eq:pressureLapABC}
            \hat f_2\ =\ \tfrac12\lap p\ =\ AB\sin(x)\cos(y)~,
        \end{equation}
        and this is displayed in \cref{fig:pressureLapABC}.

		Since $\hat h_+=0$ and $M^2=\IR^2$, it follows that the metric~\eqref{eq:fluidMetric3dReduced} on the reduced phase space $\mu^{-1}(\{c\})/\sfG_c\cong T^*\IR^2$ is precisely~\eqref{eq:flatBack2DMetric}. Hence, we may follow exactly the treatment from \cref{sec:examples2D}. Therefore, the curvature scalar $\hat R_2$ for the metric~\eqref{eq:fluidMetric3dReduced} follows directly from~\eqref{eq:flatBackCurvatureLR2D},
		\begin{equation}\label{eq:4DCurvatureABC}
            \hat R_2\ =\ \frac{\sin^2(x)+\cos^2(y)}{AB\sin^3(x)\cos^3(y)}~,
		\end{equation}
		and as in previous examples, for $\hat f_2\gtrless 0$ the metric $\hat g_2$ is Riemannian/Kleinian and the associated curvature is positive/negative. Again, when $\hat f_2=0$, both the metric and the curvature scalar are singular.

		In turn, the pull-back metric~\eqref{eq:pullbackFluidMetric3dReduced}, with $v_x$ and $v_y$ as given in~\eqref{eq:velocityABC}, is 
		\begin{equation}\label{eq:metricPullbackABC}
            g_2\ =\ [A\cos(y)+B\sin(x)]
            \begin{pmatrix}
                B\sin(x) & 0
                \\
                0 & A\cos(y)
            \end{pmatrix},
        \end{equation}
        where the vorticity is $\zeta=-A\cos(y)-B\sin(x)$. This metric is again singular when $\hat f_2=0$, however it also exhibits a further singularity when $A\cos(y)+B\sin(x)=0$, which is precisely the shear layer featuring in the streamlines of~\cref{fig:plotStreamFunctionABC} and which also corresponds to vanishing vorticity. The curvature scalar $R_2$ associated with~\eqref{eq:metricPullbackABC} is then
        \begin{equation}\label{eq:curvature2DABC}
        	R_2\ =\ \frac{B\sin(x)\big[\sin^2(x)+3\cos^2(y)\big]+A\cos(y)\big[\cos^2(y)+3\sin^2(x)\big]}{2\sin^2(x)\cos^2(y)[B\sin(x)+A\cos(y)]^3}~.
        \end{equation}

		The lines $x=n\pi$ and $y=\big(n+\frac12\big)\pi$ for all $n\in\IZ$, along which $f_2=0$, are singularities of both the metric $g$ and its curvature $R$, as was the case for the metric~\eqref{eq:fluidMetric3dReduced}. Additionally, the presence of $A\cos(y)+B\sin(x)$ in the denominator illustrates that the shear layer is a curvature singularity. See \cref{fig:curvature2DABC}. This curvature singularity arises due to the vanishing vorticity and is otherwise unseen by the pressure criterion. The shear layer is a separatrix between topologically distinct flows.

        \begin{figure}[ht]
            \vspace{15pt}
            \begin{center}
                \begin{subfigure}[b]{0.425\textwidth}
                        \includegraphics[scale=0.201]{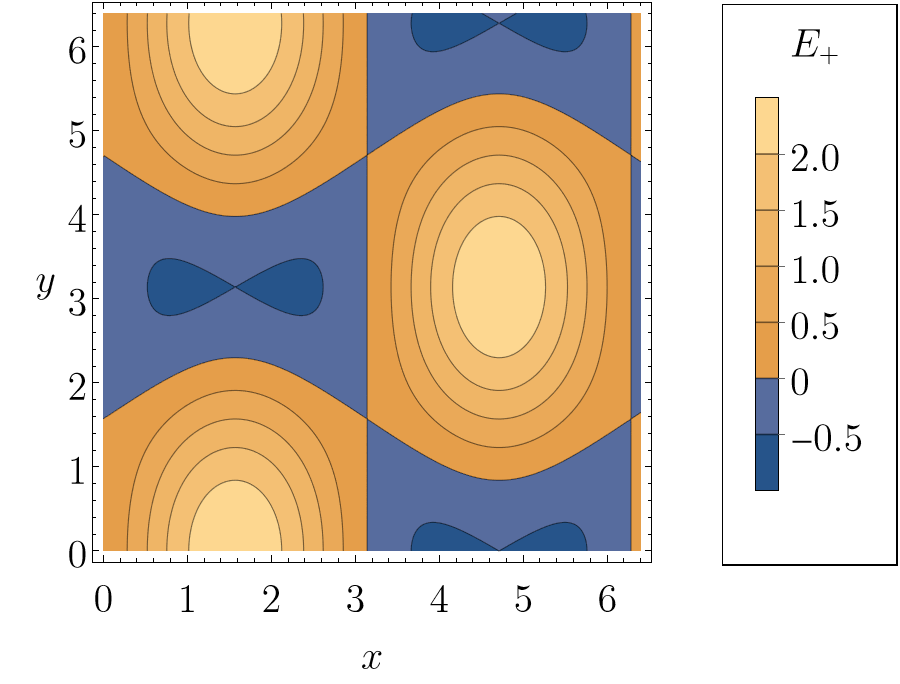}
                        \caption{Contour plot for the eigenvalue $E_+$, which vanishes both along the shear layer and along $x=\pi$.}
                \end{subfigure}
                \hspace{25pt}
                \begin{subfigure}[b]{0.425\textwidth}
                        \includegraphics[scale=0.201]{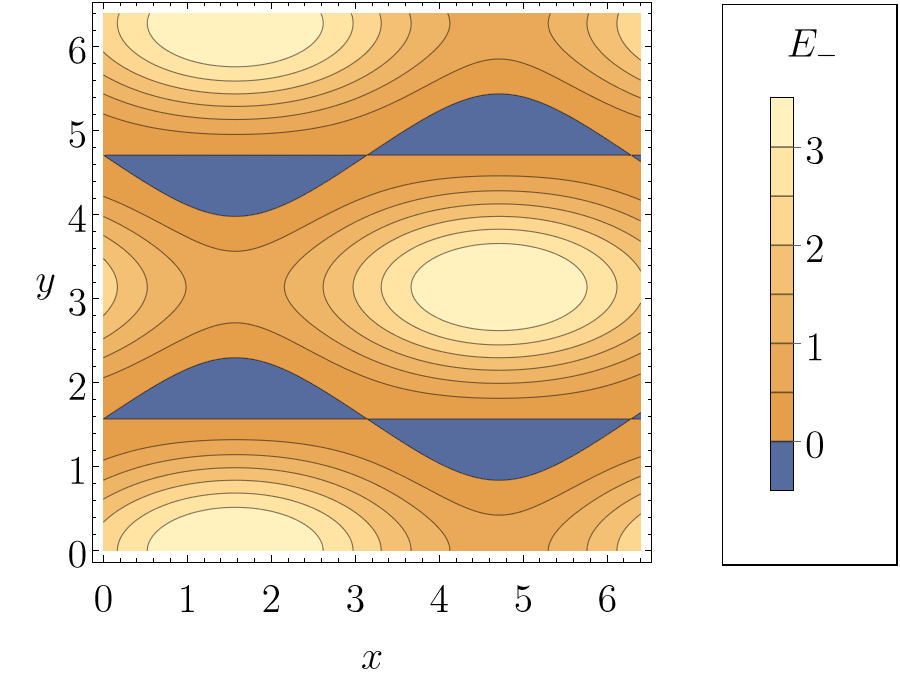}
                        \caption{Contour plot for the eigenvalue $E_-$. In addition to the shear layer, $E_-$ also vanishes along $y=\frac\pi2$ and $y=\frac{3\pi}{2}$.}       
                \end{subfigure}
                \caption{Plots of the eigenvalues of the pull-back metric~\eqref{eq:metricPullbackABC} with $g_2=\diag(E_+,E_-)$ of the Arnol'd--Beltrami--Childress flow with parameters $A=1.5$ and $B=1$. The signs of both eigenvalues change across the shear layer, where the vorticity prefactor changes sign, hence the signature of the metric is unchanged across this singularity.}
            \end{center}
        \end{figure}

        \paragraph{Hicks--Moffatt vortex.}
        We now discuss another important class of examples --- vortices of \uline{Hicks--Moffatt type}~\cite{Hicks:1899aa, Moffatt:1969aa}. Consider flows on $M\coloneqq (\IR^+ \times \IR)\times_{r^2} S^1$ equipped with
        \begin{equation}\label{eq:metricCylindricalCoordinates}
            \bg{g}_3\ \coloneqq\ \bg{g}_2+r^2\rmd\theta\otimes\rmd\theta
            \ewith
            \bg{g}_2\ \coloneqq\ \rmd r\otimes\rmd r+\rmd z\otimes\rmd z~,
        \end{equation}
        where $r\in\IR^+$, $z\in\IR$, and $\theta\in[0,2\pi)$, that is, standard cylindrical coordinates. Then, 
        \begin{equation}
        	\varphi\ =\ \log(r) \eand \hat h_+\ =\ \tfrac{1}{2r}\partial_rp
        \end{equation}
        with $p=p(r,z)$ in~\eqref{eq:warpedReduction}. Hence, the equations~\eqref{eq:reducedMA2} reduce to
        \begin{equation}
            \begin{gathered}
                \tilde\varpi_2\ =\ r\big(\varpi_2+\tfrac1rq_r\,\rmd r\wedge\rmd z\big)\,,
                \\
                \tilde\alpha_2\ =\ r\big\{\alpha_2-\big[\tfrac{1}{2r}\partial_rp+\tfrac{1}{r^2}\big(\tfrac{1}{r^2}q_\theta^2-\tfrac{1}{r}q_\theta\partial_rq_\theta\big)\big]\,\rmd r\wedge\rmd z+\tfrac{1}{r}q_r\varpi_2\big\}\,.
            \end{gathered}
        \end{equation}
        Furthermore, the requirements that the pull-backs of $\tilde\varpi_2$ and $\tilde\alpha_2$ under~\eqref{eq:incompressibilityFromLagrangian} vanish become
        \begin{subequations}\label{eq:reducedEqHills}
        	\begin{equation}\label{eq:incompHills}
				\tfrac1r\partial_r(rv_r)+\partial_zv_z\ =\ 0~,
			\end{equation}
			and
			\begin{equation}\label{eq:pressureHills}
				\tfrac1r\partial_r(r\partial_rp)+\partial_z^2p\ =\ 2\big[\partial_rv_r\partial_zv_z-\partial_rv_z\partial_zv_r-\tfrac{1}{r^2}v_r^2-\tfrac{1}{r^4}\big(v_\theta^2-\tfrac{r}{2}\partial_rv_\theta^2\big)\big]\,,
        	\end{equation}
        \end{subequations}
        which are the equations~\eqref{eq:reducedKinematicSystem} for the metric~\eqref{eq:metricCylindricalCoordinates}, with $v_\theta=v_\theta(r,z)$ arbitrary. Evidently, the first equation is simply the divergence of $v$ for such a $v_\theta$ and the left hand side of the second equation is the Laplacian of $p=p(r,z)$, both expressed in cylindrical polar coordinates.
        
        Turning now to the $2$-plectic reduction, note that we can take the moment map~\eqref{eq:2PlecticMoment} to be 
        \begin{equation}
        	\mu(x,q)\ =\ rq_r\rmd z-rq_z\rmd r~.
        \end{equation} 
        It then follows that, locally on $\mu^{-1}(\{-\rmd\psi\})/\sfG_{-\rmd\psi}$, we have
        \begin{equation}\label{eq:velocitiesHills}
    		q_r\ \coloneqq\ v_r\ =\ -\tfrac1r\partial_z\psi
            \eand
    		q_z\ \coloneqq\ v_z\ =\ \tfrac1r\partial_r\psi~,
        \end{equation}
        which can be interpreted as expressions for the velocity components in the $r$ and $z$ directions, in terms of a stream function $\psi=\psi(r,z)$ in two dimensions. Consequently, these solve the adapted divergence-free constraint~\eqref{eq:incompHills}. In fact, imposing that the pull-back of the closed differential form~\eqref{eq:2plecticReducedForm} along $(r,z)\mapsto(r,z,q_\theta)\coloneqq(r,z,v_\theta(r,z))$ vanishes, we find 
        \begin{equation}\label{eq:pressureStreamHills}
        	\begin{aligned}
            \tfrac12\big[\tfrac1r\partial_r(r\partial_rp)+\partial_z^2p\big]\ &=\ \tfrac{1}{r^2}\big[\partial_r^2\psi\partial_z^2\psi-(\partial_r\partial_z\psi)^2\big]-\tfrac{1}{r^4}(\partial_z\psi)^2
            \\
            &\kern1.5cm +\tfrac{1}{r^3}\big(\partial_z\psi\partial_r\partial_z\psi-\partial_r\psi\partial_z^2\psi\big)-\tfrac{1}{r^4}\big(v_\theta^2-\tfrac{r}{2}\partial_rv_\theta^2\big)\,,
            \end{aligned}
        \end{equation}
        that is,~\eqref{eq:reducedEqHills} with $v_r$ and $v_z$ given in terms of $\psi$ as in~\eqref{eq:velocitiesHills}. One is free to choose $\psi$ and $v_\theta$, provided they satisfy~\eqref{eq:pressureStreamHills}. Furthermore,~\eqref{eq:incompHills} is trivially satisfied for any such choices, given~\eqref{eq:velocitiesHills}.
        
        In what follows, we fix of $\psi$ and $v_\theta$ corresponding to vortices of Hicks--Moffatt type. In particular, we shall discuss a class of spherical vortices with swirl parameter $\kappa$, normalising the radius of the sphere to $1$ for convenience. For an in-depth review of such vortices, we direct the interested reader to~\cite{Abe:2022aa}.

        \begin{figure}
            \begin{center}
                \begin{subfigure}[b]{0.425\textwidth}
                    \includegraphics[scale=0.201]{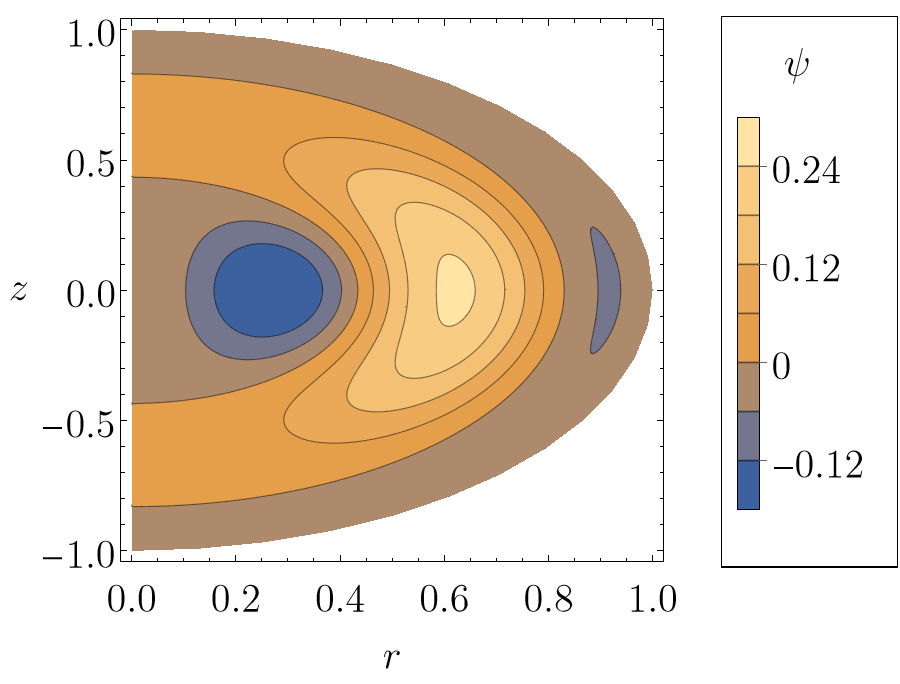}
                    \caption{A plot of the streamlines for~\eqref{eq:hicksIntStream}. The increased helicity (knotting) of the vortex lines results in layers of closed contours of alternating sign.}
                \end{subfigure}
                \hspace{25pt}
                \begin{subfigure}[b]{0.425\textwidth}
                    \includegraphics[scale=0.201]{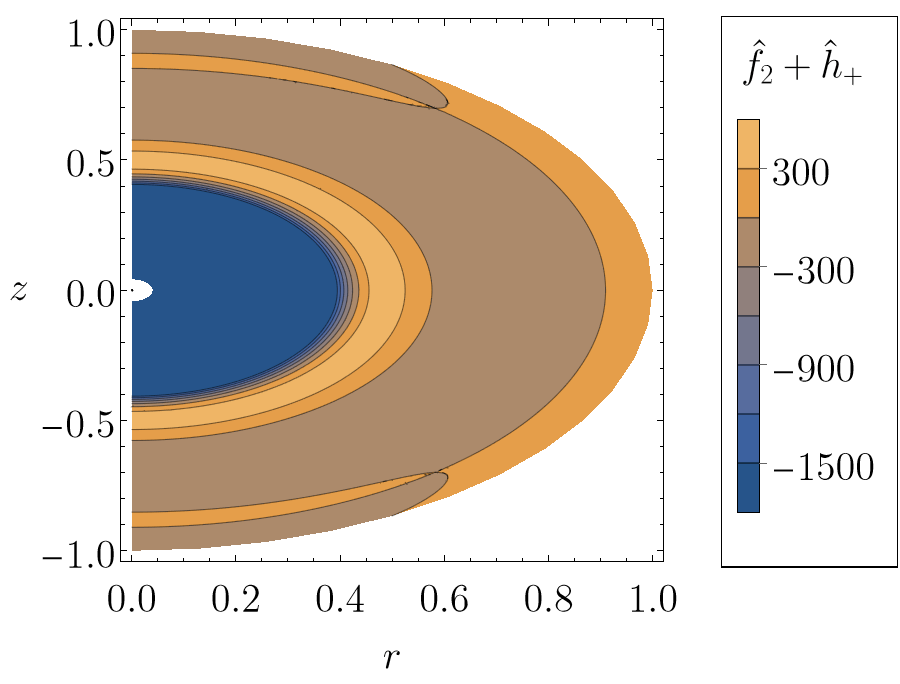} 
                    \caption{A plot of the Laplacian of pressure function $\hat f_2+\hat h_+$, as discussed in~\cref{rmk:ReducedPressure}, for the stream function $\psi_{\rm{int},10}$. \\} 
                    \label{subfig:HicksMoffatPressure}
                \end{subfigure}
                \begin{subfigure}[b]{0.425\textwidth}
                    \vspace{20pt}
                    \includegraphics[scale=0.181]{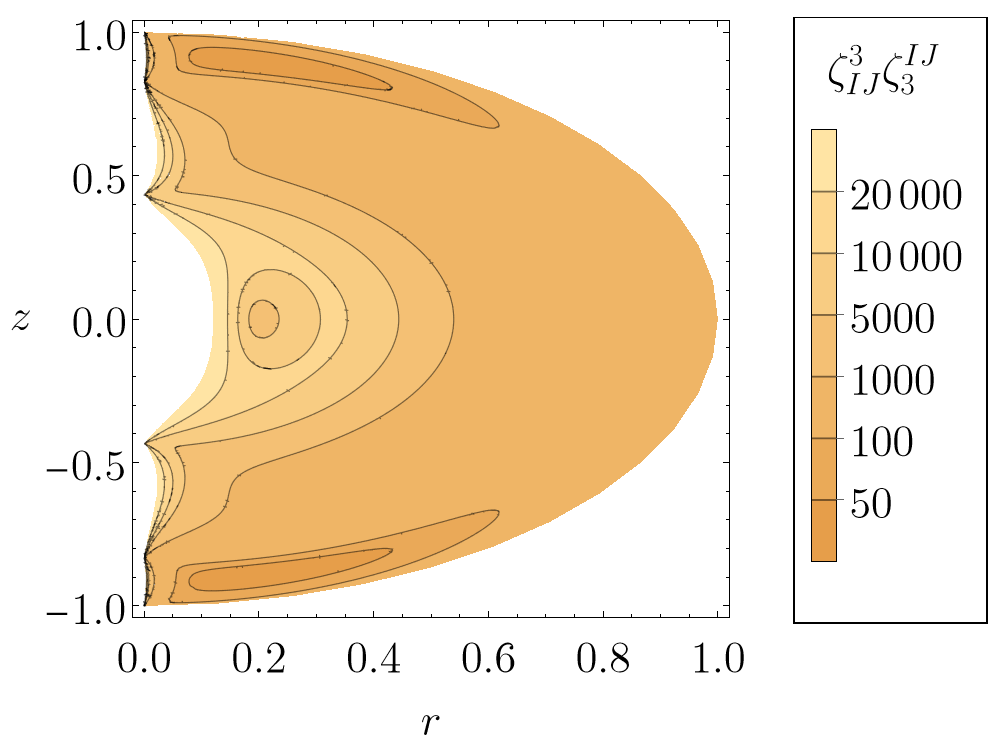}
                    \caption{A plot of the iso-lines of~\eqref{eq:ReducedVorticity}, for the Hicks--Moffatt vortex with $\kappa=10$. Note that this quantity is independent of $\theta$.}
                    \label{subfig:HicksMoffatVorticity}
                \end{subfigure}
                \hspace{25pt}
                \begin{subfigure}[b]{0.425\textwidth}
                    \includegraphics[scale=0.181]{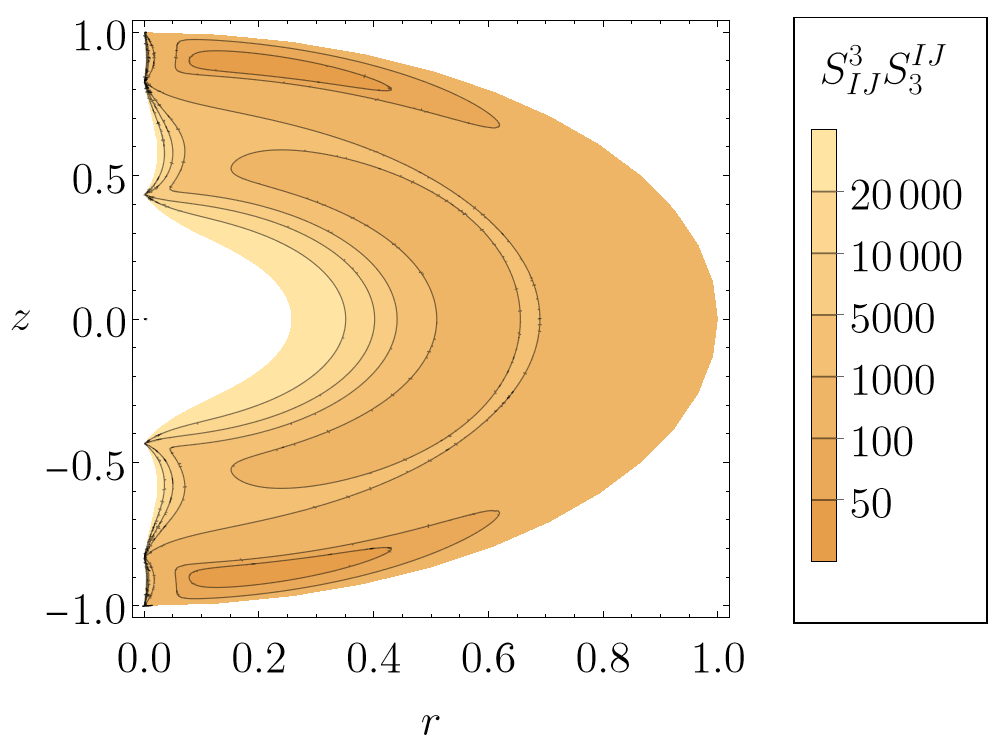}
                    \caption{A plot of the iso-lines of~\eqref{eq:ReducedStrain} for the Hicks--Moffatt vortex with $\kappa=10$. Note that this quantity is independent of $\theta$.}
                    \label{subfig:HicksMoffatStrain}
                \end{subfigure}
                \caption{A selection of plots for the Hicks--Moffatt vortex interior solution with swirl parameter $\kappa=10$. Since $\hat f_3 = \hat f_2 + \hat h_+$ in this case, the quantity shown in \cref{subfig:HicksMoffatPressure} is precisely the difference between those shown in \cref{subfig:HicksMoffatVorticity} and \cref{subfig:HicksMoffatStrain}.}
                \label{fig:HicksMoffatInterior}
            \end{center}
        \end{figure} 

		Firstly, consider a unit sphere in $\IR^3$ and set $\sigma(r,z)\coloneqq\sqrt{r^2+z^2}$ in cylindrical polar coordinates, as above. Fix the angular velocity to be
		\begin{equation}\label{eq:hicksAngleVelocity}
			v_{\theta,\kappa}(r,z)\ =\ \frac{\kappa\psi}{r}~,
		\end{equation} 
		on the whole domain. We then fix the stream function on the interior and exterior of the sphere, such that they coincide on the boundary. In particular, on the interior we set
        \begin{subequations}\label{eq:hicksIntStream}
    		\begin{equation}
    			\psi_{\rm{int},\kappa}(r,z)\ \coloneqq\ \frac32r^2\left(b(\kappa)-c(\kappa)\frac{J_{\frac32}(\kappa\sigma)}{(\kappa\sigma)^{\frac32}}\right),
    		\end{equation}
    		with 
    		\begin{equation}
    			b(\kappa)\ \coloneqq\ \frac{J_{\frac32}(\kappa)}{\kappa J_{\frac52}(\kappa)}
    			\eand
    			c(\kappa)\ \coloneqq\ \frac{\sqrt{\kappa}}{J_{\frac52}(\kappa)}~,
    		\end{equation}
        \end{subequations}
		where $J_n(x)$ is the $n$-th order Bessel function with argument $x$.\footnote{Such an explicit solution was found in the context of magneto-hydrodynamics~\cite{Prender:1956aa}. In the context of Navier--Stokes, solutions to this type are also referred to as Hill's spherical vortex with swirl.} See \cref{fig:HicksMoffatInterior} for plots of the interior solution. On the exterior of the sphere, one chooses the stream function to be independent of the swirl parameter and to match the interior solution on the boundary of the sphere, given by $\sigma^2=1$; for example, we choose 
		\begin{equation}\label{eq:hicksExtStream}
			\psi_{\rm{ext}}(r,z)\ \coloneqq\ \frac12r^2\left(1-\frac{1}{\sigma^3}\right),
		\end{equation}	
        so that the flow far from the sphere is uniform with unit speed directed along the $z$ axis, with non-zero velocity in the $\theta$-direction when $\kappa\neq0$. It is important to observe that the helicity~\eqref{eq:helicity} is non-zero if and only if the flow has non-zero swirl~\cite{Bannikova:2016aa,Tsinober:1992aa}. 
		
		\begin{figure}[ht]
            \vspace{15pt}
            \begin{center}
                \begin{subfigure}[b]{0.425\textwidth}
                    \includegraphics[scale=0.201]{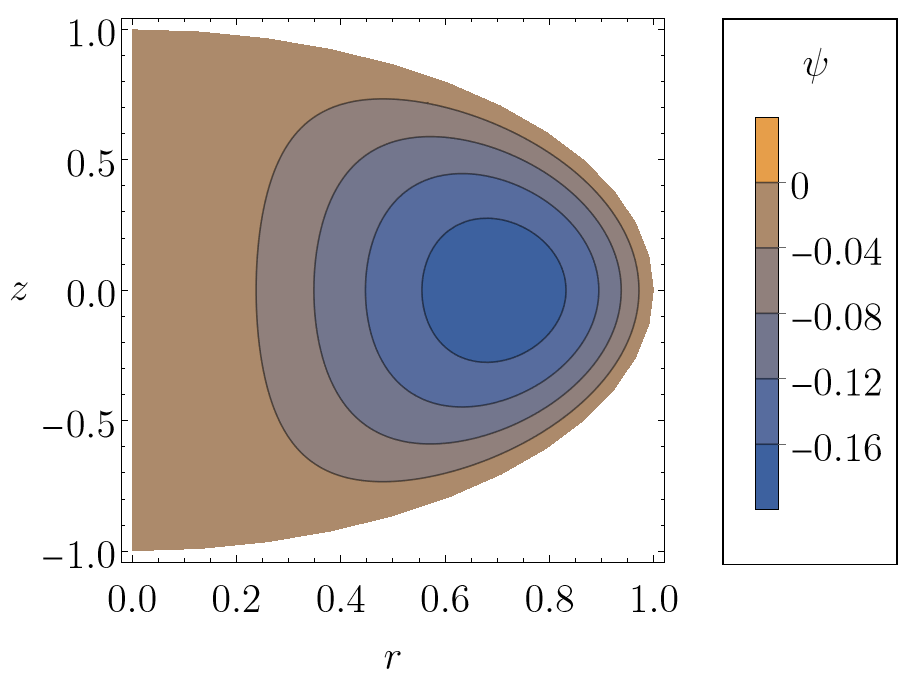}
                    \caption{Contours for $\psi_{\rm int,0}$~\eqref{eq:hillInteriorStreamFunction}. The contours are closed and concentric, forming toroidal vortex tubes when rotated around the $z$ axis to form our three-dimensional flow.}
                    \label{fig:interiorStreamFunctionContours} 
                \end{subfigure}
                \hspace{25pt}
                \begin{subfigure}[b]{0.425\textwidth}
                    \includegraphics[scale=0.201]{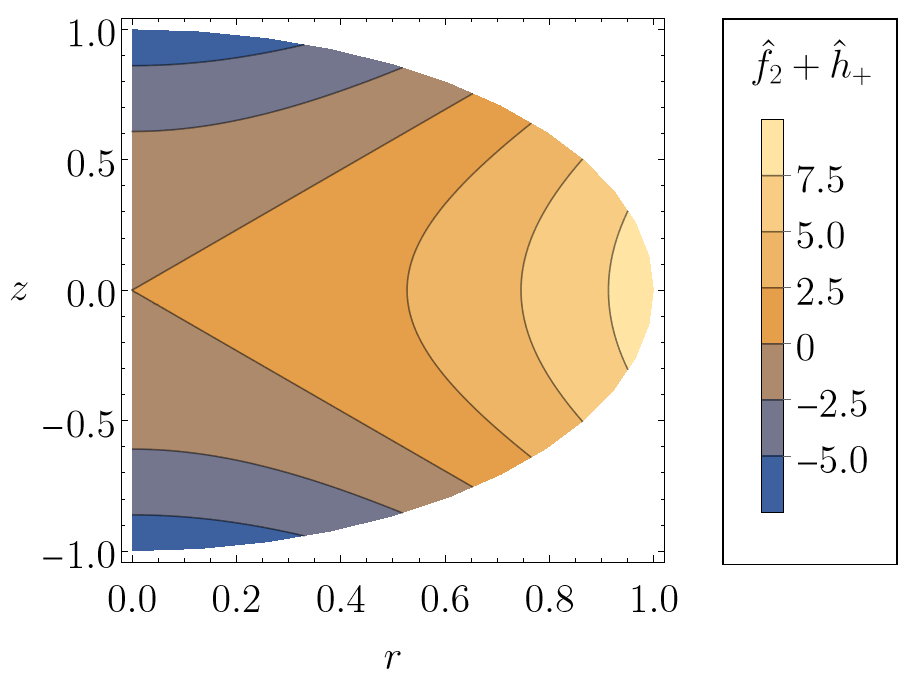}
                    \caption{$\hat f_2+\hat h_+$ on the interior of the unit sphere, with stream function~\eqref{eq:hillInteriorStreamFunction}. The function vanishes along $4r^2=3z^2$, is positive between these curves and negative outside of them.}
                    \label{fig:hillInteriorConformalFactor} 
                \end{subfigure}
                \caption{Plots of the iso-lines of the stream function~\eqref{eq:hillInteriorStreamFunction} and the function~\eqref{eq:hillInteriorConformalFactor} respectively, for the interior of Hill's spherical vortex.}
            \end{center}
            \vspace{-10pt}
        \end{figure}
		
		Let us now focus on the limiting case $\kappa=0$ when the flow has vanishing helicity. This corresponds to \uline{Hill's spherical vortex}~\cite{Hill:1894aa}. The $\theta$-component of velocity,~\eqref{eq:hicksAngleVelocity} then becomes $v_{\theta,0}(r,z)=0$. Firstly, note that for the exterior solution,~\eqref{eq:hicksExtStream} remains the same. Henceforth, we focus our attention on the interior alone. The interior solution, for which the stream function~\eqref{eq:hicksIntStream} is given by
		\begin{equation}\label{eq:hillInteriorStreamFunction}
        	\psi_{\rm{int},0}(r,z)\ \coloneqq\ \tfrac34r^2\big(r^2+z^2-1\big)\,.
        \end{equation}
		See \cref{fig:interiorStreamFunctionContours}. 

        Upon applying~\eqref{eq:incompHills}, we obtain the velocity components
        \begin{equation}
            v_r\ =\ -\tfrac32rz
            \eand
            v_z\ =\ \tfrac32\big(2r^2+z^2-1\big)\,,
        \end{equation}      
        Then, it follows from~\eqref{eq:pressureStreamHills} that the Laplacian of pressure is given by 
        \begin{equation}\label{eq:hillInteriorConformalFactor}
            \hat f_2+\hat h_+\ =\ \tfrac12\big(\partial_r^2p+\partial_z^2p+\tfrac1r\partial_rp\big)\ =\ \tfrac94\big(4r^2-3z^2\big)\,.
        \end{equation}
        See \cref{fig:hillInteriorConformalFactor}.   

        \begin{figure}[ht]
            \begin{center}
                \begin{subfigure}[b]{0.425\textwidth}
                    \includegraphics[scale=0.201]{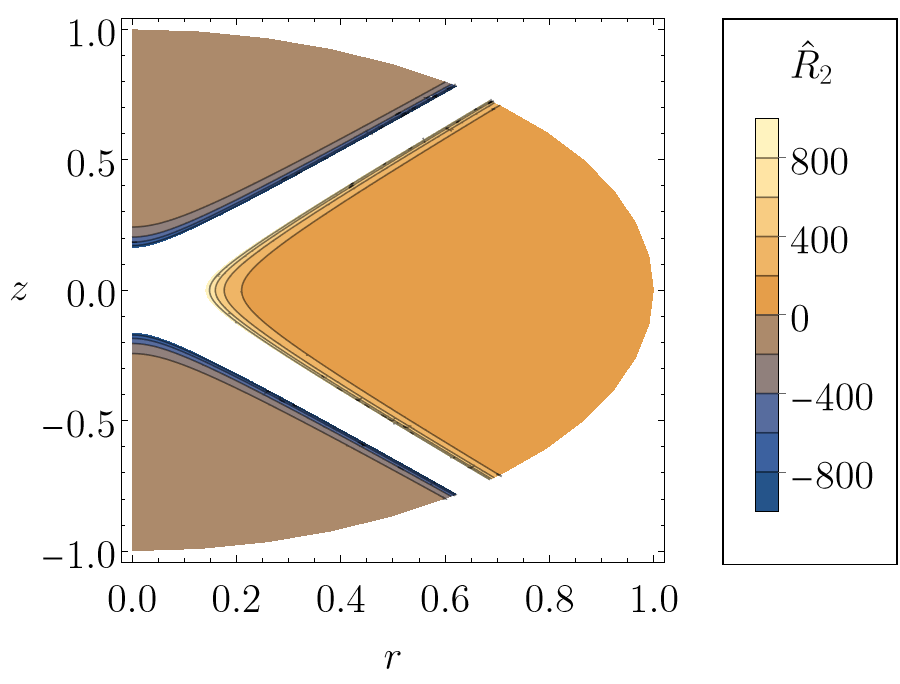}
                    \caption{A plot of $\hat R_2$ for the stream function~\eqref{eq:hillInteriorStreamFunction}. The curvature decreases in magnitude towards the boundary of the sphere and is singular along $4r^2=3z^2$.}
                    \label{fig:curvature4DHV} 
                \end{subfigure}
                \hspace{25pt}
                \begin{subfigure}[b]{0.425\textwidth}
                    \includegraphics[scale=0.201]{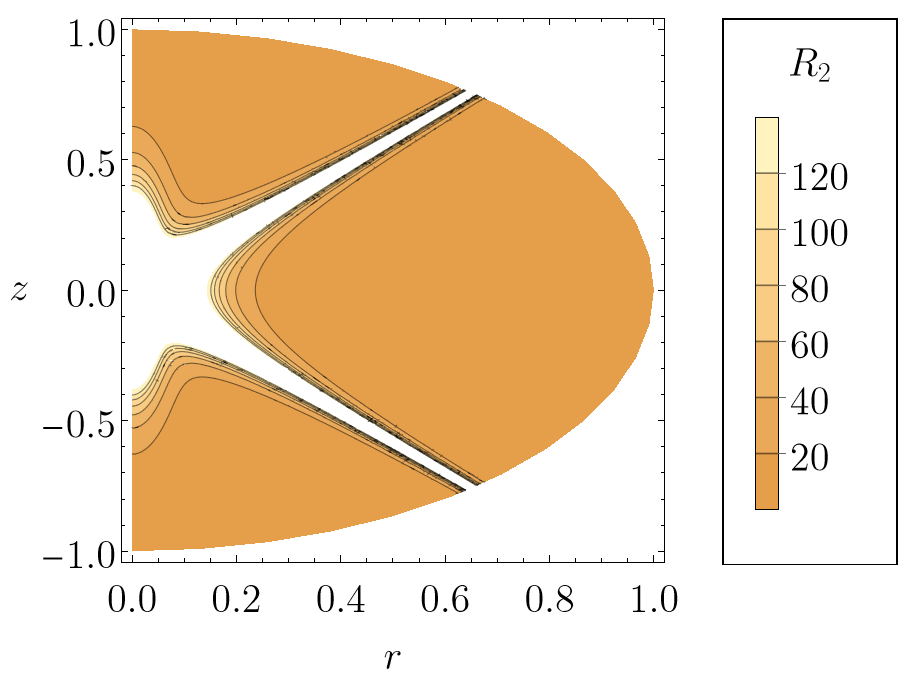}
                    \caption{A plot of $R_2$ for the stream function~\eqref{eq:hillInteriorStreamFunction}. This is singular along $100r^4-71r^2z^2-2z^4=0$ where $E_-=0$ and hence where the metric $g_2$ is degenerate.}
                    \label{fig:curvatureR2InteriorHill} 
                \end{subfigure}
                \caption{Contour plots of the curvatures~\eqref{eq:curvature4DHV} (left) and~\eqref{eq:curvatureHV} (right) respectively, for the interior of Hill's spherical vortex. Note that the curvature singularities in these two plots do not coincide, in contrast to earlier examples.}
            \end{center}
            \vspace{-10pt}
        \end{figure} 

        The metric~\eqref{eq:fluidMetric3dReduced} takes the form
        \begin{equation}\label{eq:hicksLRMetric}
            \hat g_2\ =\ 
            \begin{pmatrix}
                (\hat f_2+\hat h_+)\unit_2 & 0
                \\
                0 & \unit_2
            \end{pmatrix}.
        \end{equation}
        and $\hat R_2$ is given by~\eqref{eq:flatBackCurvatureLR2D}, with $f$ replaced by $\hat f_2+\hat h_+$. Namely
        \begin{equation}\label{eq:curvature4DHV}
        	\hat R_2\ =\ \frac{56\big(4r^2+3z^2\big)}{9\big(4r^2-3z^2\big)^3}~,
        \end{equation}
        which is plotted in \cref{fig:curvature4DHV}. When $4r^2>3z^2$, then $\hat f_2+\hat h_+>0$, the metric is Riemannian, and the curvature scalar is positive. Similarly, the metric is Kleinian and the curvature scalar negative when $\hat f_2+\hat h_+<0$ and $4r^2<3z^2$. Furthermore, the metric is singular when $4r^2=3z^2$, that is, when $\hat f_2+\hat h_+=0$ and it is clear that this singularity is also one for the curvature.        

        The pull-back metric~\eqref{eq:pullbackFluidMetric3dReduced} becomes
		\begin{equation}\label{eq:pullbackMetricHV}
			g_2\ =\ \frac94
            \begin{pmatrix}
                20r^2-2z^2 & 9rz
                \\
                9rz & 5r^2+z^2
            \end{pmatrix}.
		\end{equation}
		Its eigenvalues, displayed in \cref{fig:eigenvaluesHV}, are given by
		\begin{equation}\label{eq:eigenvaluesHV}
			E_\pm\ =\ \tfrac98\Big(25r^2-z^2\pm3\sigma\sqrt{(25r^2+z^2)}\Big)\,.
		\end{equation}
		
		 \begin{figure}[ht]
            \begin{center}
                \begin{subfigure}[b]{0.425\textwidth}
                    \includegraphics[scale=0.201]{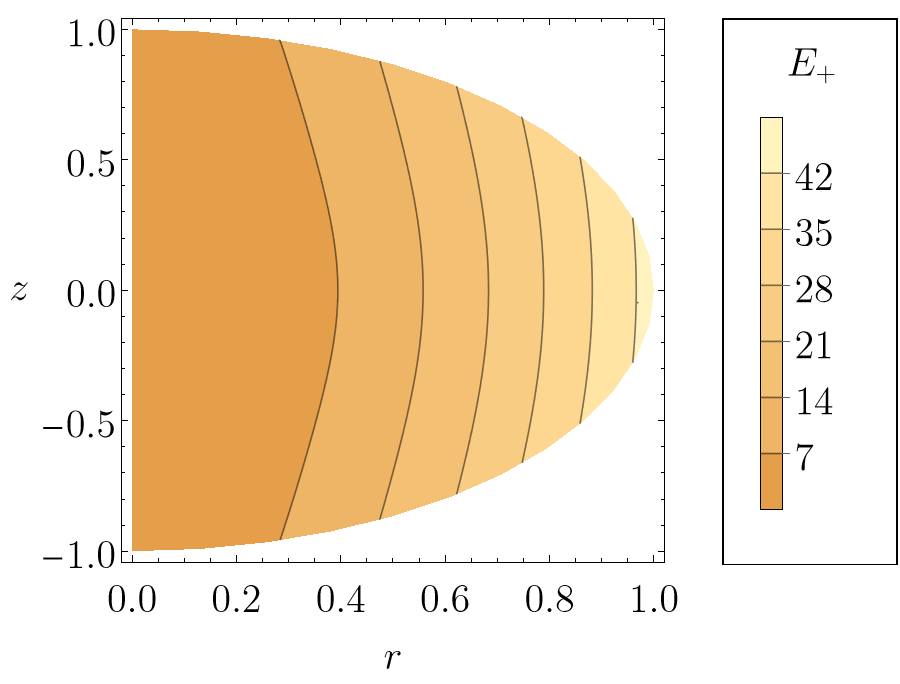}
                    \caption{A contour plot for the eigenvalue $E_+$. This eigenvalue is always positive and increases in magnitude with $z$. Hence, the signature of $g_2$ is determined by $E_-$.}   
                \end{subfigure}
                \hspace{25pt}
                \begin{subfigure}[b]{0.425\textwidth}
                    \includegraphics[scale=0.201]{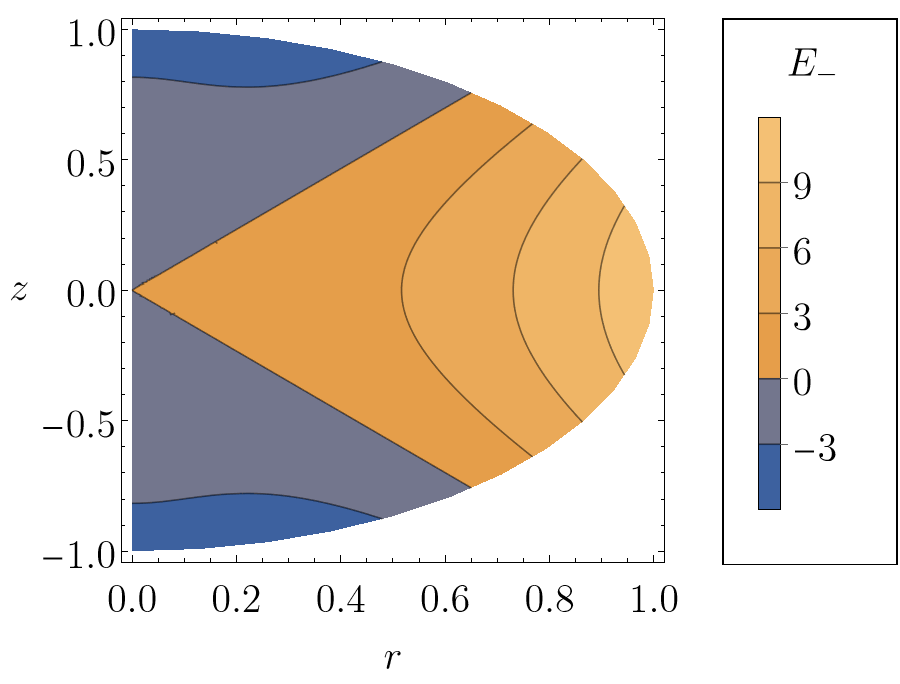}
                    \caption{A contour plot for the eigenvalue $E_-$. This eigenvalue vanishes along $100r^4-71r^2z^2-2z^4=0$, is positive between these lines and negative outside of them.}
                \end{subfigure}
                \caption{Plots of the eigenvalues~\eqref{eq:eigenvaluesHV} of the pull-back metric~\eqref{eq:pullbackMetricHV} for the interior solution of Hill's spherical vortex.}
                \label{fig:eigenvaluesHV}
            \end{center}
        \end{figure}

        Furthermore, the curvature scalar $R_2$ associated with~\eqref{eq:pullbackMetricHV} is
        \begin{equation}\label{eq:curvatureHV}
            R_2\ =\ \frac{28\big(50r^4+z^4\big)}{9\big(100r^4-71r^2z^2-2z^4\big)^2}~.
        \end{equation}
        Both the curvature $R_2$ and eigenvalue $E_+$ are non-negative, with $E_+$ vanishing only at the origin. Additionally, $R_2$ is singular precisely when $E_-$ vanishes, that is, where the metric is singular. See \cref{fig:curvatureR2InteriorHill}.

        \section{Summary and Conclusions}\label{sec:conclusions}

        We have developed a framework for studying the Poisson equation for the pressure, for incompressible flows, by formulating the concept of higher Monge--Amp{\`e}re geometry. Earlier work has been revisited and the definition of Monge--Amp{\`e}re structures extended, using higher symplectic geometry, to facilitate a study of equations in three (or more) independent variables that are not necessarily of explicit Monge--Amp{\`e}re type. 

        In contrast to the earlier work, the focal point of our investigations has shifted from the perspective of symplectic geometry to that of Riemannian and Kleinian geometries, as defined by a metric and its pull-back to higher Lagrangian submanifolds. This change of viewpoint has illuminated some seemingly important connections between fluid flows dominated by either vorticity or strain and the geometry of Lagrangian submanifolds. From the examples studied thus far, regions of the flow one might label `a vortex' are characterised by Monge--Amp{\`e}re structures with Riemannian metrics, whereas those regions in which strain dominates, are characterised by Kleinian metrics. Where vorticity or strain accumulate, the higher Lagrangian submanifold develops curvature, and the singular behaviour of the pull-back metric and the scalar curvature arising from~\eqref{eq:PEP} appears to delineate regions of the flow with distinct topological characteristics. Furthermore, where vorticity dominates over the strain, the metric is typically Riemannian with positive curvature scalar.

        We have focused on the Ricci scalar curvature, simply because, as an invariant, it is a natural starting point for attempting to identify the salient connections between the characteristics of the fluid flow and the geometry of the Lagrangian submanifolds. However, the Ricci curvature itself may reveal further insights. 

        We have also noted that there is typically a one-parameter family of metrics, with time $t$ acting as the parameter. The evolution of the metrics, as the parameter $t$ increases, will depend on whether the evolution of the dynamics is governed by the incompressible Euler or by the incompressible Navier--Stokes equations. It is through such a time evolution that the geometry we have introduced will reflect the differences in the solutions to these two sets of equations. It is therefore intriguing to speculate as to whether Monge--Amp{\`e}re structures may reveal information on the existence and/or regularity of solutions to the Euler and Navier--Stokes equations, via a notional geometric flow. A connection between Monge--Amp{\`e}re geometry and optimal transport~\cite{Kim:2010aa}, in terms of a metric whose properties relate to regularity of solutions to Monge--Amp{\`e}re equations has been made in~\cite{Donofrio:2023aa}.

        Possible future directions include a systematic study of the fully three-dimensional solutions to the Navier--Stokes equations and the associated metrics, with particular emphasis on topological properties of the higher Lagrangian submanifolds that might be characterised by the metrics and their curvature. 

        A recent companion study to this paper,~\cite{Donofrio:2022aa}, has also explored the connections between Monge--Amp{\`e}re structures and the geometry of Lagrangian submanifolds arising in a more conventional application of Monge--Amp{\`e}re geometry, in which a fully non-linear Monge--Amp{\`e}re equation lies at the heart of a model used in the study of geophysical flows. Emphasis in that paper is placed on the projection between $M$ and $L$, and singularities thereof. This is an important issue to follow up in the context of the results presented here.

        \appendix
        \addappheadtotoc
        \appendixpage 

        \appendices

        \section{Lagrangian submanifolds}\label{app:LagrangianSubmanifolds}
        
        Let $M$ be an $m$-dimensional manifold and $\pi:T^*M\rightarrow M$ its cotangent bundle. We call a submanifold $\iota: L\hookrightarrow T^*M$ \uline{locally a section} if and only if, for each $y\in L$, there exists a neighbourhood $V_y\subseteq L$ of $y$, an open and contractible set $U_y\subseteq M$, and a function $\psi_y\in\scC^\infty(U_y)$ such that $\iota(V_y)=\rmd\psi_y(U_y)$. Next, let $U_M\subseteq M$ be open and contractible and let $(x^i,q_i)$ with $i,j,\ldots=1,\ldots,m$ be local coordinates on $\pi^{-1}(U_M)\subseteq T^*M$ with $x^i$ local coordinates on $M$ and $q_i$ local fibre coordinates, respectively. Consider a Lagrangian submanifold $\iota:L\hookrightarrow T^*M$ with respect to the standard symplectic form $\omega\coloneqq\rmd q_i\wedge\rmd x^i$.

        \begin{proposition} 
            The Lagrangian submanifold $L$ is locally a section $\rmd\psi:U_M\rightarrow T^*M$ for some $\psi\in\scC^\infty(U_M)$ if and only if $\pi|_L\coloneqq\pi\circ\iota:L\rightarrow M$ is a local diffeomorphism.
        \end{proposition}

        \begin{proof}
            Suppose that $L$ is locally a section and consider a point $y\in L$. Evidently, $V_y$ and $\iota(V_y)$ are diffeomorphic. Furthermore, as the restriction $\pi|_{\rmd\psi_y(U_y)}$ is the inverse of $\rmd\psi_y: U_y\rightarrow \rmd\psi_y(U_y)\subseteq T^*M$, it follows that $U_y$ and $\rmd\psi_y(U_y)$ are diffeomorphic. As $\rmd\psi_y(U_y)= \iota(V_y)$, it then follows that $U_y$ and $V_y$ are diffeomorphic, with diffeomorphism given by $\pi|_L|_{V_y}:V_y\rightarrow U_y$ and $\iota^{-1}\circ\,\rmd\psi_y:U_y\rightarrow V_y$. Since this holds for all $y\in L$, it follows that $\pi|_L:L\rightarrow M$ is a local diffeomorphism.  

            Conversely, suppose that $\pi|_L:L\rightarrow M$ is a local diffeomorphism and let $y\in L$ be arbitrary. By the local diffeomorphism property of $\pi|_L:L\rightarrow M$, there exists a neighbourhood $V_y\subseteq L$ of $y$ such that $\pi|_L(V_y)\eqqcolon U_y$ is open and contractible in $M$ and $\pi|_L|_{V_y}$ is a diffeomorphism onto its image. Let $y^i$ be local coordinates on $V_y$. Then, $V_y\ni y^i\xmapsto{\pi|_L}x^i(y)\in U_y$. Again by the local diffeomorphism properties of $\pi|_L$, we have locally the invertibility of the Jacobian $\parder[x^i]{y^j}$ and the inverse relation $y^i=y^i(x)$. Hence, the embedding $\iota:L\hookrightarrow T^*M$ becomes 
            \begin{equation}
            	\iota\,:\,y^i\ \mapsto\ (x^i(y),q_i(y))\ =\ (x^i,q_i(y(x)))\ \eqqcolon\ (x^i,p_i(x))
            \end{equation} 
            in local coordinates. Furthermore, since $L$ is Lagrangian with respect to $\omega=\rmd q_i\wedge\rmd x^i$, we find
            \begin{equation}
                \parder[x^i]{y^j}\parder[q_i]{y^k}\ =\ \parder[x^i]{y^k}\parder[q_i]{y^j}
            \end{equation}
            upon computing $\iota^*\omega=0$.\footnote{Note that $\iota_*\big(\parder{y^i}\big)=\parder[x^j]{y^i}\parder{x^j}+\parder[q_j]{y^i}\parder{q_j}$.} Hence,
            \begin{equation}
                \begin{aligned}
                    \parder[x^i]{y^l}\left(\parder[p_j]{x^i}-\parder[p_i]{x^j}\right)\ &=\ \parder[x^i]{y^l}\left(\parder[y^k]{x^i}\parder[q_j]{y^k}-\parder[y^k]{x^j}\parder[q_i]{y^k}\right)
                    \\
                    &=\ \parder[q_j]{y^l}-\parder[y^k]{x^j}\parder[x^i]{y^l}\parder[q_i]{y^k}
                    \\
                    &=\ \parder[q_j]{y^l}-\parder[y^k]{x^j}\parder[x^i]{y^k}\parder[q_i]{y^l}
                    \\
                    &=\ \parder[q_j]{y^l}-\parder[q_j]{y^l}
                    \\
                    &=\ 0~.
                \end{aligned}
            \end{equation}
            Therefore,
            \begin{equation}
                \parder[p_j]{x^i}-\parder[p_i]{x^j}\ =\ 0~,
            \end{equation}       
            that is, the one-form $\eta\coloneqq p_i\rmd x^i$ is closed. Consequently, by the Poincar\'e lemma, there is a function $\psi_y\in\scC^\infty(U_y)$ so that $\eta=\rmd\psi_y$ (and therefore $p_i = \partial_i\psi_y$). It follows that $\rmd\psi_y(U_y) = \iota(V_y)$. As this holds for any $y\in L$, it follows that $L$ is locally a section.        
        \end{proof}

        When $\pi|_L:L\rightarrow M$ is a local diffeomorphism, we may choose coordinates on $L$ so that $\pi|_L$ becomes locally the identity in those coordinates. Put differently, we may take the $x^i$ as local coordinates on $L$ in that case.

        \section{Non-degenerate Monge--Amp{\`e}re structures}\label{app:MongeAmpereStructures}

        Let $(M,\omega)$ be a $2m$-dimensional almost symplectic manifold. Following~\cite{Lychagin:1979aa}, a differential $p$-form is called \uline{$\omega$-effective} if and only if $\omega^{-1}\intprod\alpha=0$. Whenever $p=m$ this is equivalent to requiring $\alpha\wedge\omega=0$. Then, we have the \uline{Hodge--Lepage--Lychagin theorem}~\cite{Lychagin:1979aa} (see also the text book~\cite{Kushner:2007aa} for a comprehensive treatment):

        \begin{theorem}\label{thm:HodgeLepageLychagin}
            Let $(M,\omega)$ be an almost symplectic manifold.  Then, any differential $p$-form $\alpha\in\Omega^p(M)$ has a unique decomposition 
            \begin{equation}
            	\alpha\ =\ \alpha_0+\alpha_1\wedge\omega+\alpha_2\wedge\omega\wedge\omega+\cdots
            \end{equation}
            into $\omega$-effective differential $(p-2k)$-forms $\alpha_k\in\Omega^{p-2k}(M)$. Furthermore, if two $\omega$-effective $p$-forms vanish on the same $p$-dimensional isotropic submanifolds, they must be proportional.
        \end{theorem}

        Let now $M$ be four-dimensional and $(\omega,\alpha)$ a Monge--Amp{\`e}re structure on $M$, that is, $\alpha\in\Omega^2(M)$ with $\alpha\wedge\omega=0$ and suppose that the \uline{Pfaffian} $\sfPf(\alpha)\in\scC^\infty(M)$, defined by $\alpha\wedge\alpha=\sfPf(\alpha)\omega\wedge\omega$, is non-zero. We then set~\cite{Lychagin:1993aa}
        \begin{equation}\label{eq:almostStructure4D}
            \frac{\alpha}{\sqrt{|\sfPf(\alpha)|}}\ \eqqcolon\ J_\alpha\intprod\omega
        \end{equation}
        for $J_\alpha$ an endomorphism of the tangent bundle. This yields the identity
        \begin{equation}
            J_\alpha\intprod(\alpha\wedge\omega)\ =\ \sqrt{|\sfPf(\alpha)|}\left(J^2_\alpha\intprod\omega+\frac{\sfPf(\alpha)}{|\sfPf(\alpha)|}\,\omega\right)\!\wedge\omega~.
        \end{equation}
        Upon combining this identity with the $\omega$-effectiveness of $\alpha$ and the non-degeneracy of $\omega$, we immediately see that $J_\alpha$ is an almost complex (respectively, para-complex) structure when $\sfPf(\alpha)>0$ (respectively, $\sfPf(\alpha)<0$). The differential forms $\omega$ and $J_\alpha\intprod\omega$ define the non-degenerate differential $(2,0)$- and $(0,2)$-forms with respect to $J_\alpha$. Then, we have the following result:

        \begin{proposition}\label{prop:3rd3Form2d}
            For $J_\alpha$ as defined in~\eqref{eq:almostStructure4D} there exists a differential $(1,1)$-form $K$ on $M$ such that $ K\wedge K\neq0$, $K\wedge\omega=0$, and $ K\wedge(J_\alpha\intprod\omega)=0$.
        \end{proposition}

        \begin{proof}
            Note that $\omega$ and $J_\alpha\intprod\omega$ are linearly independent. Next, let $\rho\in\Omega^2(M)$ be such that $\{\omega, J_\alpha\intprod\omega,\rho\}$ is linearly independent. By \cref{thm:HodgeLepageLychagin}, we have a unique decomposition $\rho=\rho_0+\lambda_0\omega$ with $\rho_0\wedge\omega=0$ and $\lambda_0\in\scC^\infty(M)$. Since $(J_\alpha\intprod\omega)\wedge(J_\alpha\intprod\omega)\neq0$, we may again apply \cref{thm:HodgeLepageLychagin} to obtain the unique decomposition $\rho_0=\rho_1+\lambda_1(J_\alpha\intprod\omega)$ with $\lambda_1\in\scC^\infty(M)$ such that $\rho_1\wedge(J_\alpha\intprod\omega)=0$. Since $(J_\alpha\intprod\omega)\wedge\omega=0$, we also have $\rho_1\wedge\omega=0$. Hence, $\{\omega, J_\alpha\intprod\omega,\rho_1\}$ is linearly independent, and we must also have that $\rho_1\wedge\rho_1\neq0$ since the exterior product yields a non-degenerate metric  on $\bigwedge^2T^*M$. In summary, we have thus obtained a $K\coloneqq\rho_1$ such that  $K\wedge K\neq0$, $K \wedge\omega=0$, and $K\wedge(J_\alpha\intprod\omega)=0$. 

            Finally, since $\omega$ and $J_\alpha\intprod\omega$ combine to give non-degenerate differential $(2,0)$- and $(0,2)$-form, $\Omega^{(2,0)}$ and $\Omega^{(0,2)}$, and since $K\wedge\omega=0$ and $ K\wedge(J_\alpha\intprod\omega)=0$, we conclude that $ K\wedge\Omega^{(2,0)}=0$ and $ K \wedge\Omega^{(0,2)}=0$. Since $\Omega^{(2,0)}\wedge\Omega^{(0,2)}\neq0$, $K$ must be of type $(1,1)$ with respect to $J_\alpha$.
        \end{proof}

        \section{Connections and curvatures}

		\subsection{Pull-back metric in two dimensions}\label{app:pullbackCurvature}

		In what follows, we shall provide some more details on the computation of the Levi-Civita connection and curvature scalar associated with the metric~\eqref{eq:fluidMetric2dPullBack} from~\cref{sec:geometry2dFlows}. Firstly, recall that using~\eqref{eq:vorticity2d}, the metric~\eqref{eq:fluidMetric2dPullBack} can be written in the form
		\begin{equation}
			g_{ij}\ =\ \zeta\tilde g_{ij}
			\ewith
			\tilde g_{ij}\ =\ \psi_{ij}~,
		\end{equation}
		where the indices on $\psi\in\scC^\infty(M)$ are interpreted via~\eqref{eq:psiIndexNotation}. As $g$ is, up to a sign, a conformal scaling of the Hessian metric with respect to $\psi$ when $\zeta\neq0$, we wish to exploit this to write the connection and curvature scalar of $g$ in terms of those for $\tilde g$. 
		
		\paragraph{Connection.}
		We begin by observing that 
		\begin{equation}\label{eq:threePsi}
            \bg{\nabla}_i\psi_{jk}\ =\ \psi_{ijk}+\tfrac13\big([\bg{\nabla}_i,\bg{\nabla}_j]\psi_k+[\bg{\nabla}_i,\bg{\nabla}_k]\psi_j\big)\ =\ \psi_{ijk}-\tfrac13\big(\bg{R}_{ijk}{}^l+\bg{R}_{ikj}{}^l\big)\psi_l
		\end{equation}
        and so,
		\begin{equation}\label{eq:threePsiSum}
            \bg{\nabla}_i\psi_{jk}+\bg{\nabla}_j\psi_{ik}-\bg{\nabla}_k\psi_{ij}\ =\ \psi_{ijk}+\tfrac43\bg R_{k(ij)}{}^l\psi_l~.
		\end{equation}
		Consequently, the Christoffel symbols for $\tilde g$ are given by
        \begin{subequations}\label{eq:LCPullback2dApp}
            \begin{equation}
                \begin{aligned}
                    \tilde\Gamma_{ij}{}^k\ &=\ \tfrac12\tilde g^{kl}(\partial_i\tilde g_{jl}+\partial_j\tilde g_{il}-\partial_l\tilde g_{ij})
                    \\
                    &=\ \bg{\Gamma}_{ij}{}^k+\tfrac12\tilde g^{kl}\big(\bg{\nabla}_i\psi_{jl}+\bg{\nabla}_j\psi_{il}-\bg{\nabla}_l\psi_{ij}\big)
                    \\
                    &=\ \bg{\Gamma}_{ij}{}^k+\tfrac12\Upsilon_{ijl}\tilde g^{lk}~,
                \end{aligned}
            \end{equation}
    		where we have used~\eqref{eq:threePsiSum} and introduced the notation
            \begin{equation}\label{eq:UpsilonApp}
                \Upsilon_{ijk}\ \coloneqq\ \psi_{ijk}+\tfrac43\bg{R}_{k(ij)}{}^l\psi_l~.
            \end{equation}
        \end{subequations}
        This thus verifies~\eqref{eq:hessianChristoffel}. The Christoffel symbols~\eqref{eq:pullback2DChristoffel} then follow from the usual argument for conformal rescalings (see e.g.~\cite{Besse1987aa}).
		
		\paragraph{Curvature.}
		Let us now compute the curvature scalar for~\eqref{eq:fluidMetric2dPullBack}. Firstly, we note that
        \begin{equation}\label{eq:curvaturePullBack2dApp}
            \begin{aligned}
                \tilde R_{ijk}{}^l\ &=\ \partial_i\tilde\Gamma_{jk}{}^l-\partial_j\tilde\Gamma_{ik}{}^l-\tilde\Gamma_{ik}{}^m\tilde\Gamma_{jm}{}^l+\tilde\Gamma_{jk}{}^m\tilde\Gamma_{im}{}^l
                \\
                &=\ \bg{R}_{ijk}{}^l+\tfrac12\big(\bg{\nabla}_i\Upsilon_{jk}{}^l-\bg{\nabla}_j\Upsilon_{ik}{}^l-\tfrac12\Upsilon_{ik}{}^m\Upsilon_{jm}{}^l+\tfrac12\Upsilon_{jk}{}^m\Upsilon_{im}{}^l\big)\,,
            \end{aligned}
        \end{equation}
        where we have used~\eqref{eq:LCPullback2dApp} and set $\Upsilon_{ij}{}^k\coloneqq\Upsilon_{ijl}\tilde g^{lk}$. Next, it is not too difficult to see that
        \begin{equation}
            \bg{\nabla}_i\tilde g^{jk}\ =\ -\tilde g^{jl}\tilde g^{km}\big(\psi_{ilm}-\tfrac23\bg{R}_{i(lm)}{}^n\psi_n\big)\,.
        \end{equation}
        and
        \begin{equation}
            \bg{\nabla}_i\psi_{jkl}\ =\ \psi_{ijkl}-2\bg{R}_{i(jk}{}^m\psi_{l)m}-\tfrac12\psi_m\bg{\nabla}_{(j}\bg{R}_{|i|kl)}{}^m~.
        \end{equation}
        Using these two relations, we find that $\Upsilon_{jk}{}^l=\tilde g^{lm}\Upsilon_{jkm}$,
        \begin{equation}
            \begin{aligned}
                \bg{\nabla}_i\Upsilon_{jk}{}^l\ &=\ -\tilde g^{lr}\big(\psi_{irs}-\tfrac23\bg{R}_{i(rs)}{}^n\psi_n\big)\Upsilon_{jk}{}^s
                \\
                &\kern1cm+\tfrac43\tilde g^{lm}\big(\bg{R}_{m(jk)}{}^n\psi_{in}+\psi_n\bg{\nabla}_i\bg{R}_{m(jk)}{}^n\big)
                \\
                &\kern1cm+\tilde g^{lm}\big(\psi_{ijkm}-\tfrac12\psi_n\bg{\nabla}_{(j}\bg{R}_{|i|km)}{}^n-2\bg{R}_{i(jk}{}^n\psi_{m)n}\big)\,.
            \end{aligned}
        \end{equation}
        Upon substituting this expression and~\eqref{eq:UpsilonApp} into~\eqref{eq:curvaturePullBack2dApp}, the curvature scalar~\eqref{eq:HessianScalarCurvature} then follows directly from the traces $\tilde R=\tilde g^{ij}\tilde R_{kij}{}^k$. Finally, the curvature scalar~\eqref{eq:scalarCurvature2D} then follows from the usual argument for conformal rescalings (see e.g.~\cite{Besse1987aa}).
				
		\subsection{Phase space curvature}\label{app:phaseMetricCurvature}

        We now compute the curvature of the metric~\eqref{eq:fluidMetric2d}. Before we do so, however, let us recap the vielbein formalism as it is more efficient than working in a coordinate basis.  

        \paragraph{Vielbein formalism.}
        Let $(M,g)$ be an $m$-dimensional (semi-)Riemannian manifold coordinatised by $x^i$ with $i,j,\ldots=1,\ldots,m$. Then 
        \begin{equation}
	        g\ =\ \tfrac12g_{ij}\rmd x^i\odot\rmd x^j~.
	    \end{equation}        
        We denote the vielbeins by $E_a=E_a{}^i\partial_i\in\frX(M)$ for $a,b,\ldots=1,\ldots,m$, with $(E_a{}^i)\in\scC^\infty(M,\sfGL(m))$. Dually, we have $e^a=\rmd x^ie_i{}^a\in\Omega^1(M)$, with $(e_i{}^a)\in\scC^\infty(M,\sfGL(m))$, satisfying $E_a\intprod e^b=\delta_a{}^b$. Consequently, we have 
        \begin{equation}
        	E_a{}^ie_i{}^b\ =\ \delta_a{}^b
            \eand 
            e_i{}^aE_a{}^j\ =\ \delta_i{}^j~,
        \end{equation}
        and therefore 
        \begin{equation}
        	g\ =\ \tfrac12e^b\odot e^a\eta_{ab}~,
        \end{equation}
        with $\eta_{ab}=\diag(-1,\ldots,-1,1,\ldots,1)$.

        The \uline{structure functions} $C_{ab}{}^c\in\scC^\infty(M)$ are given by
        \begin{subequations}\label{eq:structureFunctions}
            \begin{equation}
                [E_a,E_b]\ =\ C_{ab}{}^cE_c~,
            \end{equation}
            or, dually,
            \begin{equation}
                \rmd e^a\ =\ \tfrac12e^c\wedge e^bC_{bc}{}^a~.
            \end{equation}
        \end{subequations}
        The \uline{torsion} and \uline{curvature two-forms},
        \begin{subequations}
            \begin{equation}
                T^a\ =\ \tfrac12e^c\wedge e^bT_{bc}{}^a
                \eand
                R_a{}^b\ =\ \tfrac12e^d\wedge e^cR_{cda}{}^b~,
            \end{equation}
            are defined by the \uline{Cartan structure equations}
            \begin{equation}\label{eq:cartanStructureEquations}
                \rmd e^a-e^b\wedge\omega_b{}^a\ \eqqcolon\ -T^a
                \eand
                \rmd\omega_a{}^b-\omega_a{}^c\wedge\omega_c{}^b\ \eqqcolon\ -R_a{}^b~,
            \end{equation}
            where $\omega_a{}^b=e^c\omega_{ca}{}^b$ is the \uline{connection one-form}. The associated Ricci tensor and the curvature scalar are then given by
            \begin{equation}\label{eq:RicciScalarVielbein}
                R_{ab}\ \coloneqq\ R_{cab}{}^c
                \eand
                R\ \coloneqq\ \eta^{ba}R_{ab}~.
            \end{equation}
        \end{subequations}
        Furthermore, \uline{metric compatibility} amounts to requiring
        \begin{equation}\label{eq:metricCompatibility}
            \omega_{ab}\ =\ -\omega_{ba}
            \ewith
            \omega_{ab}\ \coloneqq\ \omega_a{}^c\eta_{cb}~.
        \end{equation}

        Imposing~\eqref{eq:metricCompatibility} and the torsion freeness constraint $T^a=0$ yields the Levi-Civita connection and a short calculation shows that this connection is given by
        \begin{equation}\label{eq:LCConnectionVielbein}
            \omega_{ab}{}^c\ =\ \tfrac12(C^c{}_{ab}+C^c{}_{ba}+C_{ab}{}^c)
        \end{equation}
        with indices raised and lowered using $\eta_{ab}$. In this case, the curvature scalar~\eqref{eq:RicciScalarVielbein} is
        \begin{equation}\label{eq:curvatureScalarVielbein}
            R\ =\ 2E_aC^a{}_b{}^b-C_{ab}{}^bC^a{}_c{}^c-\tfrac12C_{abc}C^{acb}-\tfrac14C_{abc}C^{abc}~.
        \end{equation}
		
		\paragraph{Connection.}
        Let now $(M,\bg{g})$ be a Riemannian manifold, and consider the the metric~\eqref{eq:fluidMetric3d} on $T^*M$ now assumed to be in $2m$ dimensions. Furthermore, let
        \begin{equation}
            \bg{E}_a\ \coloneqq\ \bg{E}_a{}^i\parder{x^i}
            \eand
            \bg{e}^a\ \coloneqq\ \rmd x^i\bg{e}_i{}^a 
        \end{equation}
        be the vielbeins and dual vielbeins on $(M,\bg{g})$ with structure functions $\bg{C}_{ab}{}^c$, and set
        \begin{equation}
            \begin{gathered}
                (\hat e^A)\ =\ (\hat e^a,\hat e_a)\ \coloneqq\ \Big(\sqrt{|\hat f|}\,\rmd x^i\bg{e}_i{}^a,\bg{E}_a{}^i\,\bg{\nabla}q_i\Big)\,,
                \\
                (\hat\eta_{AB})\ =\ 
                \begin{pmatrix}
                    \hat\eta_{ab} & \hat\eta_a{}^b
                    \\
                    \hat\eta^a{}_b& \hat\eta^{ab}
                \end{pmatrix}
                \ \coloneqq\ 
                \begin{pmatrix}
                    \textrm{sgn}(\hat f)\unit_m & 0
                    \\
                    0 & \unit_m
                \end{pmatrix}
            \end{gathered}
        \end{equation}
        for multi-indices $A,B,\ldots$. Then, the metric~\eqref{eq:fluidMetric3d} becomes
        \begin{equation}
            \hat g\ =\ \tfrac12\hat e^B\odot\hat e^A\hat\eta_{AB}~.
        \end{equation}
        Note that $\bg{e}_i{}^a$ and $\bg{E}_a{}^i$ only depend on the base manifold coordinates $x^i$ and not on the fibre coordinates $q_i$.

        Next, dually, we have $\hat E_A\intprod\hat e^B=\delta_A{}^B$ with $(\hat E_A)=(\hat E_a,\hat E^a)$ and
        \begin{equation}\label{eq:vielbeinsPhaseSpaceMetric}
            \hat E_a\ \coloneqq\ \tfrac{1}{\sqrt{|\hat f|}}\bg{E}_a{}^i\left(\parder{x^i}+\bg{\Gamma}_{ij}{}^kq_k\parder{q_j}\right)
            \eand 
            \hat E^a\ \coloneqq\ \bg{e}_i{}^a\parder{q_i}~.
        \end{equation}
        After a straightforward calculation, we obtain for $[\hat E_A,\hat E_B]=\hat C_{AB}{}^C\hat E_C$ the relations
        \begin{subequations}\label{eq:vielbeinCommutation}
            \begin{equation}
                \begin{aligned}
                    [\hat E_a,\hat E_b]\ &=\ \tfrac{1}{\sqrt{|\hat f|}} \bg{C}_{ab}{}^c\hat E_c-\hat E_{[a}\log(|\hat f|)\hat E_{b]}+\tfrac{1}{|\hat f|}\bg{R}_{abc}{}^dq_d\hat E^c~,
                    \\
                    [\hat E_a,\hat E^b]\ &=\ \tfrac{1}{2}\hat E^b\log({|\hat f|)}\hat E_a-\tfrac{1}{\sqrt{|\hat f|}}\bg{\omega}_{ac}{}^b\hat E^c~,
                    \\
                    [\hat E^a,\hat E^b]\ &=\ 0~,
                \end{aligned}
            \end{equation}
            where we have set $q_a\coloneqq\bg{E}_a{}^i q_i$ and used the identities
            \begin{equation}
                \bg{\omega}_{ab}{}^c\ =\ \bg{E}_a{}^i\bg{E}_b{}^j\left(\bg{\Gamma}_{ij}{}^k\bg{e}_k{}^c-\parder{x^i}\bg{e}_j{}^c\right)
                \eand
                \bg{R}_{abc}{}^d\ =\ \bg{E}_a{}^i\bg{E}_b{}^j\bg{E}_c{}^k\bg{R}_{ijk}{}^l\bg{e}_l{}^d~.
            \end{equation}  
        \end{subequations}
        
        Reading off the structure functions $\hat C_{AB}{}^C$ from these relations and using the formula~\eqref{eq:LCConnectionVielbein}, the Levi-Civita connection $\hat\omega_{AB}{}^C$ for the metric~\eqref{eq:fluidMetric3d} can be written in terms of the Levi-Civita connection $\bg{\omega}_{ab}{}^c$ for the background metric $\bg{g}$ as
        \begin{equation}\label{eq:LCCotangent}
            \hat\omega_{AB}{}^C\ =\ \tfrac12(\hat C^C{}_{AB}+\hat C^C{}_{BA}+\hat C_{AB}{}^C)~.
        \end{equation}
        
        \paragraph{Curvature.}
        Upon combining~\eqref{eq:vielbeinCommutation} and~\eqref{eq:LCCotangent} with~\eqref{eq:curvatureScalarVielbein}, the curvature scalar of the metric~\eqref{eq:fluidMetric3d} is given by 
        \begin{subequations}
            \begin{equation}
                \begin{aligned}
                    \hat R\ &=\ \tfrac{1}{\hat f}\bg{R}-\tfrac{1}{4\hat f^2}\bg{R}_{abc}{}^d\bg{R}^{abce}q_dq_e-(m-1)\hat\lap_{\rm B}\log(|\hat f|)-\delta_{ab}\hat E^a\hat E^b\log(|\hat f|)
                    \\
                    &\kern1cm+\tfrac{{\rm sgn}(\hat f)}{4}(m-1)(m-2)\delta^{ab}\hat E_a\log(|\hat f|)\hat E_b\log(|\hat f|)
                    \\
                    &\kern1.5cm+\tfrac14m(m-3)\delta_{ab}\hat E^a\log(|\hat f|)\hat E^b\log(|\hat f|)~,
                \end{aligned}
            \end{equation}
            where $\hat\lap_{\rm B}$ is the Beltrami Laplacian for $\hat g$. Here, $\bg{R}_{abc}{}^d$ is the Riemann curvature tensor for the background metric $\bg{g}$ and $\bg{R}$ the associated curvature scalar. In our coordinate basis, this becomes
            \begin{equation}\label{eq:curvatureScalarFluidMetricApp}
                \begin{aligned}
                    \hat R\ &=\ \frac{1}{\hat f}\bg{R}-\frac{1}{4\hat f^2}\bg{R}_{ijk}{}^l\bg{R}^{ijkm}q_kq_m-(m-1)\hat\lap_{\rm B}\log(|\hat f|)-\bg{g}_{ij}\parder[^2]{q_i\partial q_j}\log(|\hat f|)
                    \\
                    &\kern.6cm+\frac{1}{4\hat f}(m-1)(m-2)\bg{g}^{ij}\left(\parder{x^i}+\bg{\Gamma}_{ik}{}^lq_l\parder{q_k}\right)\log(|\hat f|)\left(\parder{x^j}+\bg{\Gamma}_{jm}{}^nq_n\parder{q_m}\right)\log(|\hat f|)
                    \\
                    &\kern1.1cm+\frac14m(m-3)\bg{g}_{ij}\parder{q_i}\log(|\hat f|)\parder{q_j}\log(|\hat f|)~,
                \end{aligned}
            \end{equation}
        \end{subequations} 
        where we have used~\eqref{eq:vielbeinsPhaseSpaceMetric}. This verifies~\eqref{eq:curvatureScalarFluidMetric}.

        Finally, we note that in the case of the flat background metric $\bg{g}_{ij}=\delta_{ij}$, we have $\hat f=f=\frac12\lap p$ with $\lap$ the standard Laplacian on $\IR^m$ and so, the formula~\eqref{eq:curvatureScalarFluidMetricApp} simplifies to
        \begin{equation}
            \hat R\ =\ \frac{m-1}{4f^3}[(6-m)\partial_if\partial^if-4f\lap f]~.
        \end{equation}

	\end{body}

\end{document}